\documentclass[pra,twocolumn,nofootinbib,amsmath,amssymb,aps,longbibliography]{revtex4-2}
\usepackage{graphicx}
\usepackage{dcolumn}
\usepackage{bm}
\usepackage{physics}
\usepackage{amsmath}
\usepackage{dsfont}
\usepackage{xcolor}
\usepackage{lmodern}
\usepackage{appendix}
\usepackage{amsthm}
\usepackage{ulem}

\usepackage{hyperref}
\hypersetup{colorlinks=true,linkcolor=blue,citecolor=blue,urlcolor=blue}

\usepackage{titlesec}
\titleformat{\paragraph}
{\normalfont\normalsize\bfseries}{\theparagraph}{1em}{}
\titlespacing*{\paragraph}
{0pt}{3.25ex plus 1ex minus .2ex}{1.5ex plus .2ex}

\usepackage{hyperref}
\usepackage[bottom,splitrule]{footmisc}
\usepackage{scrextend}
 \deffootnote[1em]{1.5em}{1em}{\textsuperscript{\thefootnotemark}}

\newtheorem{corollary}{Corollary}

\newtheorem{definition}{Definition}

\newtheorem{remark}{Remark}

\definecolor{aliceblue}{rgb}{0.94, 0.97, 1.0}

\newcommand{\add}[1]{ \textcolor{blue}{#1}}

\def\blk{\color{black}}

\usepackage{mwe}
\usepackage[framemethod=tikz]{mdframed}
\usetikzlibrary{shadows}
\usepackage{tikz}
    \usetikzlibrary{shadows}
\usepackage{tcolorbox}
    \tcbuselibrary{skins}

\mdfdefinestyle{mystyle}{%
    innerlinewidth = 0.0pt,
    outerlinewidth = 1pt,
    linecolor = blue,
    tikzsetting = {draw = black, line width = 0.5pt},
    roundcorner=10pt,
    linecolor=blue,
    linewidth=1pt,
    topline=true,
    frametitleaboveskip=1.5\dimexpr-\ht\strutbox\relax,
}

\newenvironment{mybox}[1][]{%
\ifstrempty{#1}%
    {\mdfsetup{%
    style = mystyle,
    }}%
    {\mdfsetup{%
    style = mystyle,
    frametitle={%
    \tikz{[baseline=(current bounding box.east),outer sep=0pt]
    \node[draw = white, line width = 2pt, text = blue, anchor=east,rectangle,
    fill=aliceblue, rounded corners, drop shadow]
    {\strut #1};
    }}}}%
    \begin{mdframed}[]\relax%
}{\end{mdframed}}

\def\XXint#1#2#3{{\setbox0=\hbox{$#1{#2#3}{\int}$}
     \vcenter{\hbox{$#2#3$}}\kern-.5\wd0}}

\makeatletter
\renewcommand\@makefntext[1]{\leftskip=2em\hskip-2em\@makefnmark#1}
\makeatother

\usepackage[colorinlistoftodos]{todonotes}

\begin{document}
	
    \title{Intermediate Times Dilemma for Open Quantum System:\\  Filtered Approximation to The Refined Weak Coupling Limit}

	

\author{Marek Winczewski}
\author{Antonio Mandarino}
\author{Gerardo Suarez}
\author{Robert Alicki}
\author{Michał Horodecki}
\affiliation{International Centre for Theory of Quantum Technologies, University of Gdansk, Jana Ba\.zy\'nskiego 1A, Gda\'nsk, 80-309, Poland}

	\date{\today}
	
\begin{abstract}
The famous Davies-GKSL secular Markovian master equation is tremendously successful in approximating the evolution of open quantum systems in terms of just a few parameters. However, the fully-secular Davies-GKSL equation fails to accurately describe time scales short enough, i.e., comparable to the inverse of differences of frequencies present in the system of interest. A complementary approach that works well for short times but is not suitable after this short interval is known as the quasi-secular master equation. Still, both approaches fail to have any faithful dynamics in the intermediate time interval. Simultaneously, descriptions of dynamics that apply to the aforementioned ``grey zone'' often are computationally much more complex than master equations or are mathematically not well-structured. The filtered approximation (FA) to the refined weak coupling limit has the simplistic spirit of the Davies-GKSL equation and allows capturing the dynamics in the intermediate time regime. At the same time, our non-Markovian equation yields completely positive dynamics. We exemplify the performance of the FA equation in the cases of the spin-boson system and qutrit-boson system in which two distant time scales appear. 
\end{abstract}
	
\maketitle


\section{Introduction}


\color{black}

The exact dynamics of systems with a small number of degree of freedom 
can be tackled exactly, whereas in a case in which the dimension of a system becomes large, the exact dynamics become computationally not tractable. Hence, it is inevitable to rely on a statistical approach to derive the evolution of a single main unit and consider the remaining part of the whole system as a noisy environment. In analogy with classical systems showing dissipation, the standard terminology refers to these types of systems as open quantum systems~\cite{AlickiLendi1987, Breuer+2006}. However, in addition, a plethora of novel phenomena arise, viz., dephasing, decoherence, and revival of coherence among them all.

Open quantum systems are ubiquitous in physics, and 
the simplest description of their dynamics is by means of Markovian master equation called Davies-GKLS one \cite{davies1974markovian, gorini1976completely,lindblad1976generators}. 
Recently however, open systems  
  showing a behavior far from  the Markovian dynamics are becoming a blooming topical area~\cite{whitney2008staying}. 
In particular, when the environment is a solid state system, relaxation and 
decoherence can be mediated by other bosonic fields, e.g., phonons~\cite{QDnonM1, QDnonM2} or magnons~\cite{magnons1, magnons2}, 
whose behavior is typically not Markovian. In particular, the interaction with the acoustic phonons 
has been experimentally shown to be the mechanism responsible for the damping Rabi-oscillations 
and Rabi-frequency renormalization in $\text{InGaAs/GaAs}$ quantum dots~\cite{RabiPhonon1, RabiPhonon2}. 
Another intense field of research where the use of the correct master equation is still a matter of 
controversy is the modeling of the coherent excitation energy transfer 
in chemical compounds, such as chromophores in light-harvesting complexes~\cite{lightharv1, lightharv2, collini2010}.

As noted in \cite{Alicki1989} at various stages of evolution, the character of evolution can change from Markovian to non-Markovian and vice-versa. A system was proposed, where at the initial and final stages the Markovian description works (albeit with different master equations), while in the intermediate times, one needs a different description. The one working in the initial stage, proposed in \cite{AlickiLendi1987}, is now called quasi-secular (do not confuse with partial secular of \cite{cattaneo2019local}), while the other one is the Davies equation, i.e. the full-secular.  (See Box 1 for presenting this phenomenon in a more general scenario.) 
To describe dynamics in the whole time, one can use the so-called Bloch-Redfield equation~\cite{bloch1957generalized, redfield1957theory}, 
perhaps the earliest master equation. The latter better accounts for finite coupling with bath (e.g. offering a better steady state). However, it may not preserve positivity. 

The first attempt to provide a completely positive alternative was done in \cite{Alicki1989}.
This equation was later independently discovered and developed in Refs.~\cite{Rivas_2017,Rivas_2019}, under the name of {\it refined weak coupling limit} (this terminology was also used in the context of coarse-grained master equation in~\cite{benatti2009, benatti2010}). The equation was derived using a cumulant expansion, therefore for convenience, we so will term it shortly {\it cumulant equation}.
In principle, it is valid for all times scales and therefore interpolates between the two Markovian evolutions mentioned above, as we show later herein. In a later period, an effort was undertaken to modify 
the Bloch-Redfield  in order to make the dynamics complete positive~\cite{Mayenz-Lidar-coarse-graining,galve2017microscopic,farina2019open,Davidovic2020} (see also  Refs. \cite{hartmann2020accuracy,jeske2015bloch} in this context). 
 
There is a problem with the above equations is that they require full knowledge of spectral density, and involve integration over all frequencies. While this is doable for small systems, it will become clearly less and less feasible for larger systems. This is in manifest contrast with the powerful easiness of Davies equation that requires 
the knowledge of merely few physical parameters. 

In this paper, we provide an approximated version of the cumulant equation, which retains complete positivity as well as non-Markovianity, 
but also exhibits the simplicity of the Davies equation.
The coefficients appearing in the proposed evolution equation are not integrals anymore, and they depend only on the Bohr spectrum of the system. As a starting point, we consider the physical Hamiltonian of the system weakly coupled with a reservoir, for which the effects, such as the Lamb-shift, have already accounted within a proper effective theory~\cite{RenormalizationPaper,Lobejko_Mean-Force,correa2023potential}, as required for the validity of equations in the long-time limit.

We illustrate our equation with a three-level system, the simplest system where the effect of intermediate time considered in \cite{Alicki1989} can be seen. 
We show that cumulant equation interpolates between two Markovian evolutions -- quasi-secular and full secular -- reproducing one of them for short times and the other one for long times. We also show that filtered approximation evolution 
does the job, being a very good approximation of the cumulant one. Importantly, FA equation is computationally much less expensive than the original cumulant equation.
We also show that the dynamics predicted by the cumulant equation and FA equation have a non-Markovian character.

We then prove for general open systems that our dynamical equation interpolates between 
quasi-secular and full-secular equations. 
We do a sanity check comparing both the cumulant equation and FA equation with the exact dynamics given 
by the hierarchical equations of motion (HEOM) theory and the Redfield equation in the case of qubit.


\begin{widetext}
~~
\vspace{-2em}
\begin{mybox}[Box 1 - Intermediate Time Dynamics]
\label{box:1}
\begin{minipage}[t]{1.\textwidth}
\includegraphics[width = 0.31\linewidth]{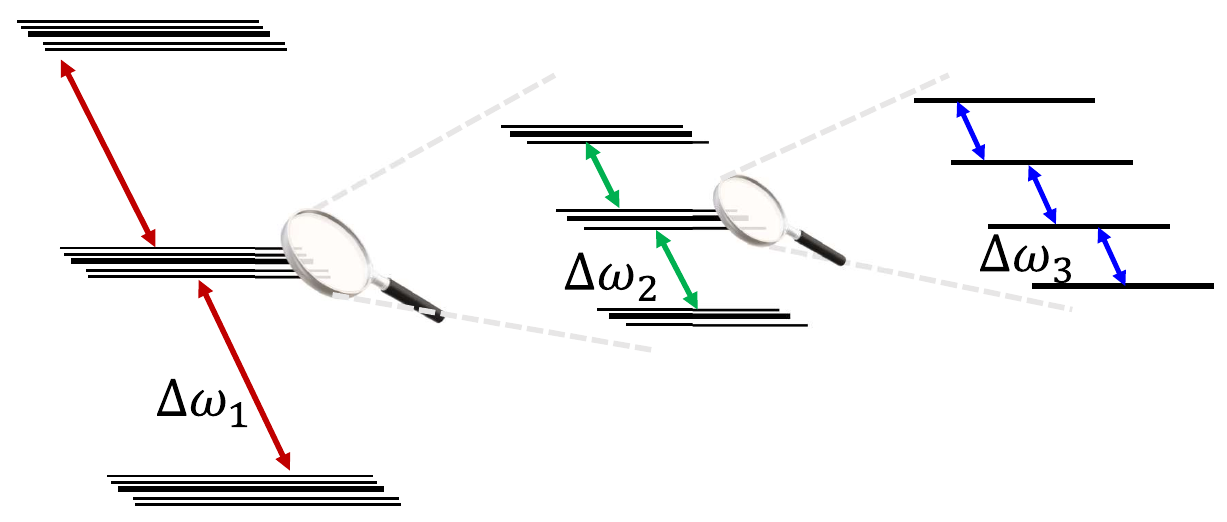}
\includegraphics[width = 0.35\linewidth]{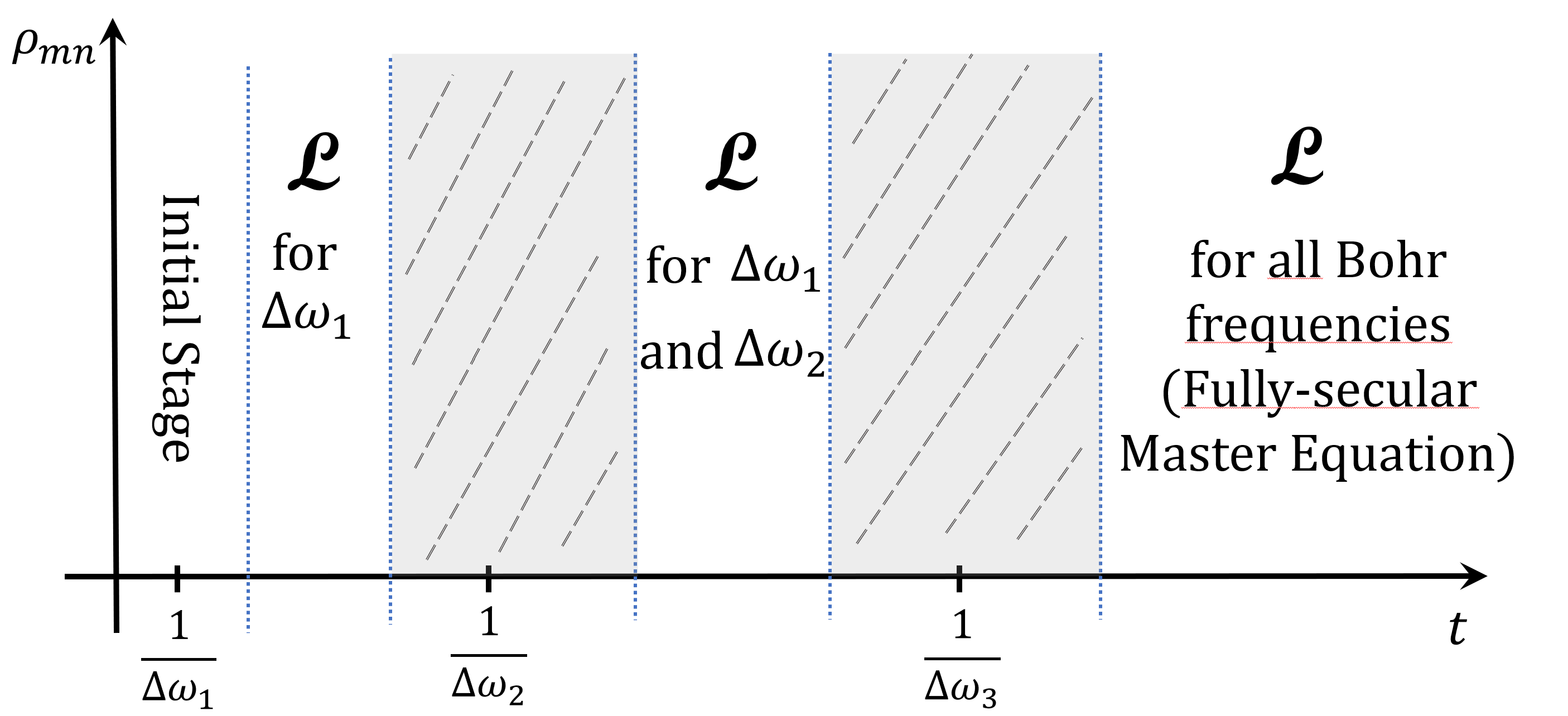}
\includegraphics[width = 0.33\linewidth]{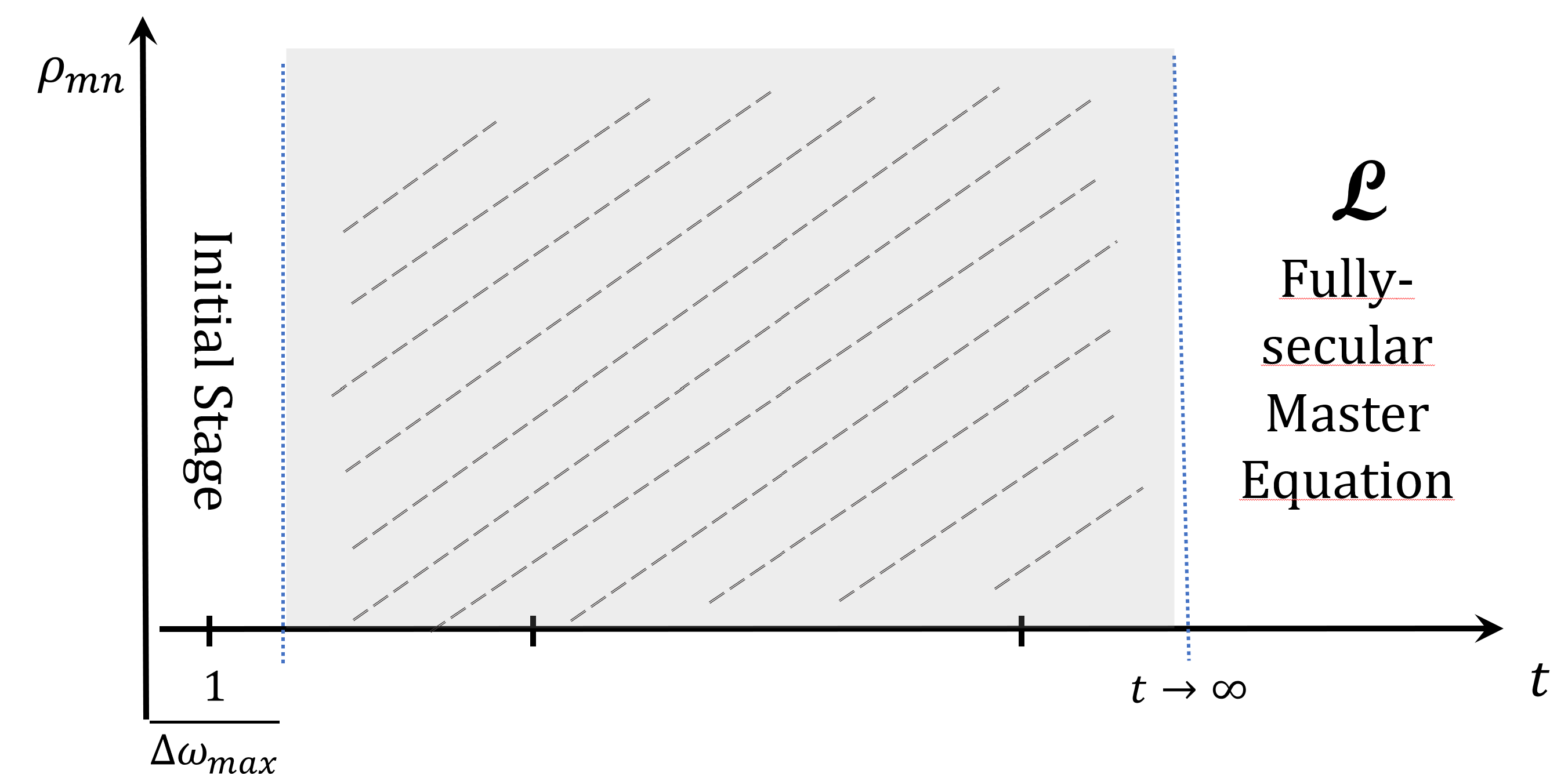}
\end{minipage}%
\hfill
\begin{minipage}{1 \textwidth}
\vspace{1em}
The relevant time scales governing the dynamics of an open quantum system are related to the inverse of the system's Bohr frequencies of the transition induced by the system-bath interaction. The time-energy uncertainty relation written as $ t \geq \frac{1}{\Delta \omega }$ tells how long the system should interact with the bath in order that the dynamics is able to `resolve' all the transition frequencies of the system. 
To easier visualize the concept, let us imagine a system that has a coarse-grained level structure as depicted in the left picture. We denote as $\Delta \omega_1$ the largest transition frequencies between two consecutive `quasi-bands' (subgroups of energy levels). 
A deeper look, namely letting the system and bath interact longer, enables us to resolve firstly the Bohr frequencies of the order of $\Delta \omega_2 <\Delta \omega_1$ between a sublevel of `quasi-bands,' and afterward also the third quasiband $\Delta \omega_3 <\Delta \omega_2 <\Delta \omega_1$. 
We stop our pictorial description after three steps, but in principle, it will last till all the transitions frequency are probed. 
This translates into the time evolution plot of 
a generic element of the system reduced density matrix $\rho_{mn}$ as depicted in the middle graphic. 
At the beginning, for short times, $t< \frac{1}{\Delta \omega_1},$ we observe the `initial stage' of evolution, 
as the shortest characteristic time scale of the system, i.e., $\frac{1}{\Delta \omega_1}$, sets the resolution for the dynamical equations.
When $t \gtrsim \frac{1}{\Delta \omega_1} $ we are in the regime where the dynamics is correctly described by a Davies-type quasi-secular generator 
$\mathcal{L}_{\Delta \omega_1}$ that includes jump operator proportional to frequencies of the order of $\Delta \omega_1$. After a while this description will not be valid anymore and we are in the first ``grey zone'' of the evolution in time. At this moment no valid Davies-like master equation can be derived. At later times, $t \gtrsim \frac{1}{\Delta \omega_2} $ 
a less quasi-secular (or more secular) generator can be written, namely $\mathcal{L}_{\Delta \omega_1, \Delta \omega_2}$. After this period we encounter the second ``grey zone'', and eventually for $t \gtrsim \frac{1}{\Delta \omega_2}$ the dynamics will be described by the (fully-) secular Davies-GKSL generator $\mathcal{\mathcal{L}}^{fs} \equiv \mathcal{L}_{\Delta \omega_1, \Delta \omega_2, \Delta \omega_3}.$

In the right graphic, we show the situation when only one frequency difference can be assumed $\Delta \omega_{max}$ (we are in the presence of a quasicontinuum structure, $\omega_{max}$ is the frequency difference between the ground state and the highest excited level) and the only valid Davies-generator is the secular one in the  $t \rightarrow \infty$ limit.

\vspace{1em}
\end{minipage}%
\end{mybox}
\end{widetext}




\section{Time scales in open system dynamics}\label{sec:TimeScales}
An intriguing question we address for systems weakly interacting with a single thermal bath is the following: 
is it possible to derive a simple dissipative equation that catches some non-Markovian features while retaining during the entire evolution the complete positivity? The answer might be exceptionally important for systems with a complex energetic structure.

The energetic configuration (i.e., the transition frequencies) of a generic system relates to the different time scales identifiable in its dynamical evolution, see e.g.~\cite{Cresser17,cattaneo2019local}. 
If we can group frequency differences into two groups that are well separated from each other (the same analysis applies to more groups) marked by $\Delta \Omega$ and $\Delta \omega$,
for time 
\begin{align}
    \frac{1}{\Delta\Omega} \ll t \ll \frac{1}{\Delta \omega},
\end{align}
the quasi-secular master equation holds \cite{AlickiLendi1987}. Briefly speaking, a quasi-secular master equation consists of a dissipator derived with respect to a modified Hamiltonian. In the modified Hamiltonian, the almost degenerated energy levels, i.e., separated by energy difference $\Delta\omega$, are replaced with (perfectly) degenerated levels of averaged energy. Next, the standard derivation procedure of the Markovian master equation takes place~\cite{AlickiLendi1987,Breuer+2006}. In the final step, which is performed in the Schr\"odinger frame, the modified Hamiltonian is (by hand) replaced with the original, non-degenerated one.  Furthermore, for times
\begin{align}
    t\gg \frac{1}{\Delta \omega},
\end{align}
the (fully-) secular equation holds~\cite{levy2014local,cattaneo2019local}. For more groups, we have more time scales, and series of equations, which starts with ``most quasi-secular'' (or ``least secular'') and becomes more and more ''secular'' (till fully-secular equation), see for an illustrative description the Box 1.

The question we want want to address in the paper is: 
{\it is there a simple, accurate, and well-structured evolution equation that describes the whole range of time, hence covering also times $t\simeq \frac{1}{\Delta\omega}$ ?}

This question is important as in a generic system, 
there is no division between the well-separated groups of frequency differences at all, 
so there is even no chance to implement quasi-secular equations. 
The mentioned issue becomes especially significant for complex systems.

\subsection{Notes on the dynamics at intermediate times}\label{sec:NoteJ}
As noted in~\cite{Davies2} the 
most widespread microscopically derived master equation, the GKSL-Davies equation ~\cite{davies1974markovian} fails to describe time regimes where differences of Bohr frequencies are too small in comparison to the inverse of time relaxation of the system.
In Ref.~\cite{Alicki1989} this problem was considered in depth. 
Here, an oscillator with small  anharmonicity is considered, 
so that there are two energy scales: the basic one is given by the frequency of the oscillator, and a much smaller one is introduced via an anharmonicity coefficient $\chi$.

Then for short times, the environment does not have enough time (due to time-energy uncertainty) to distinguish the levels separated by small energy differences, so that the secular approximation is only done with respect to the basic frequency. This results in the master equation marked as type I in Ref.~\cite{Alicki1989},  where there is a cross-talking between those near degenerated Bohr frequencies. Such a master equation is supposed to work well for times evolution
satisfying 
\begin{align}
    t\ll \frac{1}{\chi}.
\end{align}
In Ref.~\cite{AlickiLendi1987} a  procedure of forming such an equation for a general quantum system was outlined.
A similar but much simpler situation -- the three-level system -- has been considered~\cite{Brumer_2018}, where experimental realization was proposed to show that for short times the equation of type I is the correct description. 
This type of equation has been later considered in the literature under the name of quasi-secular
such as in~\cite{McCauley2020} and~\cite{trushechkin2021unified}.

For times long enough, i.e. 
\begin{align}
    t\gg \frac{1}{\chi},
\end{align}
the nearby transitions are recognized by the environment, and the full secular approximation is legitimate, and then the master equation marked as type III in Ref.~\cite{Alicki1989} works, which is the standard Davies-GKSL master equation.  

Alas, there is a ``grey zone''. Namely, when the times are comparable to $1/\chi$, in which none of the equations (neither quasi-secular, nor fully-secular) works well.

One way to 
cover the grey zone is to use Bloch-Redfield equation \cite{redfield1957theory,bloch1957generalized}. Its time independent version is of comparable simplicity as Davies-GKLS equation, albeit less accurate than a more complicated time dependent version. Yet, as mentioned Redfield equation does not preserve positivity of the density matrix in general. In Ref.~\cite{Alicki1989} an completely positive evolution, aka ''the refined weak coupling limit'' \cite{Rivas_2017}, that 
 in particular covers  the grey zone was proposed. This methodology was further developed in Refs.~\cite{Rivas_2010,Rivas_2012,Rivas_2017,Rivas_2019,Mayenz-Lidar-coarse-graining}. However, the refined weak coupling approach is much more involved then Davies-GKSL equations (the complexity is comparable to time-dependent Bloch-Redfield equation).

There were other attempts to 
have completely positive evolution, that covers the ''grey zone''.
In Ref.~\cite{Mayenz-Lidar-coarse-graining}, 
the authors derive a \add{CPTP} master equation which might work in principle in any chosen, particular time scale.
However, it depends on an additional phenomenological parameter, the coarse-graining time. In Ref.~\cite{Mayenz-Lidar-coarse-graining} the fixed coarse-graining time is optimized with respect to the entire evolution, however it is expected that time-dependent coarse-graining procedure can reproduce the exact evolution even better\footnote{This is not evident from ~\cite{Mayenz-Lidar-coarse-graining}, still short-times behavior of dash-dotted black line in Fig.~5 therein suggests such a possibility.}.

Thus, if we are interested in the system evolution over a given time scale, we have to set the parameters suitably, i.e. we shall not be able to describe evolution through the whole time having fixed parameters of equation. The reason behind it is intuitive: since two extreme cases (quasi- and fully-secular dynamics) are Markovian with generators having completely different structures, there cannot be a Markovian generator that could produce a dynamics recovering the extreme generators as proper limiting cases. 
This might only happen if the generator is changed adiabatically, which requires much longer 
times scales than the inverse of any difference of Bohr frequencies in the system.



Essentially, the whole art still consists in making physically motivated approximations, while {\it keeping complete positivity and trace preservation } (CPTP property) of the evolution, in hope for obtaining simple equations that nevertheless work pretty well (for all times scales). 

\section{Refined weak coupling: the cumulant equation}
\label{sec:RWC}

In this Section, we recall 
the refined weak coupling limit 
here referred to as ''the cumulant equation''.  
In this treatment of the dynamics of an open quantum system weakly interacting with its environment, 
the Markovianity of the time evolution of the open system is not presumed. 
Therefore, in general, the time evolution of the reduced density operator 
is described with a time-dependent generator rather than with a quantum dynamical semigroup.

We consider a system interacting with a bath, with the following total Hamiltonian:
\begin{align}
    &H_{\mathcal{S}+\mathcal{R}} = H_\mathcal{S}+H_\mathcal{R}+\lambda H_\mathcal{I},\\
    &H_\mathcal{I} = \sum_i A_i \otimes B_i,
\end{align}
where $H_\mathcal{S}$ is the physical (measurable) Hamiltonian of the system, i.e., the one including any Lamb shift, 
 and  without loss of generality, we assume
that average of $B_i$ over bath state vanishes. The system Hamiltonian $H_\mathcal{S}$ is chosen such that in the interaction picture 
 the system will undergo only a dissipative dynamics~\cite{RenormalizationPaper}.
To derive the cumulant equation we follow the approach in Ref.~\cite{Alicki1989}, which we summarize in the following. 
The reduced dynamics of the system in the interaction picture is defined as the partial trace over the reservoir degrees of freedom 
\begin{equation}
\label{eq:tot_evolution}
    \rho_S(t) = \Lambda(t) \rho_S = \Tr_\mathcal{R} \left[ U(t, t_0)  \rho_S(t_0) \otimes \rho_R U^\dagger(t, t_0) \right]
\end{equation}
where $U(t, t_0) = \mathcal{T} \exp  \left\{ - i \lambda \int_{t_0}^t dt'  \tilde{H}_\mathcal{I}(t') \right\}$ is the  time-ordered unitary propagator of the whole system.

A formal expansion of the reduced dynamics can be written as $\Lambda(t)  \equiv \exp \left\{ \sum_{n=1}^\infty \lambda^n \tilde{K}^{(n)}(t) \right\},$ then assuming that the bosonic bath is approximately Gaussian 
(such that it is enough to consider up to the second-order correlation function), 
it is possible to consider terms up to the second order in $\lambda,$ arriving at the expression $\Lambda(t)  = \exp \left\{ \lambda^2 \tilde{K}^{(2)}(t) \right\}$, while the first term $K^{(1)}(t)$ results to be null due to the centralization of the bath operators.

Finally, one can write the time evolution in the interaction picture of the reduced density operator $\tilde{\rho}_S(t)$ 
associated with the open system, that is given by the equation:
~~~~
\vspace{-0.1em}
~~
\begin{mybox}[Box 2 - The FA equation]\label{box2}
\begin{minipage}[t]{1.\textwidth}
\includegraphics[width = 0.95\linewidth]{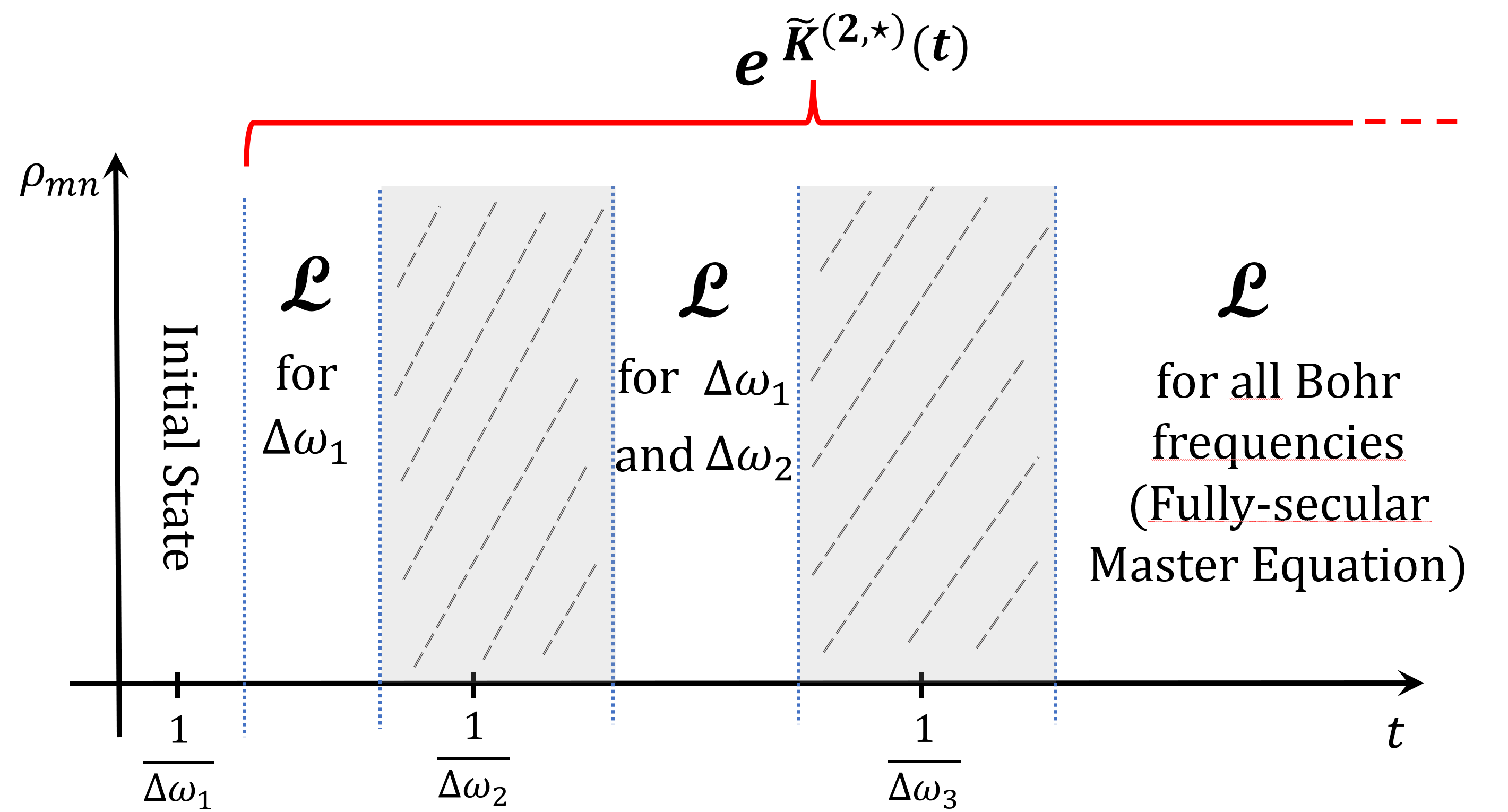}
\end{minipage}%
\hfill
\begin{minipage}{1 \textwidth}
\vspace{1em}
\label{box:2}
We introduce  a novel approximation technique when deriving the cumulant equation in the weak-coupling limit. The resulting filtered approximation (FA) is suitable for a variety of scenarios and works for almost all time regimes. The evolution will be guided by $e^{\tilde{K}^{(2,\star)}(t)}$  in contrast to the GKLS dynamics given by $e^{\mathcal{L}t}$.

The FA equation has an advantage of being particularly simple, while being  an impressively good (for all examined cases) approximation of the cumulant equation. Within this formalism, we preserve the complete positivity maintaining the physical insight of Bloch-Redfield equations, 
and we take trace of non-Markovianity where the Markovian-master equations, either quasi- or fully-secular, fail. 
\vspace{1em}
\end{minipage}
\end{mybox}
\begin{align}\label{eqn:GeneratorK}
	\tilde{\rho}_S (t) = e^{\tilde{K}^{(2)}(t)} \tilde{\rho}_S (0),
\end{align}
where the superoperator $\tilde{K}^{(2)}(t)$ 
 reads

\color{black}
\begin{align}
\label{eqn:KTime}
\tilde{K}^{(2)}(t)\tilde{\rho} = \frac{1}{\hbar^2} \int_0^t ds \int_0^t du \, e^{i (\omega^\prime s - \omega w)} \left< \tilde{B}_j (s) \tilde{B}_i (u) \right>_{\tilde{\rho}_B} \nonumber \\ 
\times \left( \tilde{A}_i(s) \rho \tilde{A}_j^\dag(u) - \frac12 \{\tilde{A}_j^\dag(u) \tilde{A}_i(s), \tilde{\rho} \} \right). 
\end{align}
The above superoperator takes into account dynamics guided by a time-dependent Hamiltonian $H_{\mathcal{S}+\mathcal{R}}(t)= \sum_i A_i(t) \otimes B_i $. The time dependence can also be present at the level of  $H_{\mathcal{S}}(t)$. 
Both cases can describe enthralling physical situations.  
The former case can be used to model all those situations where there is the possibility to tune in time the system-bath coupling, making it possible to recover the adiabatic limit (of a slowly varying coupling) or the Floquet picture of an interaction periodically switched on and off; the latter covers systems perturbed by an external field, 
such as atoms shined by laser light or in a magnetic field varying in time.

In this paper, we consider a total Hamiltonian $H_{\mathcal{S}+\mathcal{R}}$ time independent, then when moving into the frequency domain, the superoperator $\tilde{K}^{(2)}$  is given by 
\begin{align}\label{eqn:ActionOfCumulant}
	&\tilde{K}^{(2)}(t) \tilde{\rho}_S =
	\frac{1}{\hbar^2}\sum_{i,j} \sum_{\omega, \omega^\prime} \gamma_{ij}(\omega,\omega^\prime,t) \nonumber\\
	&\times  \left(A_i (\omega) \tilde{\rho}_S A_j^\dagger (\omega^\prime) - \frac{1}{2} \left\{A_j^\dagger (\omega^\prime) A_i (\omega), \tilde{\rho}_S \right\} \right).
\end{align}
In the above formula (that is in GKSL form) the  time-dependent relaxation coefficients $\gamma_{ij} (\omega,\omega^\prime,t)$ are elements of a positive semi-definite matrix, and are given by:
\begin{align}
	&\gamma_{ij} (\omega,\omega^\prime,t) = \int_0^t ds \int_0^t dw~ e^{i (\omega^\prime s - \omega w)} \left< \tilde{B}_j (s) \tilde{B}_i (w) \right>_{\tilde{\rho}_B}, \label{eqn:DEFgamma}
\end{align}
	where $\left< AB \right>_\sigma \equiv \tr \left\{AB\sigma\right\}$ and are the bath operators in the interaction picture $\tilde{B}_j (s)$. 
The above formula can be integrated into the following form
\begin{align}\label{eqn:KossakowskiCumulantIntegrated}
    \gamma_{ij} (\omega,\omega^\prime,t) =& e^{i \frac{\omega^\prime-\omega}{2}t} \int_{-\infty}^{\infty} d\Omega~
	    \left[t~ \mathrm{sinc} \left(\frac{\omega^\prime-\Omega}{2}t\right)\right] \nonumber \\
	    &\times \left[t~ \mathrm{sinc} \left(\frac{\omega-\Omega}{2}t\right)\right]  R_{ji} (\Omega),
\end{align}	
where $R_{ji} (\omega)$ is the autocorrelation function of the bath given by $R_{ji} (\omega) =\left< {B}_j (\omega) {B}_i \right>_{\tilde{\rho}_B}$, and $B (\omega)$ denotes the Fourier transform of $B (t)$ . 
\color{black}

Finally, the jump operators ${A}_i(\omega)$, ${A}_i^\dagger(\omega)$ are defined in the usual way:
\begin{align}\label{eqn:jumpStd}
	{A}_i(\omega)=\sum_{\epsilon^\prime - \epsilon=\hbar \omega} \Pi(\epsilon) {A}_i \Pi(\epsilon^\prime).
\end{align}
Here, the sum runs over Bohr frequencies of the {\it renormalized} 
Hamiltonian $H_\mathcal{S},$ and the operators  $\Pi(\epsilon)$  are then projections onto the eigenspace relative to the eigenenergy $\epsilon$ of the 
Hamiltonian $H_\mathcal{S}$ of the open system. 
This requirement is because the above evolution results from the central limit like theorem -- so that only cumulant up to the second-order is kept -- which requires eliminating all the {\it systematic} effects of noise (see~\cite{Alicki1989, RenormalizationPaper}).
We then have an evolution that is consistent with thermodynamics - i.e., leads to a stationary state which is the Gibbs state of the observed Hamiltonian, unlike in the approach that is widespread in the literature, where the stationary state is the Gibbs state of the bare Hamiltonian~\cite{Breuer+2006,Rivas_2017}.

From now on, we assimilate a system of units for which all physical constants are set to $1$.

\begin{remark}
Our notation differs slightly from the one used in reference~\cite{Rivas_2017}. This discrepancy can be most easily observed in equation (\ref{eqn:ActionOfCumulant}).  Namely, the positions of $\omega$ and $\omega^\prime$ in the Kossakowski matrix $\gamma$ are interchanged. We choose the arrangement of indices for which complete positivity of $\tilde{K}^{(2)}(t)$ is almost evident.
\label{rem:notation}
\end{remark}

\section{Approximations of the cumulant equation}
\label{sec:approximation} 

In this Section, we present the main results of the paper. Nonetheless, before showing how non-Markovian approximation arrives from the cumulant equation we first discuss the long-times regime (and its subtleties) in which the cumulant equation reproduces Markovian evolution. Next, we present a method to cope with the complexity problems emerging in the original cumulant equation by means of an approximation procedure. These will allow to have a ready-to-use dynamical equation \textit{\`a la Davies} that in addition allows to capture the non-Markovian evolution of the system in the intermediate time regime, for an illustrative description see Box 2. 

\subsection{The Markovian approximation}

The (fully-secular) Davies-GKSL master equation can be readily obtained from the cumulant equation~\cite{Rivas_2017}, by performing a specific limit on its superoperator. The integrated form of Davies-GKSL master equation follows then from the following replacement in Eq. \eqref{eqn:GeneratorK}.
\begin{align}\label{eqn:LongTimeLimitCumulant}
     \tilde{K}^{(2)}(t) \longrightarrow  t \tilde{\mathcal{L}}^\mathrm{fs} = t\left(\lim_{\tau \to +\infty} \frac 1\tau \tilde{K}^{(2)}(\tau)\right) \stackrel{t \approx \infty}{\approx} \tilde{K}^{(2)}(t),
\end{align}
where the limit in r.h.s. can be easily computed using standard methods, and $\tilde{\mathcal{L}}^\mathrm{fs}$ is the interaction picture time-independent generator of a semigroup of the Davies equation (see Sec. \ref{app:sec:RelaxGlob}). This generator is also labeled as fully-secular to distinguish it from the ones valid in other time regimes. The above procedure results in the Markovian master equation in secular approximation; therefore, all interesting memory effects are neglected. This situation motivates us to derive another kind of approximation that interpolates between non-Markovian dynamics and the Markovian case in a way that does not neglect memory effects, i.e, the filtered approximation (FA) of the cumulant equation.

\begin{remark}
    In rigorous mathematical terms $ t \tilde{\mathcal{L}}^\mathrm{fs}$ is not a long-time limit of the cumulant superoperator $\tilde{K}^{(2)}(t)$~\cite{RenormalizationPaper}. However, for
    the renormalized (no Lamb shift term) cumulant equation, studied here,
    for some large $t_0$ we have 
    \begin{align}
        \forall_{t>t_0} ~~ \norm{\tilde{K}^{(2)}(t)-t \tilde{\mathcal{L}}^\mathrm{fs}}_1 < C.
    \end{align}
    where $C$ is an absolute contant.
    Thus, the cumulant equation  reproduces the dynamics of the (renormalized) Davies-GKSL equation in the long-time limit~\cite{MW_preparation}~(cf. Ref.~\cite{Rivas_2017}). Therefore, the approximation in Eq.~\eqref{eqn:LongTimeLimitCumulant} can be understood effectively in the sense of resulting dynamics. Still, the long-time limit of the cumulant superoperator ${\tilde{K}^{(2)}(t)}$ is given by a more involved expression~\cite{RenormalizationPaper}.
\end{remark}


\subsection{Filtered approximation (FA)}\label{sec:SM}
Here we provide the main result of the paper - the ''filtered approximation'' (FA). 
The resulting equation is much simplified but the dynamics is still non-Markovian, in a sense that the dynamical map of the FA equation does not constitute a quantum dynamical semigroup. 


In simple words, in the secular approximation the ``sinc'' functions in the integrands of the elements of the Kossakowski matrix in the generator are replaced with Dirac's deltas. In our approximation method, that starts at the level of cumulant equation dynamical map, only one ``sinc'' function (out of two, in a product appearing in integrands) is replaced. In this way we retain non-Markovianity, but severely simplify the structure. Next, we recover the CPTP structure by square root techniques~\cite{Vacchini,Davidovic2020}. The mathematical details of the derivation are described in Section \ref{app:sec:SM} of the Appendix.
\begin{align}
\label{eq:singlestar}
    &\gamma_{ij} (\omega,\omega^\prime,t) \approx \gamma_{ij}^\star (\omega,\omega^\prime,t) \nonumber \\
    &= 2\pi t e^{i \frac{\omega^\prime-\omega}{2}t}
	     \mathrm{sinc} \left(\frac{\omega^\prime-\omega}{2}t\right) \sum_k R^{\frac 12}_{jk} (\omega^\prime) R^{\frac 12}_{ki} (\omega),
\end{align}
where $R^{\frac 12}_{ij} (\omega) \equiv(R^{\frac 12}(\omega))_{ij}$. Furthermore, $\gamma_{ij}^\star (\omega,\omega^\prime,t)$ is a positive semi-definite matrix, which guarantees CPTP dynamics, and it is a generalization of 
{\it Markovian relaxation rates} $\gamma_{ij} (\omega)=2\pi R_{ij} (\omega)$ known from the Davies equation.
In particular, the terms diagonal in $\omega'$ 
$\omega's$ 
reproduce the Markovian rate: $\gamma_{ij}^\star(\omega,\omega,t)=2 \pi t R_{ij}(\omega)$.

We can now write the FA approximation of the cumulant equation as follows:
\begin{align}\label{eqn:DynMapStar}
    \tilde{\rho}^\star_S (t) = e^{\tilde{K}^{(2,\star)}(t)} \tilde{\rho}_S (0),
\end{align}
where the generator  $\tilde{K}^{(2,\star)}(t)$ is obtained from Eq.  \eqref{eqn:ActionOfCumulant} by inserting 
the above approximated time-dependent relaxation coefficient $\gamma_{ij}^\star (\omega,\omega^\prime,t)$, in place of $\gamma_{ij}(\omega,\omega^\prime,t)$.


The remarkable property of this approximation is the simplicity of the final formulas. Here, the coupling to the reservoir is described with a finite (and small) number of parameters; however, the dynamics is CPTP and non-Markovian in the above sense. Therefore, this approximation is of perfect use in situations in which no good microscopic and theoretical model of the bath spectral density is known. In these cases, $R$ can be constructed with phenomenological values (measured in an experiment) and used to describe the system's dynamics in short and intermediate times that might be beyond the reach of spectroscopy. 
We note here that the technique of splitting the spectral density by means of the square root was considered, see e.g.~\cite{SquareRoot1,Davidovic2020}.

\begin{remark}
    In Appendix \ref{app:sec:SM2} we derive yet another type of approximation of the cumulant equation. The II-type approximation is not less complex in computation than the cumulant equation itself, however (alike FA approximation) it has a property of being ''cutoff-stable'' (see Appendix \ref{app:sec:SM2} for more details).
\end{remark}

\section{The Spin-Boson model dynamics}

In this Section, we substantiate our results on the popular testbed of the transverse spin-boson model~\cite{Surez1992,DiVincenzo_1995,Steane_1998,BEZ_2000}, that despite its simplicity does not have an exact analytical solution (at $T>0$) for all the time regimes. We first, briefly describe the model. Subsequently, we compare the dynamics obtain with the cumulant equation and its regularizations with the ones obtained with Davies-GKSL equation, Bloch-Redfield equation and exact numerics. Finally, we present the evidence for non-Markovianity of the FA aproximation dynamics.

\subsection{The system}

\label{sec:spin-boson}

The spin-boson model concerns a two-level system linearly coupled to a bath of harmonic oscillators. The Hamiltonians of the model take the form:
\begin{align}
    H_S = \frac{\omega_0}{2} \sigma_z,~~
    H_B = \sum_k \omega_k a_k^\dagger a_k,\\
    H_\mathcal{I} = \sigma_x \otimes \sum_k g_k \left( a_k + a_k^\dagger\right),
\end{align}
where $\omega_0$ is the transition frequency between ground and excited state, denoted $\ket{g}$, $\ket{e}$ respectively. In the following basis:
\begin{align}
    \ket{g} = \begin{pmatrix}
0 \\
1 
\end{pmatrix},~
    \ket{e} = \begin{pmatrix}
1 \\
0 
\end{pmatrix},
\end{align}
we obtain the exact form of the system's jump operators
\begin{align}
    &A(\mp \omega_0) = \sigma_\pm = \frac{\left(\sigma_x \pm i \sigma_y\right)}{2}.
\end{align}
Furthermore, $B = \sum_k g_k \left( a_k + a_k^\dagger\right)$ are the reservoir's operators, with $a_k^\dagger$, $a_k$ being the usual creation and annihilation operators respectively. Then, in the continuum limit, for a heat bath in a Gibbs state, we obtain the following~\cite{Rivas_2010,Rivas_2017}
\begin{align}
    R(\Omega) = J(\Omega)\left( N\left(T,\Omega \right)+1\right),
\end{align}

where $N (T,\Omega)=\left[\exp{{\Omega}/{(T)}}-1\right]^{-1}$ is the Bose-Einstein distribution at temperature $T$. We transform to the neutral system of units, by setting $\omega_0$ as a reference frequency, so that $N (T,\Omega)=N (T_\mathrm{eff},\frac{\Omega}{\omega_0})$, with $T_\mathrm{eff}=T/\omega_0$. The real temperature would read in this case $T= \frac{\hbar \omega_0} { k_{B}}T_\mathrm{eff}$, and $J(\Omega)$ is the spectral density.

\subsection{Cumulant vs Redfield, Davies and exact numerics}

Using the content of subsection \ref{sec:spin-boson} above, we obtain a specific form for the action of the superoperator of the cumulant equation (cf. Ref.~\cite{Rivas_2017}, see Remark~\ref{rem:notation})
\eqref{eqn:ActionOfCumulant}:
\begin{align}\label{eqn:SBcumulant}
	&\tilde{K}^{(2)}(t) \tilde{\rho}_S =
	\sum_{\mu,\nu=\pm} \Gamma_{\mu\nu}(t)   \left(\sigma_\mu \tilde{\rho}_S \sigma_{\nu}^\dagger - \frac{1}{2} \left\{\sigma_\nu^\dagger \sigma_\mu, \tilde{\rho}_S \right\} \right),
\end{align}
where we follow the notation from reference~\cite{Rivas_2017}, by denoting $\Gamma_{\mu\nu}(t) \equiv \gamma(-\mu \omega_0,- \nu \omega_0,t)$ (see equation \eqref{eqn:KossakowskiCumulantIntegrated}).

The cumulant equation for the spin-boson model has been studied before in the literature~\cite{Rivas_2019}. However, a comparison of the cumulant equation dynamics, and its (also non-Markovian) alternatives with an exact solution is required, to assess and compare their accuracy. In this section,
we use a testbed of the spin-boson model to provide such a comparison. We follow the choice of parameters proposed in Ref.~\cite{Rivas_2019}, as we want to compare our results with the results therein. This choice goes beyond the weak coupling regime, in which the second order dynamical equation are most accurate. Still, (anticipating a bit) similar qualitative conclusions hold for any other reasonable set of parameters (see Fig.~\ref{fig:heom_plots}). Yet weaker coupling diminishes discrepancy between solutions due to different dynamical equations, and therefore obscures the presentation.  

In Fig.~\ref{fig:heom_plots} (bottom plot) we observe that the solution to the dynamics due to the cumulant equation has higher fidelity with the numerically exact solution, obtained with the hierarchical equations of motion (HEOM), than any other considered equation, i.e., Bloch-Redfield equations~\cite{cattaneo2019local}, Davies-GKSL Markovian master equation and the FA equation. In the middle and top plot of Fig.~\ref{fig:heom_plots} we observe that the cumulant equation predictions for the evolution of populations and coherence follow closely the exact solution due to HEOM. For long times populations predicted by the cumulant equation matches the predictions obtained with Davies-GKSL equation, hence does not match the HEOM dynamics (top plot in Fig.~\ref{fig:heom_plots}). The discrepancy for longer times is due to a mismatch between bare and renormalized (effective) Hamiltonians. Here, the dynamics for all methods is computed w.r.t. to the bare Hamiltonian.  To recover an agreement between HEOM and the cumulant equation at later times the renormalized Hamiltonian should be used for cumulant equation instead~\cite{RenormalizationPaper,Lobejko_Mean-Force}, however this way accuracy at shorter times can be diminished~\cite{correa2023potential}. The renormalization methodology is also applicable for other equations presented in the simulation herein, see Remark~\ref{rem:MF_Ren} at the end of this subsection. Finally, the beyond-weak coupling regime makes the discrepancy non-negligible, and visible in the plot.


\begin{figure}
    \includegraphics[width=1 \columnwidth]{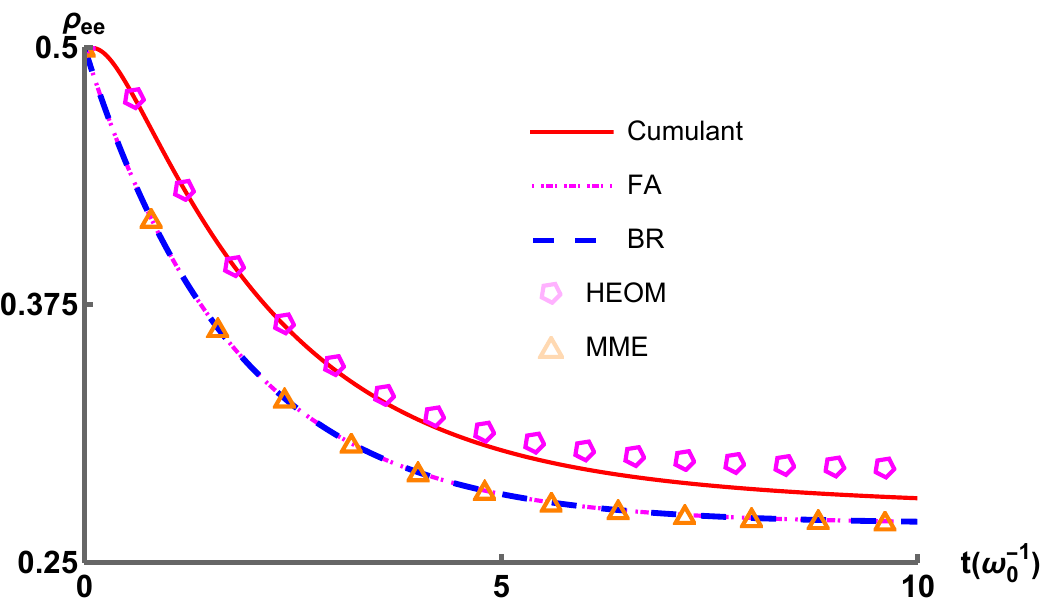}
    \includegraphics[width=1 \columnwidth]{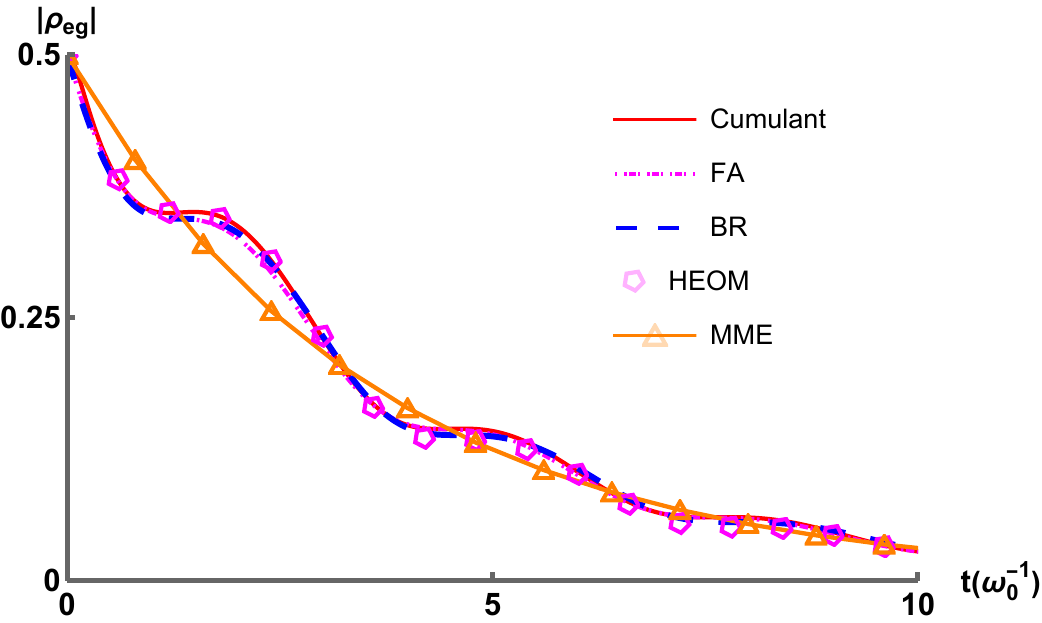}
    \includegraphics[width=1 \columnwidth]{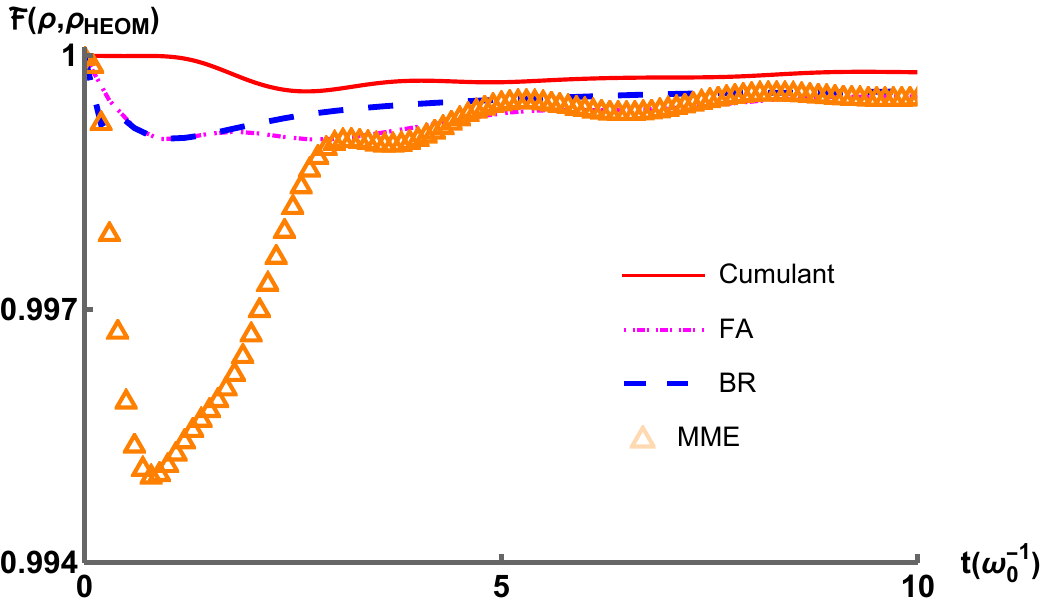}
    \caption{The plots show  a comparison of the original cumulant equation for the spin-boson model with the Hierarchichal Equations of motion (HEOM) which is  a method to obtain numerically exact solutions and the Bloch-Redfield equation which is the most popular equation for non-markovian dynmaics. In 
    (a) we report the evolution of the excited state population, in (b) the absolute value of the coherence between the ground state and the excited one. 
    The parameters chosen for this simulation are $T=1$, $\omega_{c}=5$, $\alpha=0.05$, $\omega_{0}=1$ and absolute and relative tolerance for numerical integration $\epsilon=10^{-14}$ for all of the methods above.  
    In (c) we report the fidelity calculated between the density matrix obtained considering the Bloch-Redfield equation and the cumulant equation with respect to the solution of the HEOM considered as target state. The HEOM simulation a maximum hierarchy of seven ADOS. We can observe all Non-Markovian approaches are close to the numerically exact dynamics, a fact that is reflected in their really high value of fidelity. The cumulant equation shows a better accuracy than the Bloch-Redfield equation. In the plots the cumulant equation corresponds to the red solid line, the FA equation corresponds to the dash-dotted pink line, the Bloch-Redfield equation (BR) corresponds to the dashed blue curve, while the orange curve with markers denote the Markovian master equation (MME).}
    \label{fig:heom_plots}
\end{figure}

The plots in Fig.~\ref{fig:heom_plots} let us also compare the dynamics provided by the FA equation with the solutions due to the cumulant equation, numerically exact method of HEOM, and other equations. In the top plot of Fig.~\ref{fig:heom_plots}, we observe that the evolution of populations predicted by the FA equation matches almost perfectly the evolution obtained with Davies-GKSL for all times, therefore discrepancy with the cumulant equation at shorter times must be reported. Interestingly, the middle plot in Fig.~\ref{fig:heom_plots} shows that the FA equation predicts the evolution of coherence that well matches the evolution obtained with the cumulant equation and HEOM. Finally, we report that the FA equation exhibits quite high fidelity with the exact solution (bottom plot in Fig.~\ref{fig:heom_plots}). The dynamics provided by the FA equation has better fidelity with exact solution that Davies-GKSL equation (especially at short times), and similar fidelity to Bloch-Redfield equation, however a bit worse fidelity than cumulant equation itself. The details of the simulation and an elaborated description of methods are available in the Appendix.

\begin{remark}\label{rem:MF_Ren} 
    While the dynamics due to the HEOM, presented here, converges to the Gibbs state w.r.t. an effective, the so-called mean-force Hamiltonian~\cite{Lobejko_Mean-Force,trushechkin2022}, the dynamics due to other equation does not. Using renormalization technique proposed in Ref.~\cite{RenormalizationPaper} and the mean-force Hamiltonian corrections in Refs.~\cite{Lobejko_Mean-Force,trushechkin2022} it is possible to modify Davies-GKSL, FA and the cumulant equation so that they converge to mean-force Gibbs state as HEOM does (formally up to the second order). However, increase in fidelity for long times would result in loss of accuracy at shorter times scales~\cite{correa2023potential}. There is a hope that more sophisticated approach to renormalization can work for all time scales. 
\end{remark}

\subsection{Application of the approximated dynamics}

\begin{figure*}[ht]
\includegraphics[width=1 \columnwidth]{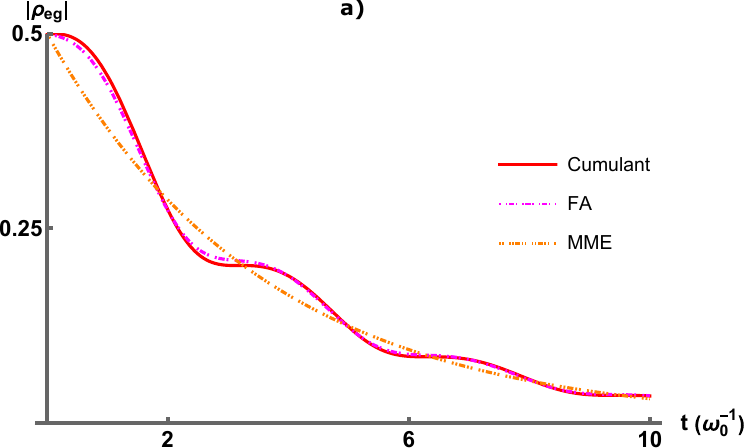}
\includegraphics[width=1 \columnwidth]{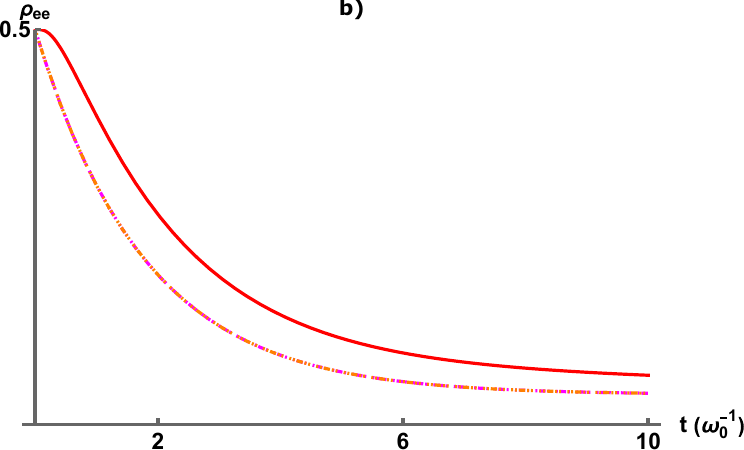}
\caption{Spin-Boson model -- Interaction picture evolution of the coherence (left), and population of the excited state (right) of the reduced density matrix of the system $\rho_S$. The initial state of the system is ${\rho_S}_{ij}(0)=\frac 12$. The reservoir is a heat bath at temperature $T_\mathrm{eff}=1$ in units of $\omega_0$. The short-dash-dotted orange curve is the evolution computed via secular Davies master equation, the solid red curve is the evolution obtained with the cumulant equation in Eq.~\eqref{eqn:GeneratorK}, and the dashed pink curve refers to the FA equation in Eq.~\eqref{eqn:DynMapStar}. In the b) plot populations due to FA equation and MME coincide. The system-bath coupling constant and the cutoff frequency are $\alpha =0.05$, $\omega= 5 \omega_0$, respectively (see also Ref.~\cite{Rivas_2017}). } 
\label{fig:comp_SB}
\end{figure*}

Here, we address a comparison of the different techniques that we have so far discussed. In particular, we report the evolution of the population of its excited state ${\rho_S}_{ee}$ and of the modulus of the coherence ${\rho_S}_{eg}$ between its excited and the ground state, having as initial state $\left(\rho_S(0)\right)_{ij}=1/2$. The FA equations are readily obtain by replacing $\Gamma(t)$ in equation \eqref{eqn:SBcumulant} by $\Gamma^\star_{\mu\nu}(t) \equiv \gamma^\star(-\mu \omega_0,- \nu \omega_0,t)$ (see Eq.~\eqref{eqn:DynMapStar}).

In Fig.~\ref{fig:comp_SB}, we compare in greater detail the dynamics obtained with the FA equation with the dynamics given by the cumulant equation and the dynamics given by the secular Davies-GKSL equation on the test ground of the spin-boson model. In particular, we observe that the evolution of the population predicted by the FA equation (almost exactly, up to small oscillations) matches the evolution predicted by the Davies-GKSL equation, whereas the cumulant equation forecasts slower decay of the excited state population. Regarding the evolution of the (modulus of) coherence, we observe a good agreement between the predictions of the FA equation and the cumulant equation. The values of the modulus of coherence given by the FA equation and the cumulant equation oscillate around the decay curve predicted by the secular Davies-GKSL equation. 

The deviation of the evolution predicted by the FA equation from the evolution predicted by the Davies-GKSL equation (Markovian master equation in secular approximation) is a good sign that our approaches can be used to model an experimental situation where the standard quantum optics master equation is no longer valid~\cite{Andersson2019}. At the same time, we refer~\cite{suarez1992memory} for an extensive treatment of the non-positivity of the Bloch-Redfield equation for the spin-boson model.




This promising result allows us to surmise that the FA equation we have introduced can be applied to tackle the experimental situation showing  Fano coherences and coherent transport of excitation in light-harvesting systems~\cite{Brumer_2018, koyu2021}.

\subsection{A note on the spectral density}\label{sec:SD}

In this place, we discuss our choice of the profile of the spectral density used to plot the dynamics obtained with the Davies-GKSL equation, the Bloch-Redfield equation, the cumulant equation, its FA approximation and exact numerics of HEOM. Following Ref.~\cite{Rivas_2017}, we choose an Ohmic spectral density with exponential cutoff profile, i.e, $J(\Omega) = \alpha \Omega \times e^{-\frac{\abs{\Omega}}{\omega_c}}$, with cutoff frequency $\omega_c$. The presence of cutoff removes divergences from Eq.~\eqref{eqn:KossakowskiCumulantIntegrated}, and consequently from Kossakowski matrix of cumulant superoperator. It allows to perform numerical integration and obtain the results plotted in Fig. \ref{fig:comp_SB} a) and b). The magnitude of the coupling constant (for the sake of Fig.~\ref{fig:heom_plots}, Fig.~\ref{fig:comp_SB} and Fig.~\ref{fig:TD}) is $\alpha=0.05$, and $\omega_c = 5 \omega_0$. Our choice of parameters allows to compare our findings with the results in Ref.~\cite{Rivas_2017}, and falls within a reasonable range of values~\cite{hartmann2020accuracy}. Despite the fact, that formally the above choice of the magnitude of the coupling constant goes beyond the weak coupling regime, we observe in the bottom plot of Fig.~\ref{fig:heom_plots}, high fidelity  of the cumulant equation and FA equation dynamics (w.r.t. HEOM), what confirms the above statement.
We remark that, different choices of the spectral density and cutoff profile account for different rates at which the system approaches the equilibrium. In particular, it is evident, from the plots in Fig.~\ref{fig:comp_SB} that the cumulant equation  with $\omega_c =5\omega_0$  oscillates around the Markovian evolution, whereas higher cutoff frequency might result in completely different behavior.


\subsection{Evidence of non-Markovianity}

A natural to see if the dynamics is non-Markovian~\cite{BLP} is to study the evolution of the trace distance between two states, namely $\frac12 || \rho - \sigma || = \frac12 \Tr{\sqrt{(\rho - \sigma)^\dagger (\rho - \sigma)}}$. A Markovian evolution, tending monotonically to the thermal equilibrium, will force two states
to become less and less distinguishable as time passes. 
In Fig. \ref{fig:TD} we report the evolution of the trace distance 
between two initially orthogonal states $\sigma_{ii}=\frac 12,~~\sigma_{eg}=-\frac 12i$ and $\rho=\sigma^*$ for the spin-boson model developed in Section \ref{sec:spin-boson}.


\begin{figure}[h!]
    \centering
     \includegraphics[width=1 \columnwidth,trim={0cm 0cm 0 0},clip]{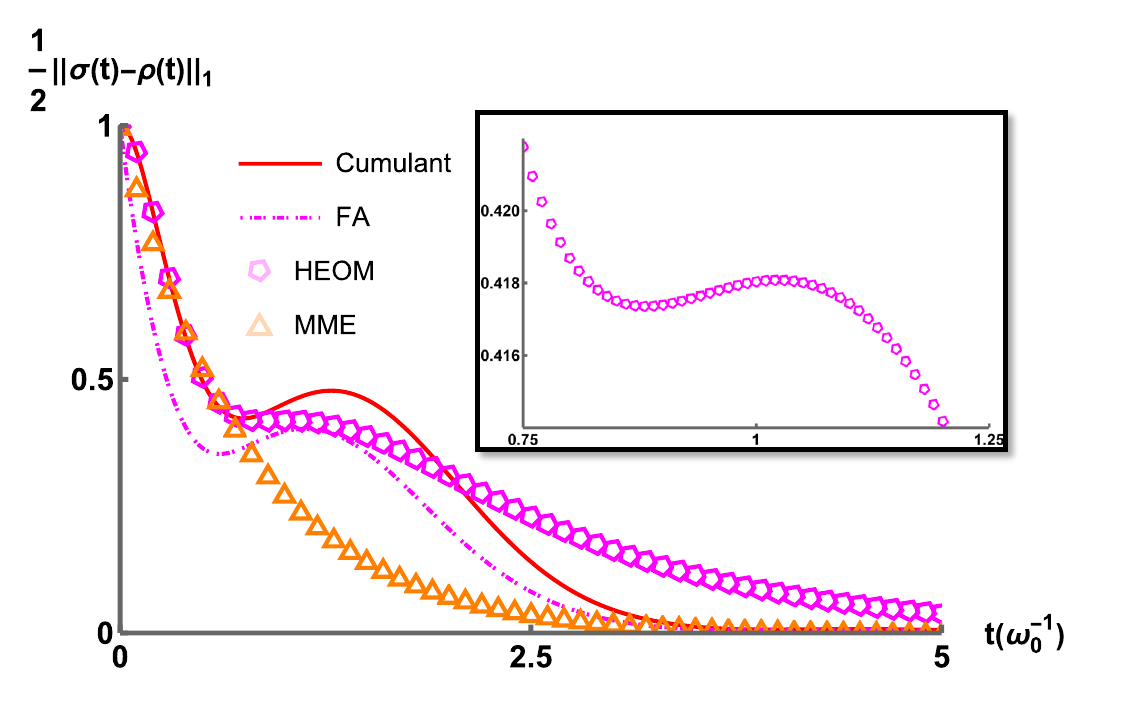}
    \caption{ The plot of trace distance between the reduced density matrices obtained with initial states $\sigma_{ii}=\frac 12,~~\sigma_{eg}=-\frac 12i$, $\rho=\sigma^*$ and different dynamical equations. Dynamics given by the FA equation is given by a short-dash-dotten pink curve, by the cumulant equation by a solid red curve, by the fully-secular Markovian master equation (MME) by a short-dash-dotted orange curve, and the dynamics obtained with HEOM is given by pink markers. The non-monotonic behaviors of curves, obtained with the cumulant equation, FA equation and HEOM, indicate the non-Markovianity of dynamics. The inset highlights the non-monotonic behavior of the trace distance curve corresponding to HEOM. Here, the temperature is $T_\mathrm{eff}=5$, the system-bath coupling is $\alpha=0.05$, and $\omega_c=5\omega_0$. The cumulant equation shows a better accuracy than the Bloch-Redfield equation. In the plots the cumulant equation corresponds to the red solid line, the FA equation corresponds to the dot-dashed pink line, the BR equation corresponds to the dashed blue line, while the grey  hombuses and orange triangles denote the Bloch-Redfield equation with the partial secular approximation and the GKLS equation respectively}
    \label{fig:TD}
\end{figure}

An increase in distinguishability (see Fig.~\ref{fig:TD}) between the two orthogonal states evolving under the same dynamical map is indeed a signature of how the cumulant equation, our FA equation and HEOM exhibit non-Markovian dynamics and the backflow of information from the environment to the system~\cite{BLP, RHP, CKR11} (see also Fig.~2 in Ref.~\cite{Rivas_2017}). The non-monotonic behavior of trace distance evolution provided by HEOM is less evident than for the cumulant equation or for FA equation. This surprising feature of the cumulant and FA equation requires and deserves a further investigation that is beyond the scope of this manuscript. 

Here, we need to report that we were able to show the non-monotonicty of the trace distance solely in a regime of relatively strong coupling, i.e., $\alpha=0.05$ and strong driving ($T_\mathrm{eff}=5$). While in Ref.~\cite{Rivas_2017} the non-Markovianity for the cumulant equation was shown at low temperatures ($\alpha=0.05$, $T_\mathrm{eff}=0$), the FA equation only  becomes non-Markovian at higher temperatures (even for $T_\mathrm{eff}=1$ we didn't observe the effect). Furthermore, in the regime in which the cumulant and the FA equations are most accurate, i.e., in the weak coupling regime, we were unable to provide evidence for non-Markovianity of dynamics via non-monotonicity of the trace distance. This brings us to question pure existence of non-Markovianity (indicated by non-monotonicity of the trace distance) in the weak coupling regime.

\section{Convergence to quasi- and fully-secular master equations.} \label{sec:Interpolation}
The necessity to employ a fully- or quasi-secular master equation can be traced back to the problem of choosing a proper time interval when the dynamics of the system should be studied. 
To clarify the following discussion, we consider a complex system  whose transition frequencies can be grouped into two or more well-separated groups. 
Obviously, the simplest case presenting this configuration is a three-level system (also known as a qutrit) -- see Section \ref{sec:qutrit_b} for the details.

In principle the cumulant equation properly approximates the dynamics in all time regimes~\cite{Alicki1989,RenormalizationPaper}. In particular, the cumulant equation reproduces quasi- and fully-secular master equations for short and long time scales, respectively.
Similarly, generator of the FA equation $\tilde{K}^{(2,\star)}(t)$  developed in in Section \ref{sec:SM}, reproduces the generators of the quasi- and fully-secular master equations in short and long times regimes, respectively. Namely, we  have 
\begin{align}
    &{\tilde{K}^{(2,\star)}(t)} \stackrel{t \to \infty}{\approx} {t \tilde{\mathcal{L}}^\mathrm{fs}},
    \\
    &{\tilde{K}^{(2,\star)}(t)} \stackrel{\frac{1}{\Delta\Omega} \ll t \ll \frac{1}{\Delta \omega}}{\approx} {t \tilde{\mathcal{L}}^\mathrm{qs}},
\end{align}
where $\Delta\omega$ and $\Delta \Omega$ are as in Sec.~\ref{sec:TimeScales} (see also Box 1).
The above long-time limit is understood effectively at the level of dynamics, as in the case of the cumulant equation. The remnant (bounded) terms in dissipators reflect the non-Markovian characteristics of dynamical equations, yet as time increases the diagonal (Markovian) parts of the superoperatros dominate and the influence of the bounded terms becomes less and less important for the dynamics~\cite{MW_preparation}.

From Eq.~\eqref{eq:singlestar} it is almost evident that the  FA equation reproduces the fully-secular master equation in the long-time limit.  The basic idea of the proof that $\star$-regularized cumulant equation, i.e., FA equation quasi-secular master equation (in adequate time regime) is to gather the transition frequencies in (at least two) well-separated groups. Next, a series of approximations can be justified. A more thorough explanation and detailed proofs of the above properties are provided in Appendix~\ref{app:sec:Relax}.
Below we provide a sketch of the proof in the special case of V-type three level system.

Consider 
such a system with transition frequencies $\{\omega_i\}_{i=1,2}=\pm\left\{\omega_0\pm\Delta/2\right\}$, in the $\frac{1}{\omega_0} \le t \le \frac{1}{\Delta}$ time regime ($\omega_0, \Delta>0$). The sketch of our proof start with writing approximate form of the superoperator ${\hat{K}}^{(2)}(t)$ (which refers to cumulant equation or FA equation) in the Schr{\"o}dinger picture~(see Eq.~
\eqref{eqn:ActionOfCumulantSch}, and also Ref.~\cite{RenormalizationPaper}). For conciseness of presentation we skip the extra indices of jump operators.
\begin{widetext}
\begin{align}\label{eqn:step1}
	{\hat{K}}^{(2)}(t) \rho_S =-it \left[H_\mathcal{S},\rho_S\right]+ \sum_{\omega, \omega^\prime} \left(1+\frac{it}{2} (\omega-\omega^\prime)\right)\hat{\gamma} (\omega,\omega^\prime,t)   \left(A (\omega) \rho_S A^\dagger (\omega^\prime) - \frac{1}{2} \left\{A^\dagger (\omega^\prime) A (\omega), \rho_S \right\} \right)+\cdots 
\end{align}
In the next step we group Bohr frequencies (in the case of V-system the grouping is trivial), and perform the following approximation to the elements of Kossakowski matrix
\begin{align}
    \left(1+\frac 12 [it(\omega-\omega^\prime)]\right)\hat{\gamma} (\omega,\omega^\prime,t) \approx \left(1+\frac{it}{2} (\mathrm{Sgn}(\omega)\omega_0-\mathrm{Sgn}(\omega^\prime)\omega_0)\right)\hat{\gamma} (\mathrm{Sgn}(\omega)\omega_0,\mathrm{Sgn}(\omega^\prime)\omega_0,t)\equiv \Gamma_{\pm\pm}. 
\end{align}
The above approximation is justified in the $\frac{1}{\omega_0} \le t \le \frac{1}{\Delta}$ time regime, provided that spectral density of the reservoir does not vary too much with each group. At this point, we also skip the residual terms
\begin{align}\label{eqn:step2}
    \eqref{eqn:step1}&\approx-it \left[H_\mathcal{S},\rho_S\right]&+ \sum_{\pm_1,\pm_2} \sum_{\pm_3,\pm_4} \Gamma_{\pm_1\pm_2}   \left(A (\pm_1\omega_0\pm_3\Delta) \rho_S A^\dagger (\pm_2\omega_0\pm_4\Delta) - \frac{1}{2} \left\{A^\dagger (\pm_2\omega_0\pm_4\Delta) A (\pm_1\omega_0\pm_3\Delta), \rho_S \right\} \right) 
\end{align}
The second summation allows to introduce ''quasi-secular'' jump operators $A^\text{qs}(\omega_0)=A(\omega_1)+A(\omega_2)$.
\begin{align}\label{eqn:step3}
    \eqref{eqn:step2}=-it \left[H_\mathcal{S},\rho_S\right] + \sum_{\pm_1,\pm_2} \Gamma_{\pm_1\pm_2}   \left(A^\text{qs} (\pm_1\omega_0) \rho_S {A^\text{qs}}^\dagger (\pm_2\omega_0) - \frac{1}{2} \left\{{A^\text{qs}}^\dagger (\pm_2\omega_0) A^\text{qs} (\pm_1\omega_0), \rho_S \right\} \right)
\end{align}
Subsequently, we perform the secular approximation, that eliminates terms with opposite signs.
\begin{align}
    \eqref{eqn:step3}=-it \left[H_\mathcal{S},\rho_S\right] + \sum_{\omega=\pm\omega_0} \Gamma_{\pm\pm}   \left(A^\text{qs} (\omega) \rho_S {A^\text{qs}}^\dagger (\omega) - \frac{1}{2} \left\{{A^\text{qs}}^\dagger (\omega) A^\text{qs} (\omega), \rho_S \right\} \right). 
\end{align}
The final step is to show that 
\begin{align}
    \Gamma_{\pm\pm} = \hat{\gamma} (\pm \omega_0,\pm \omega_0,t) \approx t \gamma(\pm \omega_0),
\end{align}
where $\gamma(\omega)$ is the Kossakowski matrix of the fully-secular Davies-GKSL equation. The above is readily obtained (with ''$=$'') for FA approximation of the cumulant equation (see Eq. \eqref{eq:singlestar}). However, in the case of the cumulant equation (that should a priori correctly describe all times regimes), the proof requires a deeper insight. The above considerations yield
\begin{align}
    {\hat{K}}^{(2)}(t) \rho_S \approx -it \left[H_\mathcal{S},\rho_S\right] + t \sum_{\omega=\pm\omega_0} {\gamma}(\omega)   \left(A^\text{qs} (\omega) \rho_S {A^\text{qs}}^\dagger (\omega) - \frac{1}{2} \left\{{A^\text{qs}}^\dagger (\omega) A^\text{qs} (\omega), \rho_S \right\} \right). 
\end{align}
For a more detailed and general version of the proof see Appendix \ref{app:sec:Relax}.
\end{widetext}

\begin{figure}[ht]
    \centering
    \includegraphics[width=1 \columnwidth]{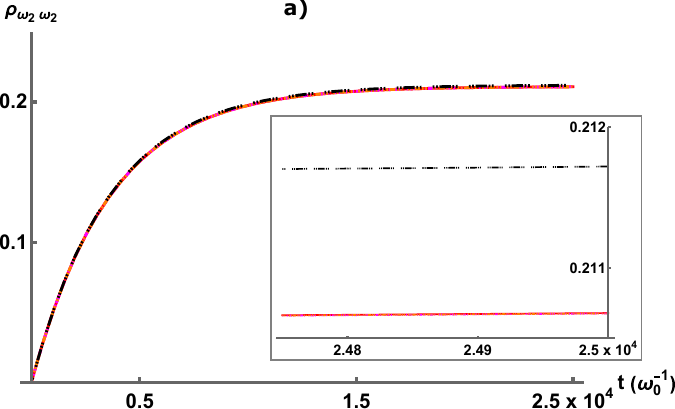}
    \includegraphics[width=1 \columnwidth]{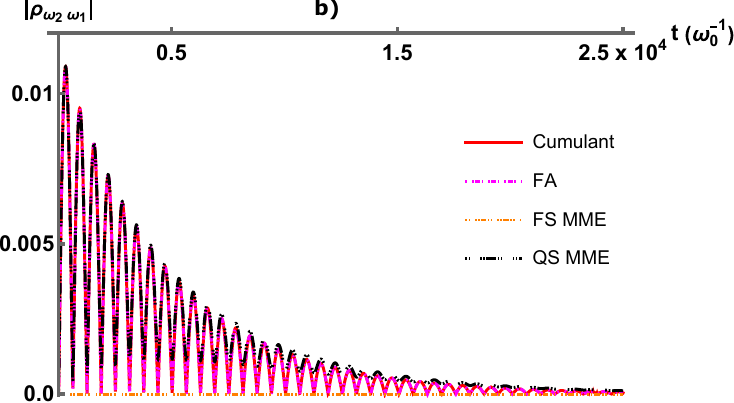}
    \caption{Qutrit-Boson model -- interaction picture evolution of the a) population ${\rho}_{\omega_2 \omega_2}$ of the excited state b) modulus of coherence $\abs{{\rho}_{\omega_2 \omega_1}}$ between the two excited levels. The inset in a) highlights the long-times behavior of the curves. The initial state of the system is the ground state. The dashed pink curve represents the evolution given by the FA equation, the short-dash-dotted orange curve represents the evolution provided by the fully-secular Markovian master equation, and the long-dash-dotted black line represents the evolution given by the quasi-secular Markovian master equation. The cumulant equation given by solid red curve is indistinguishable from the curve for the FA equitation in the given resolution. The reservoir is a heat bath at effective temperature $T_\mathrm{eff}=1$ in units of $\omega_0$. The spectral density oh the reservoir is Ohmic with exponential cutoff, i.e., $J(\omega) = \alpha \omega \exp{- \abs{\frac{\omega}{\omega_c}}}$, where $\omega_c=10 \omega_0$.  The choice of the square of coupling constant $\alpha \approx 1.76\times 10^{-5}$ and the splitting parameter $\Delta\omega= \omega_0/100,$ correspond to the following ratio between the rate $\gamma$ of spontaneous processes and the splitting $\Delta\omega$ between excited states: $\gamma/\Delta\omega=0.01$, which can be classified as a weak coupling. See also Fig.~\ref{fig:3lvlV2}. 
    }
    \label{fig:3lvlV1}
\end{figure}

\subsection{Application of approximations to the qutrit-boson model}
\label{sec:qutrit_b}

In this Section, we chose yet another model to test our approach. As we announced in the introduction, the cumulant equation and FA equation are able to properly describe the dynamics in all time regimes. The simplest example of a system in which more than one time scale appears is a three level V-type (or $\Lambda$) system. For this kind of system, quasi- and fully-secular master equations, that a priori describe different and disjoint time scales, yield distinct prediction~\cite{Brumer_2018,koyu2021}. Here, we show that cumulant equation and FA equation indeed interpolate between the predictions of  quasi- and fully-secular master equations in the ``grey zone'' (intermediate time regime), whereas in short and long time scales they reproduce the evolution given by quasi- and fully-secular equations, respectively. Even more importantly, we show that at each time the cumulant equation and the FA equation are closer in the trace distance to the exact solution than any secular MME. 
\begin{figure}[ht]
\includegraphics[width=1 \columnwidth]{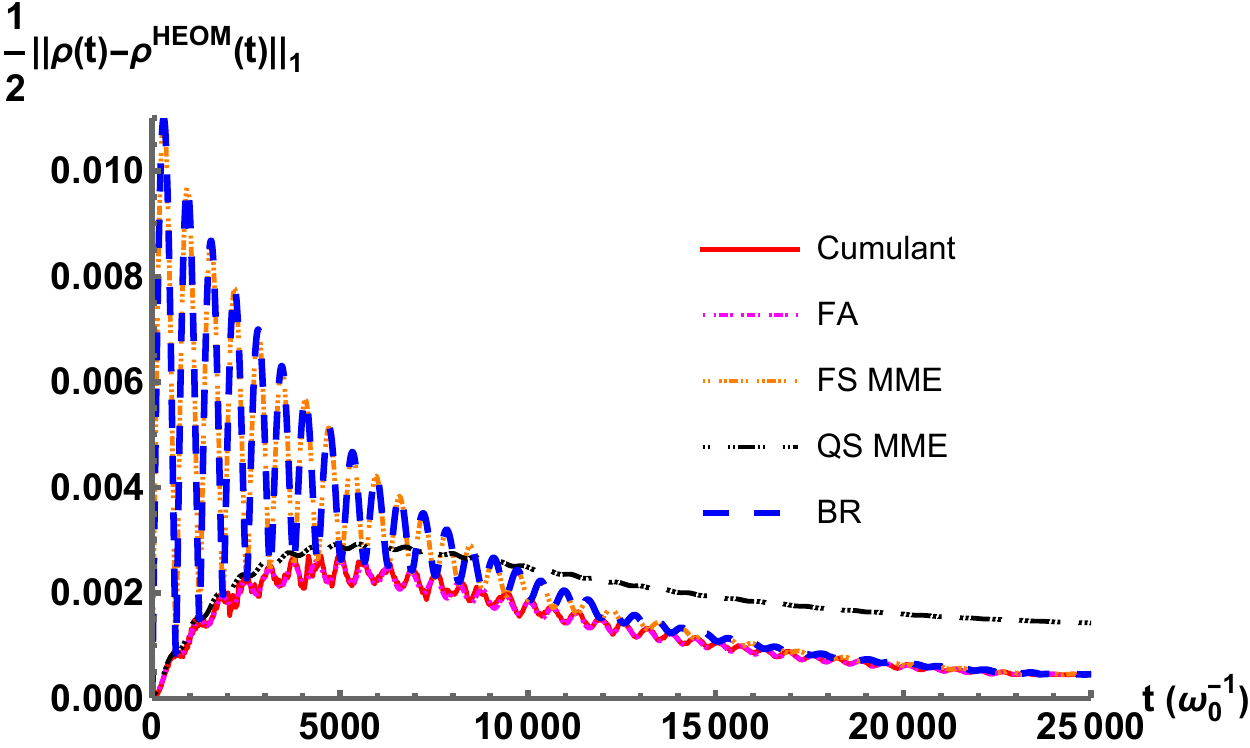}
\caption{ The trace distance comparison between dynamics obtained with various dynamical equation and numerically exact method of HEOM. The parameters and initial conditions of the simulation are the same as in Fig.~\ref{fig:3lvlV1}. The dash-dotted pink curve corresponds to the FA equation, the red solid curve corresponds to the cumulant equation, the short-dash-dotted orange line corresponds to the fully secular master equation, The dashed blue line corresponds to the BR equation and finally the long-dash-dotted black curve corresponds to the quasi secular master equation. } 
\label{fig:3lvlheom}
\end{figure}




\begin{figure*}[ht]
\includegraphics[width=1 \columnwidth]{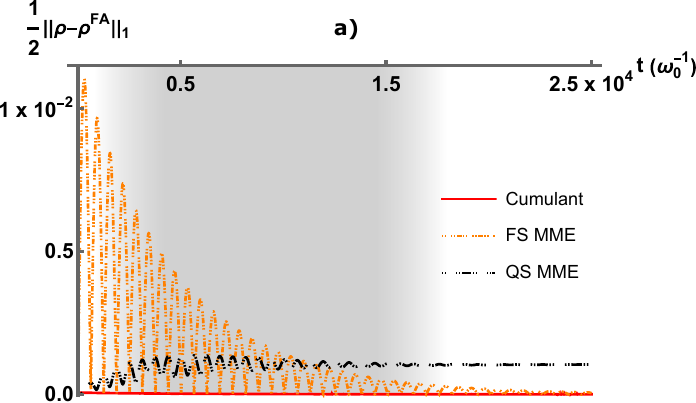}
\includegraphics[width=1 \columnwidth]{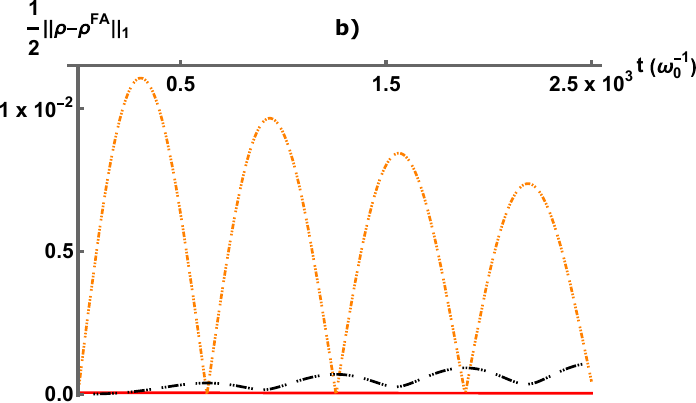}
\includegraphics[width=1 \columnwidth]{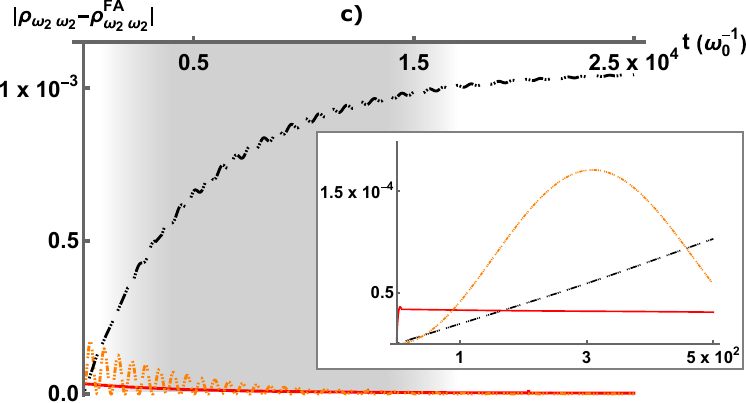}
\includegraphics[width=1 \columnwidth]{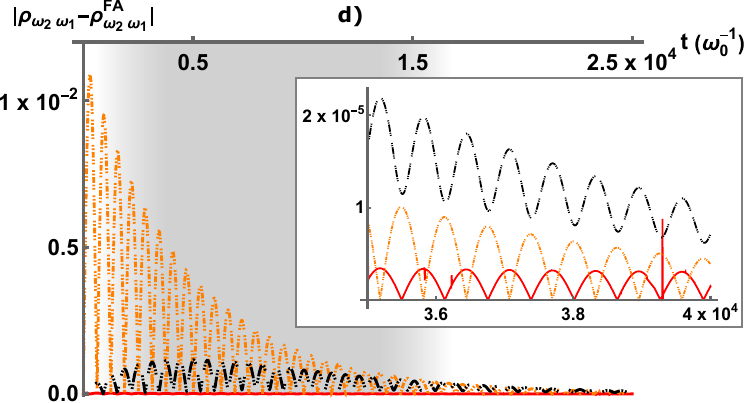}
\caption{ Qutrit-Boson model -- alternative (w.r.t. Fig.~\ref{fig:3lvlV1}) comparison between dynamics predicted by the fully-secular Markovian master equation, quasi-secular Markovian master equation and FA equation. The parameters and initial conditions of the simulation are the same as in Fig.~\ref{fig:3lvlV1}. Plots depicts difference between the FA equation and other equations regarding a) trace distance between solutions, b) also trace distance, but shorter times, c) absolute value of difference in populations of the highest energy level, and d) modulus of difference in coherence between excited states. 
The short-dash-dotted orange curves correspond to the fully-secular MME, the long-dash-dotted black curves correspond to the quasi-secular MME, and the solid red curves correspond to the cumulant equation.}
\label{fig:3lvlV2}
\end{figure*}

The qutrit-boson model concerns a three-level system linearly coupled to a bath of harmonic oscillators. The Hamiltonians of the model take the form:
\begin{align}
    &H_S = (\omega_0+\frac{\Delta \omega}{2}) \ket{\omega_2}\bra{\omega_2}+(\omega_0-\frac{\Delta \omega}{2}) \ket{\omega_1}\bra{\omega_1},\\
    &H_B = \sum_k \omega_k a_k^\dagger a_k,\\
    &H_\mathcal{I} = \left( \ket{g}\bra{\omega_2}+\ket{g}\bra{\omega_1}+\text{h.c.}\right) \otimes \sum_k g_k \left( a_k + a_k^\dagger\right),
\end{align}
where $\omega_0$ is the arithmetic average of transition frequencies between ground and excited states, 
namely $\omega_1 = \omega_0 - \frac{\Delta \omega}{2}$ and $\omega_2 = \omega_0 + \frac{\Delta \omega}{2}$. The continuous limit is conducted in standard way to give $J(\Omega) = \alpha \Omega \times e^{- \frac{\abs{\Omega}}{\omega_c}}$.

For this configuration the fully-secular master equation for the time evolution
(in the Schr\"odinger picture)  of $\rho^{\mathrm{fs}}(t)$ 
depends on the jumps operator relative to all the four
possible transition frequency in the system and it has the form: 
\begin{equation}
\begin{split}
&\dot{\rho}_S^\mathrm{fs}(t) = - i \left[ H_S, \rho \right] \\
&+ \sum_{\omega = \pm \omega_1, \, \pm \omega_2}  \gamma(\omega) \left ( A(\omega)  \rho A^\dagger (\omega)  - \frac12 \{ A^\dagger(\omega)  A(\omega) , \rho\} \right ). 
\end{split}
\end{equation} 
where the jump operators for the fully-secular master equation are
$ A(\omega_1) =  A^\dagger(-\omega_1) = \ketbra{g}{\omega_1}$ and $ A(\omega_2) = A^\dagger(-\omega_2) =  \ketbra{g}{\omega_2}$.
In contrast, the quasi-secular master equation takes into account the situation in which the two excited levels can be considered degenerate. Therefore, the frequency $\Delta\omega$ defining the short time scale enters at the level of the coherent evolution, while the dissipative part depends only on the average frequency $\omega_0$. 
Thus, the evolution of the reduced density matrix of the system in the quasi-secular approach $\rho^{\mathrm{qs}}(t)$
\begin{equation}
\begin{split}
    &\dot{\rho}_S^\mathrm{qs}(t) = - i \left[ H_S, \rho \right]  + \sum_{\omega = \pm \omega_0} \gamma(\omega) \left ( A^\mathrm{qs}(\omega)  \rho {A^\mathrm{qs}}^\dagger (\omega) \right.\\
    &\left.- \frac12 \{ {A^\mathrm{qs}}^\dagger(\omega)  A^\mathrm{qs}(\omega) , \rho\} \right ), 
\end{split}
\end{equation}
and the quasi-secular jump operators are $ A^\mathrm{qs}(\omega_0) = A^{\mathrm{qs}^\dagger}(-\omega_0)=\ketbra{g}{\omega_1} + \ketbra{g}{\omega_2})$ and it is clear that $A^\mathrm{qs}(\omega) \neq A(\omega)$ 
are jump operators of degenerated levels (see Section \ref{app:sec:Relax} of Appendix).

It is easy to notice that having a unique jump operator guiding the absorption 
from the ground state to both the excited states and one for emission, induces a coupling evolution 
of the population of these two levels and the consequent generation of coherence between them. 
This coherence cannot be observed only by looking at the fully-secular master equation in which evolution of coherence is decoupled from the evolution of populations. Instead, fully-secular MME is much better at later times for which quasi-secular MME fails by construction. Sufficiently long lasting interaction between an open system and reservoir allows to resolve arbitrarily closely placed energy levels. After that quasi-secular description can no longer hold.

In Fig.~\ref{fig:3lvlV1} we compare the evolution of the population ${\rho}_{\omega_{2}\omega_{2}}$ of highest energy level
and the coherence ${\rho}_{\omega_{2}\omega_{1}}$ between the two excited levels given by quasi- and fully-secular master equations, the cumulant equation, and derived here FA equation. The initial state of the open system is the ground state. As expected, the fully-secular master equation, due to its structure in which coherences and population evolve in a decoupled way, does not predict creation of coherence in the transient period (see Fig.~\ref{fig:3lvlV1} b)). On the contrary, the cumulant equation, the FA equation, and the quasi-secular master equation, capture the short-time behavior and process of establishing coherence. The possibility to catch the short-time behavior by the quasi-secular master equation was noticed in Ref.~\cite{Brumer_2018}. 
Therein the fully-secular master equation was compared to the quasi-secular master equation for a particular physical system, that is easy to be experimentally studied. As a consequence of the weak coupling the predictions for the evolution of populations due to different dynamical equation differ only marginally regarding numerical values (see Fig.~\ref{fig:3lvlV1} a)). However, we observe (see inset in Fig.~\ref{fig:3lvlV1} a)) that for the long times the quasi-secular MME gives different, inaccurate prediction for the magnitude of population than other equations, as expected. This is because of degenerate Bohr spectrum used to construct the quasi-secular generator as discussed above. 

In order to check the validity of the investigated herein equations in Fig.~\ref{fig:3lvlheom} we compare them with the numerically exact method of HEOM. More precisely, we plot curves of evolution of the trace distance between HEOM and quasi-secular MME, fully-secular MME, the cumulant equation, and the FA equation. Low numerical values of the trace distance in the plot in Fig.~\ref{fig:3lvlheom} indicate good proximity of each considered dynamical equation to the exact solution. Furthermore, it is evident that for short times quasi-secular equation performs better that the fully-secular one, whereas for longer times the table turns and the fully-secular equation outperforms the quasi-secular one. As expected the cumulant equation and the FA equation outperform both quasi- and fully-secular MME's, and Bloch-Redfield equation for all times, with cumulant equation being slightly closer to the exact solution that the FA equation. The long times difference between all considered equations and HEOM has the same explanation as in the case of the spin-boson model, see Remark~\ref{rem:MF_Ren}.


As it was discussed in Section \ref{sec:Interpolation}, the FA equation as well as the cumulant equation, should interpolate between dynamics provided by the quasi- and fully-secular master equation. For short times they should reproduce the dynamics provided by the quasi-secular MME, and for longer times their proximity to the solution provided by the fully-secular MME shall be observed. In Fig.~\ref{fig:3lvlV2} we provide a detailed comparison between dynamics provided by the FA equation and the cumulant equation, quasi-secular MME, and fully-secular MME. We plot the time evolution of i) trace distance between different solutions  for entire evolution (Fig.~\ref{fig:3lvlV2}a) and for short times (Fig.~\ref{fig:3lvlV2}b); ii) absolute value of difference in exited state population, with inset that focuses on short times; iii) modulus of difference of coherence with inset showing long times behavior. First, we observe  that (as expected) for short-times at each plot the solution due to the quasi-secular MME differs less from the dynamics due to the FA equation than than the solution due to the fully-secular MME. Consistently with the prediction, at later times the solution due to the fully-secular MME is closer to the dynamics predicted by FA than the quasi-secular MME. This behavior exemplifies how the FA equation interpolates between quasi- and fully-secular dynamics. Secondly, the solution obtained with the cumulant equation is always (except inset in Fig.~\ref{fig:3lvlV2} c))  much closer to the solution obtained with FA than solutions by quasi- and fully-secular master equations. Therefore, by the triangle inequality for the trace distance argument, we conclude that the cumulant equation also has the interpolation property. Yet we emphasize that at each point of evolution solutions due to the FA and cumulant equations are closer to the exact dynamics than solutions due to considered master equations (see Fig.~\ref{fig:3lvlheom}). Therefore, the interpolation property of the FA and the cumulant equations is just a consequence of them being pretty accurate in full time span.


\section{Conclusions}

We have provided a new tool to study the dynamics of open systems by proposing a filtered approximation (FA) of the cumulant equation. Our approach reverberates into a ``ready-to-use'' approach similar to the Davies-GKSL one that allows 
to compute the dynamical evolution of an open quantum system just by the knowledge of several discrete 
points of the spectral density sampled at the transition frequency of the system.

The obtained evolution, i.e., FA equation, inherits the property of cumulant equation,  that,  
unlike any completely positive master equation with time-independent generator, it covers the full range of times. In particular, it interpolates between quasi- and fully-secular Markovian master equations. 
The dynamical equation (FA equation) we have derived is as easy to use as the Davies master equation but allows us to witness the non-Markovian dynamics in the intermediate time scale. Therefore, as we have exemplified, the FA equation combines low computational complexity with high accuracy at all times scales, at least in the weak coupling regime.

For larger systems, where we would need to further reduce the computational complexity of simulations, FA equation can be  simplified by applying partial secular approximation. 
This was originally proposed for Bloch-Redfield master equation in \cite{cattaneo2019local}. However, it can be applied in the same way to cumulant equation as well as to FA equation. 

The FA equation while being perfect tool to study single bath for not too strong coupling,
regarding two or more baths it may not be a good tool to obtain steady state. Indeed, as we showed, for long times it can be approximated to a good extent by Davies equation, while
it is known, that for two baths and two coupled systems between the baths, for small coupling, the quasi-secular master equation works much better \cite{hofer2017markovian,Gonzalez_2017}. However, we believe that FA will still work pretty well in the  transient regime. 
 
There are plenty of possible developments that can be carried out with our toolkit. One can, for instance, address how the dissipation of a compound system influences the emergence of collective phenomena, e.g.,  synchronization and dissipative phase transitions,  or how quantum correlations are built-in time. 
It can also pave the way for new techniques for quantum noise spectroscopy, with the dissipating system used as a probe to study the environmental degree of freedom. 

Moreover, it is an interesting question of how to extend our result to time-dependent Hamiltonians, 
for the problem of interest in control theory, to potentially address the role of non-Markovian feedback and in thermodynamic protocols like work extraction.  The time-dependent scenario would also be important for the description of noise in a quantum computer - the evolution there is clearly non-Markovian, and can also be considered to be a weak coupling because otherwise, we would not be able to keep quantum information in the quantum computer at all~\cite{Alicki-2002}. 

Our approach can also be employed to address many tantalizing research areas at the boundary between quantum physics, chemistry, and biology. 
The coherent transports of excitation in a noisy environment suggest unavoidable open quantum system descriptions that take into account all the times scales, 
and possible manipulation of such effects of paramount importance for life sciences can lead to the emerging of new technologies, like quantum-enhanced photovoltaic cells.


\section*{Acknowledgements}
This work is supported by  Foundation for Polish Science (FNP), IRAP project ICTQT, contract no. 2018/MAB/5, co-financed by EU  Smart Growth Operational Programme. AM is also supported by (Polish) National Science Center (NCN): MINIATURA  DEC-2020/04/X/ST2/01794. MW acknowledges grant PRELUDIUM-20 (grant number: 2021/41/N/ST2/01349) from the National Science Center.

\vspace{2.0em}
\bibliographystyle{apsrev4-2}
\bibliography{references}

\newpage
\begin{widetext}
~~\\
\appendix
\begin{center}
		{\Huge Appendix}
\end{center}
\addappheadtotoc

\section{Divergences}\label{app:sec:div}

In this Section, we discuss a problem of divergences appearing in the cumulant equation approach, although it was not present in Davies's weak coupling approach. Let the spin-boson model with transition frequency $\omega_0$ (see the main text), serve as our testing ground. We choose the reservoir to be a heat bath at temperature $T=0$, and spectral density $J(\Omega)=\alpha \Omega$. From this choice, we obtain $R(\Omega)=\alpha \Omega$, and:
\begin{align}
    &\gamma (\omega,\omega^\prime,t) = \alpha e^{i \frac{\omega^\prime-\omega}{2}t} \int_{0}^{\infty} d\Omega~\Omega
	    \left[t~ \mathrm{sinc} \left(\frac{\omega^\prime-\Omega}{2}t\right)\right] \left[t~ \mathrm{sinc} \left(\frac{\omega-\Omega}{2}t\right)\right].
\end{align}
The integral in the equation above is divergent. This fact constitutes a major problem regarding the description of the system's evolution. The simplest candidate for a countermeasure to this issue is to choose the cutoff profile and cutoff frequency. Let us choose then, two different cutoff profiles $(1)$ a ``sharp'' cutoff, and  $(2)$ a exponential one~\cite{Rivas_2017}, both with cutoff frequency $\omega_0 \ll \omega_c < + \infty$. 
\begin{align}\label{eqn:COreg}
    \gamma (\omega,\omega^\prime,t) \stackrel{\text{''sharp''}}{\longrightarrow} \gamma^{(1)} (\omega,\omega^\prime,t) &= \alpha e^{i \frac{\omega^\prime-\omega}{2}t} \int_{0}^{\omega_c} d\Omega~\Omega
	    \left[t~ \mathrm{sinc} \left(\frac{\omega^\prime-\Omega}{2}t\right)\right] \left[t~ \mathrm{sinc} \left(\frac{\omega-\Omega}{2}t\right)\right] \stackrel{t \gg \frac{1}{\omega_0}}{\approx} \delta_{\omega,\omega^\prime} \delta_{\omega,\omega_0} 2\pi t \alpha \omega_0, \\
	 \gamma (\omega,\omega^\prime,t) \stackrel{\text{''exp''}}{\longrightarrow} \gamma^{(2)} (\omega,\omega^\prime,t) &= \alpha e^{i \frac{\omega^\prime-\omega}{2}t} \int_{0}^{+\infty} d\Omega~\Omega e^{-\frac{\Omega}{\omega_c}}
	    \left[t~ \mathrm{sinc} \left(\frac{\omega^\prime-\Omega}{2}t\right)\right] \left[t~ \mathrm{sinc} \left(\frac{\omega-\Omega}{2}t\right)\right] \nonumber\\
     &\stackrel{t \gg \frac{1}{\omega_0}}{\approx} \delta_{\omega,\omega^\prime} \delta_{\omega,\omega_0} 2\pi t \alpha \omega_0 e^{-\frac{\omega_0}{\omega_c}},\label{eqn:COreg2}
\end{align}
As we obser in equations \eqref{eqn:COreg}, \eqref{eqn:COreg2} above, the elements of the Kossakowski matrix, and hence the dynamics, depend on the choosen cutoff function, and thence spectral density profile.

\section{ Approximations of cumulant equation}
\label{app:Regulazitations_details}

In this Section, we derive the main result of our paper, i.e., the filtered approximation (FA) equation. The FA equation is constructed by replacing the Kossakowski matrix of the cumulant equation $\gamma_{ij} (\omega,\omega^\prime,t)$ with its positive semi-definite approximation $\gamma^
\star_{ij} (\omega,\omega^\prime,t)$. In parallel, we also develop another kind of approximation to the cumulant equation, i.e., $\star\star$-approximation, not mentioned in the main text. We base our motivation on an observation that in the weak coupling limit, the famous Davies-GKSL equation (usually) provides a good approximation to the true dynamics, independently of the cutoff profile of the spectral density function. This situation motivates us to develop an approximation procedure for the cumulant equation, yielding dynamical equations that depend on a smaller number of parameters of the system yet preserving the non-Markovianity and CPTP properties of the cumulant equation. In this way, we derive the FA equation, which is almost as simple as the Davies-GKSL equation, and also $\star\star$-approximation, in which only the ''spontaneous processes part'' of the Kossakowski matrix is subjected to modifications.

\subsection{Filtered approximation (FA)}\label{app:sec:SM}

In order to derive our FA equation we start with the formulas for time-dependent relaxation coefficients $\gamma_{ij} (\omega,\omega^\prime,t)$ in the cumulant equation, 
\begin{align}\label{def:FTofB}
	&\gamma_{ij} (\omega,\omega^\prime,t) = \int_0^t ds \int_0^t dw~ e^{i (\omega^\prime s - \omega w)} \left< \tilde{B}_j (s) \tilde{B}_i (w) \right>_{\tilde{\rho}_B},
\end{align}
which can be written in terms of Fourier transform of interaction picture reservoir's operators:
\begin{align}
	B_i \left(\Omega \right) = \frac{1}{2 \pi} \int_{-\infty}^{\infty} du~ e^{i \Omega u} \tilde{B}_i (u),
\end{align}
then using the formula above and the fact that $\left< \tilde{B}_j (s) \tilde{B}_i (w) \right>_{\tilde{\rho}_B}=\left< \tilde{B}_j (s-w) B_i  \right>_{\tilde{\rho}_B}$ (for example for a thermal reservoir), we integrate over time variables, in order to obtain alternative formula for $\gamma_{ij} (\omega,\omega^\prime,t)$.
\begin{align}\label{eqn:gammaPreIntegrated}
    &\gamma_{ij} (\omega,\omega^\prime,t) = e^{i \frac{\omega^\prime-\omega}{2}t} \int_{-\infty}^{\infty} d\Omega~
	    \left[t~ \mathrm{sinc} \left(\frac{\omega^\prime-\Omega}{2}t\right)\right]  \left[t~ \mathrm{sinc} \left(\frac{\omega-\Omega}{2}t\right)\right]  R_{ji} (\Omega),
\end{align}
where $\mathrm{sinc}(x) \equiv \frac{\sin (x)}{x}$, and $R_{ji} (\Omega)=\left< {B}_j (\Omega) {B}_i \right>_{\tilde{\rho}_B}$. 

The next crucial ingredient in our approximation are two models of Dirac delta function of width $\frac{1}{\tau}$. 
\begin{align}\label{eqn:models}
    \delta^{(1)}_\tau (x) = \frac{\sin \tau x}{\pi x},~~\delta^{(2)}_\tau (x) = \frac{\sin^2 \tau x}{\pi x^2 \tau}.
\end{align}
These models converge to Dirac delta distribution for $t\to +\infty$, provided the rest of the integrand behaves ``well'' (is bounded and does not oscillate too fast) or the range of integration is finite. Here, we assume that the (sharp) cutoff frequency $\max\{\abs{\omega},\abs{\omega^\prime}\}  \le \omega_c \ll \Lambda$, for which $J(\Lambda) \approx 0$ exists but can be very large ($\omega_c$ refers to cutoff profile as in Section \ref{sec:NoteJ}).  At this point, we are ready to perform a series of approximations:
\begin{align}
    &\gamma_{ij} (\omega,\omega^\prime,t)  
    = e^{i \frac{\omega^\prime-\omega}{2}t} \int_{-\infty}^{\infty} d\Omega~
	    \left[t~ \mathrm{sinc} \left(\frac{\omega^\prime-\Omega}{2}t\right)\right]  \left[t~ \mathrm{sinc} \left(\frac{\omega-\Omega}{2}t\right)\right]  R_{ji} (\Omega)\\
	    &\stackrel{(I)}{=}e^{i \frac{\omega^\prime-\omega}{2}t} \sum_k\int_{-\infty}^{\infty} d\Omega~
	    \left[t~ \mathrm{sinc} \left(\frac{\omega^\prime-\Omega}{2}t\right)R^{\frac 12}_{jk} (\Omega)\right]  \left[t~ \mathrm{sinc} \left(\frac{\omega-\Omega}{2}t\right)R^{\frac 12}_{ki} (\Omega)\right]   \\
	    &\stackrel{(II)}{=} e^{i \frac{\omega^\prime-\omega}{2}t} \sum_k \lim_{\Lambda \to +\infty}\int_{-\Lambda}^{\Lambda} d\Omega~
	    \left[t~ \mathrm{sinc} \left(\frac{\omega^\prime-\Omega}{2}t\right)R^{\frac 12}_{jk} (\Omega)\right]  \left[t~ \mathrm{sinc} \left(\frac{\omega-\Omega}{2}t\right)R^{\frac 12}_{ki} (\Omega)\right] \\
	    &= e^{i \frac{\omega^\prime-\omega}{2}t} \sum_k \lim_{\Lambda \to +\infty}\int_{-\Lambda}^{\Lambda} d\Omega~
	    \left[2\pi \delta^{(1)}_{\frac t2} (\omega^\prime-\Omega)R^{\frac 12}_{jk} (\Omega)\right]  \left[2\pi \delta^{(1)}_{\frac t2} (\omega-\Omega)R^{\frac 12}_{ki} (\Omega)\right] \\
	   &\stackrel{(III)}{\approx}e^{i \frac{\omega^\prime-\omega}{2}t} \sum_k \lim_{\Lambda \to +\infty} \int_{-\Lambda}^{\Lambda} d\Omega~
	    \left[2\pi \delta^{(1)}_{\frac t2} (\omega^\prime-\Omega)R^{\frac 12}_{jk} (\omega^\prime)\right]  \left[2\pi \delta^{(1)}_{\frac t2} (\omega-\Omega)R^{\frac 12}_{ki} (\omega)\right] \\
	    &\stackrel{(IV)}{\approx} e^{i \frac{\omega^\prime-\omega}{2}t} \sum_k \lim_{\Lambda \to +\infty} \int_{-\Lambda}^{\Lambda} d\Omega~
	    \left[t~ \mathrm{sinc} \left(\frac{\omega^\prime-\Omega}{2}t\right)R^{\frac 12}_{jk} (\omega^\prime)\right]  \left[2\pi \delta(\Omega-\omega)R^{\frac 12}_{ki} (\omega)\right]   \\
	    &= 2\pi t e^{i \frac{\omega^\prime-\omega}{2}t} 
	     \mathrm{sinc} \left(\frac{\omega^\prime-\omega}{2}t\right) \sum_k R^{\frac 12}_{jk} (\omega^\prime) R^{\frac 12}_{ki} (\omega),
\end{align}
where $R^{\frac 12}_{ij} (\Omega)  \equiv \left(R^{\frac 12}(\Omega)\right)_{ij}$. In the first step $(I)$ we decomposed positive semi-definite matrix $R(\Omega)$ into it's square roots, in step $(II)$ we introduced cutoff frequency and limiting procedure, in the third step $(III)$ we used models of delta function (see equation (\ref{eqn:models})) to approximate each square bracket separately, then in $(IV)$ step we substitute one of models of delta function with the proper Dirac delta (it is not important which model is approximated). The last approximation works the better the smoother is spectral density function and the longer is the time $t$. The above calculations are consistent with $\omega=\omega^\prime$ case, that can be otherwise calculated using the second delta function model $\delta^{(2)}_\tau (x)$. Finally we obtain:
\begin{align}
    \gamma_{ij} (\omega,\omega^\prime,t) \approx \gamma_{ij}^\star (\omega,\omega^\prime,t) = 2\pi t e^{i \frac{\omega^\prime-\omega}{2}t}
	     \mathrm{sinc} \left(\frac{\omega^\prime-\omega}{2}t\right) \sum_k R^{\frac 12}_{jk} (\omega^\prime) R^{\frac 12}_{ki} (\omega),
\end{align}
here $\star$ denotes the derived FA approximation. It can be shown that $\gamma_{ij}^\star (\omega,\omega^\prime,t)$ is a positive semi-definite matrix, what assures completely positive and trace preserving dynamics (CPTP). Unfortunately, in general case $\gamma_{ij}^\star (\omega,\omega^\prime,t)$ lacks one of the properties (indices swapping) that exact $\gamma_{ij} (\omega,\omega^\prime,t)$ exhibits, i.e.  $\gamma_{ij} (\omega^\prime,\omega,t) =e^{i(\omega-\omega^\prime)t}\gamma_{ij} (\omega,\omega^\prime,t)$. However, in particular for the (still extremely important) case of reservoir being electromagnetic field we have $\left[R^{\frac 12} (\omega^\prime), R^{\frac 12} (\omega)\right]=0$, and the lost symmetry is recovered. 
We conclude this part with the following remark. 
\begin{remark}
The approximation derived in this Section can be especially useful in situations in which nothing more than $R(\omega)$ (for Bohr's frequencies) and the level structure of the open system is known. In fact, only the knowledge of (at most) $(\dim R)(\dim R-1)$ real numbers is required. These quantities, which can be easily determined in an experiment, can then be used to determine the system's behavior for intermediate and shorter times (that might be out of reach of spectroscopy). This can be contrasted with methods in which we know $R(\Omega)$ in the whole range, and we perform numerical integration.
\end{remark}

\subsection{$\star\star$-type approximation}
\label{app:sec:SM2}


In this Section we develop a different variant of the evolution, i.e, $\star\star$-approximation, that requires the knowledge of the full spectral density in a same way as the cumulant equation. This description is obtained via a more delicate treatment in which only the vacuum component of the superoperator (''spontaneous processes'') is subjected to modifications. Our motivation is an observation that FA equation, alike Davies-GKSL equation, is ''cutoff-stable''. Namely, the FA dynamics does not diverge with growing cutoff parameter. Therefore, in $\star\star$-approximation we remove the divergences by modifying solely the ''vacuum component'' of the integrals in the Kossakowski matrix.  
\begin{definition}[Of a decent reservoir]\label{def:decentR}
We call a reservoir a decent one iff in the stationary state
\begin{align}
    &R_{ji} (\Omega) = J_{ji}(\Omega) \left(N\left(T(\Omega),\Omega\right)+1\right),\\
    \forall_{i,j}~~&\lim_{\Omega \to +\infty} \Omega J_{ji}(\Omega) N\left(T(\Omega),\Omega\right) =0, \label{app:eqn:limit}
\end{align}
where, and $T(\Omega) \ge 0$ is a generalized (local) temperature such that $T(-\Omega)=T(\Omega)$, $J_{ji}(\Omega)$ is a spectral density such that $-J_{ij}(-\Omega)=J_{ji}(\Omega)$, and (in dimensionless units $\hbar=c=k_B=1$):
\begin{align}
    N\left(T(\Omega),\Omega\right) = \frac{1}{e^{\frac{\Omega}{T(\Omega)}}-1}.
\end{align}
\end{definition}
The condition in equation (\ref{app:eqn:limit}), is sufficient for convergence of integral within the dissipator associated with stimulated processes. Furthermore, at this point, we want to make the following remarks.
\begin{remark}\label{rem:sgnSD}
Importantly, $R_{ij}(\Omega)=\frac{1}{2\pi} \gamma_{ij}(\Omega) \in M_\mathbb{C}^{\mathcal{I} \times \mathcal{J}}$ is positive semi-definite matrix. Therefore, $J_{ij}(\Omega) \in M_\mathbb{C}^{\mathcal{I} \times \mathcal{J}}$ is a (hermitian) positive semi-definite matrix for $\Omega \ge 0$, and negative semi-definite matrix for $\Omega < 0$.
\end{remark}
\begin{remark}
One must be careful when using the above Definition in ``even cases'' , for example: $J(\Omega) \sim \Omega^2$. In this case we should modify initial choice to $J(\Omega)= \mathrm{sgn}(\Omega)\Omega^2$ whenever possible.
\end{remark}

We come with the following Corollary of Definition \ref{def:decentR}.
\begin{corollary} For a decent reservoir in a stationary state and $\Omega \ge 0$, we have the following relations:
    \begin{align}
    R_{ji} (\Omega)&=e^{\frac{\Omega}{T(\Omega)}}R_{ij}(-\Omega),\\
    R_{ji} (-\Omega) &= J_{ij}(\Omega) N(T(\Omega),\Omega).
    \end{align}
\end{corollary}

\begin{proof}
The proof is straightforward from Definition (\ref{def:decentR}).
\end{proof}

From the Definition (\ref{def:decentR}), and equation (\ref{eqn:gammaPreIntegrated}) combined, we obtain the following form of time-dependent relaxation coefficients $\gamma_{ij} (\omega,\omega^\prime,t)$ of the cumulant equation.

\begin{align}
    &\gamma_{ij} (\omega,\omega^\prime,t) = e^{i \frac{\omega^\prime-\omega}{2}t} \int_{0}^{\infty} d\Omega~  N\left(T(\Omega),\Omega\right) \left(
	    \left[t~ \mathrm{sinc} \left(\frac{\omega^\prime-\Omega}{2}t\right)\right]  \left[t~ \mathrm{sinc} \left(\frac{\omega-\Omega}{2}t\right)\right] J_{ji}(\Omega) \right. \nonumber \\
	   &\left. +
	    \left[t~ \mathrm{sinc} \left(\frac{\omega^\prime+\Omega}{2}t\right)\right]  \left[t~ \mathrm{sinc} \left(\frac{\omega+\Omega}{2}t\right)\right] J_{ij}(\Omega)
	    \right) +e^{i \frac{\omega^\prime-\omega}{2}t} \int_{0}^{\infty} d\Omega~ J_{ji}(\Omega)
	     \left[t~ \mathrm{sinc} \left(\frac{\omega^\prime-\Omega}{2}t\right)\right]\left[t~ \mathrm{sinc} \left(\frac{\omega-\Omega}{2}t\right)\right]
\end{align}
The first integral, in the equation above, is convergent due to the mentioned condition within the Definition (\ref{def:decentR}). The second integral (which diverges in the absence of cutoff profile) will be subjected to a modification. Similarly, as in Section (\ref{app:sec:SM}) we choose $\max\{\abs{\omega},\abs{\omega^\prime}\} \ll \Lambda <+\infty$.

\begin{align}
    &\int_{0}^{\infty} d\Omega~ J_{ji}(\Omega)
	     \left[t~ \mathrm{sinc} \left(\frac{\omega^\prime-\Omega}{2}t\right)\right]\left[t~ \mathrm{sinc} \left(\frac{\omega-\Omega}{2}t\right)\right] \\
	     &\stackrel{(I)}{=} \lim_{\Lambda \to +\infty}\int_{0}^{\Lambda} d\Omega~ J_{ji}(\Omega)
	     \left[t~ \mathrm{sinc} \left(\frac{\omega^\prime-\Omega}{2}t\right)\right]\left[t~ \mathrm{sinc} \left(\frac{\omega-\Omega}{2}t\right)\right]\\
	     &\stackrel{(II)}{=} \sum_k \lim_{\Lambda \to +\infty}\int_{0}^{\Lambda} d\Omega~ 
	     \left[t~ \mathrm{sinc} \left(\frac{\omega^\prime-\Omega}{2}t\right)J^{\frac 12}_{jk}(\Omega)\right]\left[t~ \mathrm{sinc} \left(\frac{\omega-\Omega}{2}t\right)J^{\frac 12}_{ki}(\Omega)\right]\\
	     &\stackrel{(III)}{\approx} \sum_k \lim_{\Lambda \to +\infty}\int_{0}^{\Lambda} d\Omega~ 
	     \left[t~ \mathrm{sinc} \left(\frac{\omega^\prime-\Omega}{2}t\right) \frac{ \left(\mathrm{sgn}(\omega^\prime) J(\omega^\prime) \right)^{\frac 12}_{jk}}{\mathrm{sgn}(\omega^\prime)}\right]\left[t~ \mathrm{sinc} \left(\frac{\omega-\Omega}{2}t\right)\frac{ \left(\mathrm{sgn}(\omega) J(\omega) \right)^{\frac 12}_{ki}}{\mathrm{sgn}(\omega)}\right]\\
	     &= \sum_k\frac{ \left(\mathrm{sgn}(\omega^\prime) J(\omega^\prime) \right)^{\frac 12}_{jk}}{\mathrm{sgn}(\omega^\prime)}\frac{ \left(\mathrm{sgn}(\omega) J(\omega) \right)^{\frac 12}_{ki}}{\mathrm{sgn}(\omega)} \lim_{\Lambda \to +\infty}\int_{0}^{\Lambda} d\Omega~ 
	     \left[t~ \mathrm{sinc} \left(\frac{\omega^\prime-\Omega}{2}t\right) \right]\left[t~ \mathrm{sinc} \left(\frac{\omega-\Omega}{2}t\right)\right]\\
	     &=2 \pi \sum_k\frac{ \left(\mathrm{sgn}(\omega^\prime) J(\omega^\prime) \right)^{\frac 12}_{jk}}{\mathrm{sgn}(\omega^\prime)}\frac{ \left(\mathrm{sgn}(\omega) J(\omega) \right)^{\frac 12}_{ki}}{\mathrm{sgn}(\omega)} \lim_{\Lambda \to +\infty}\int_{0}^{\Lambda} d\Omega~ 
	     \delta^{(1)}_{\frac t2} (\omega^\prime - \Omega)\left[t~ \mathrm{sinc} \left(\frac{\omega-\Omega}{2}t\right)\right]\\
	     &\stackrel{(IV)}{\approx}2 \pi \sum_k\frac{ \left(\mathrm{sgn}(\omega^\prime) J(\omega^\prime) \right)^{\frac 12}_{jk}}{\mathrm{sgn}(\omega^\prime)}\frac{ \left(\mathrm{sgn}(\omega) J(\omega) \right)^{\frac 12}_{ki}}{\mathrm{sgn}(\omega)} \lim_{\Lambda \to +\infty}\int_{0}^{\Lambda} d\Omega~ 
	     \delta (\omega^\prime - \Omega)\left[t~ \mathrm{sinc} \left(\frac{\omega-\Omega}{2}t\right)\right]\\
	     &= 2 \pi t H(\omega^\prime) \mathrm{sinc} \left(\frac{\omega-\omega^\prime}{2}t\right) \sum_k\frac{ \left(\mathrm{sgn}(\omega^\prime) J(\omega^\prime) \right)^{\frac 12}_{jk}}{\mathrm{sgn}(\omega^\prime)}\frac{ \left(\mathrm{sgn}(\omega) J(\omega) \right)^{\frac 12}_{ki}}{\mathrm{sgn}(\omega)} \\
	     &\stackrel{(V)}{\approx} 2 \pi t H(\omega)H(\omega^\prime)  \mathrm{sinc} \left(\frac{\omega-\omega^\prime}{2}t\right) \sum_k\frac{ \left(\mathrm{sgn}(\omega^\prime) J(\omega^\prime) \right)^{\frac 12}_{jk}}{\mathrm{sgn}(\omega^\prime)}\frac{ \left(\mathrm{sgn}(\omega) J(\omega) \right)^{\frac 12}_{ki}}{\mathrm{sgn}(\omega)} \\
	     &= 2 \pi t H(\omega)H(\omega^\prime)  \mathrm{sinc} \left(\frac{\omega-\omega^\prime}{2}t\right)  \sum_k J^{\frac 12}_{jk}(\omega^\prime) J^{\frac 12}_{ki}(\omega),
\end{align}
where $J^{\frac 12}_{ij}(\omega)\equiv\left(J^{\frac 12}(\omega)\right)_{ij}$, and $H(\omega)$ is a Heaviside step function.  In the first step $(I)$ we introduce the limiting procedure, in the second step $(II)$ we decompose positive semi-definite matrix $J(\Omega)$ into its square roots. In the third step $(III)$, similarly as in Section (\ref{app:sec:SM}), we employ the filtering property of models of delta function separately in each square bracket, sign function is added so that square root is well defined (see remark (\ref{rem:sgnSD})). It is worth noting, that there is no need to employ the $\mathrm{sign}$ function if one anticipates and considers that only $\omega,\omega^\prime>0$ case is non-zero from the very beginning. In the next step $(IV)$ we interchange model of delta function with Dirac delta distribution. In the fifth step $(V)$ we multiply the expression with additional Heaviside step function to retain positive semi-definiteness of final approximation, this is consistent with long times approximation due to presence of $\mathrm{sinc}\left(\frac{\omega-\omega^\prime}{2}t\right)$ function. In the last step we simplify the expression, utilizing Heaviside functions. Finally:
\begin{align}
    &\gamma_{ij} (\omega,\omega^\prime,t) \approx \gamma^{\star \star}_{ij} (\omega,\omega^\prime,t)  =  e^{i \frac{\omega^\prime-\omega}{2}t} \int_{0}^{\infty} d\Omega~  N\left(T(\Omega),\Omega\right) \left(
	    \left[t~ \mathrm{sinc} \left(\frac{\omega^\prime-\Omega}{2}t\right)\right]  \left[t~ \mathrm{sinc} \left(\frac{\omega-\Omega}{2}t\right)\right] J_{ji}(\Omega) \right. \nonumber \\
	   &\left. +
	    \left[t~ \mathrm{sinc} \left(\frac{\omega^\prime+\Omega}{2}t\right)\right]  \left[t~ \mathrm{sinc} \left(\frac{\omega+\Omega}{2}t\right)\right] J_{ij}(\Omega)
	    \right) +2\pi t H(\omega)H(\omega^\prime)  e^{i \frac{\omega^\prime-\omega}{2}t} 
	    \mathrm{sinc} \left(\frac{\omega^\prime-\omega}{2}t\right) \sum_k J^{\frac 12}_{jk}(\omega^\prime) J^{\frac 12}_{ki}(\omega). \label{eqn:DoubleStarFormula}
\end{align}

An interesting property of $\star\star$-approximation described above is the fact that not only $\gamma^{\star \star}_{ij} (\omega,\omega^\prime,t)$ is a positive semi-definite matrix but it also exhibits all properties (indices swapping) of exact $\gamma_{ij} (\omega,\omega^\prime,t)$ matrix.

For any two level system, for example the considered spin-boson model, $\gamma^{\star \star}_{ij} (\omega,\omega^\prime,t)$ simplifies \footnote{One can also ``full-secular'' approximation to the ``vacuum'' part by eliminating terms for which $\omega^\prime \neq \omega$.} to:
\begin{align}
    &\gamma^{\star \star}_{ij} (\omega,\omega^\prime,t)  = e^{i \frac{\omega^\prime-\omega}{2}t} \int_{0}^{\infty} d\Omega~  N\left(T(\Omega),\Omega\right) \left(
	    \left[t~ \mathrm{sinc} \left(\frac{\omega^\prime-\Omega}{2}t\right)\right]  \left[t~ \mathrm{sinc} \left(\frac{\omega-\Omega}{2}t\right)\right] J_{ji}(\Omega) \right. \nonumber \\
	   &\left. +
	    \left[t~ \mathrm{sinc} \left(\frac{\omega^\prime+\Omega}{2}t\right)\right]  \left[t~ \mathrm{sinc} \left(\frac{\omega+\Omega}{2}t\right)\right] J_{ij}(\Omega)
	    \right) +2\pi t H(\omega)H(\omega^\prime)J_{ji}(\omega). \label{eqn:DoubleStarQubit}
\end{align}

\begin{figure}[t!]
    \centering
    \includegraphics[width=0.49 \columnwidth]{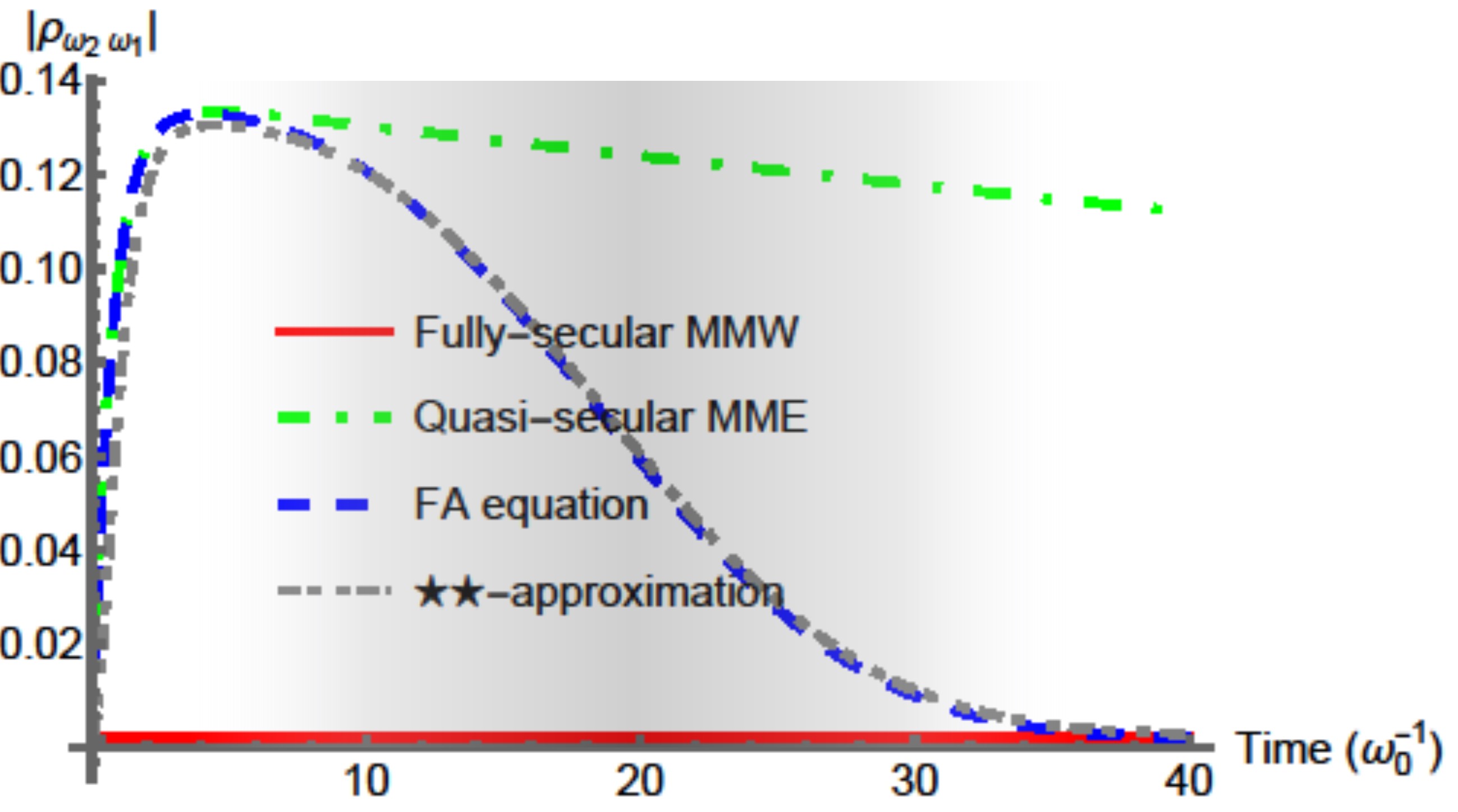}
     \includegraphics[width=0.49 \columnwidth]{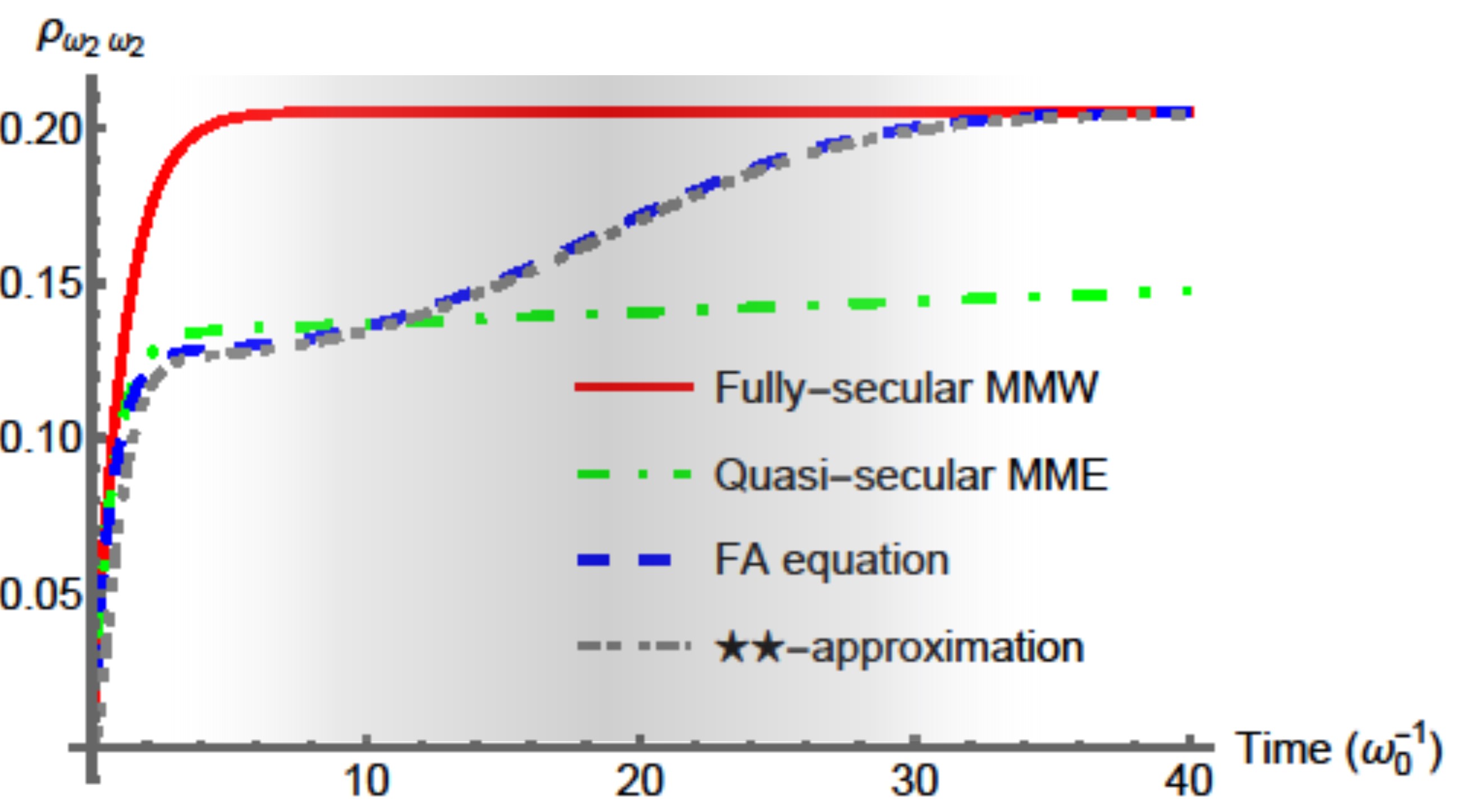}
    \caption{Qutrit-Boson model -- Interaction picture evolution of the coherence ${\rho_S}_{\omega_2 \omega_1}$ between the two excited levels (top), and population ${\rho_S}_{\omega_2 \omega_2}$ of the second excited $\omega_2$ (bottom) 
    of the reduced density matrix of the system $\rho_S$. The initial state of the system is $\rho_S(0)= \ketbra{g}{g} $.
    The reservoir is a heat bath at temperature $T_\mathrm{eff}=1$ in units of $\omega_0$. 
    The solid red curve is the evolution computed via the fully-secular master equation; the green dash-dotted curve is the evolution according to the quasi-secular master equation.  
    The dashed blue line is the evolution computed with the FA equation \eqref{eqn:DynMapStar}, and the grey short-dashed curve is the evolution given by the $\star\star$-approximation to the cumulant equation.
    The shaded grey zone indicates the region of intermediate times. The splitting parameter is $\Delta\omega = 2\pi\omega_0 \times 10^{-2}$, and $\alpha=0.05$. }
    \label{fig:Ev_qutrit2}
\end{figure} 

In Fig.~\ref{fig:Ev_qutrit2}, we observe that in the case of the qutrit-boson system, the dynamics provided by $\star\star$-approximation is close to the dynamics given by the FA equation\footnote{Parameters employed in  Fig.~\ref{fig:Ev_qutrit2} are beyond the weak-coupling regime, obtained curves might significantly differ from the exact solution.}. Indeed, $\star\star$-approximation interpolates between the quasi- and fully-secular Markovian master equations in the same way as the FA equation. As the spectral density function $J(\Omega)=\alpha \Omega$ has no explicit cutoff, the $\star\star$-approximation can be considered to be ''cutoff stable'' approximation to the refined weak coupling limit. However, at low temperatures, the evolution given by the $\star\star$-approximation is expected to be Markovian since, at these conditions, solely the Markovian ''vacuum'' component persists.  

We conclude this Section with the following remark.
\begin{remark}\label{app:rem:existence}
The derived approximations are ``cutoff stable'', i.e., the equations we propose are sensitive to the spectral density (and cutoff) profile but integrals in the Kossakowski matrix do not diverge if one considers the limit $\omega_c \to + \infty$. This modification was only possible by considering that the cutoff frequency exists but might be arbitrary large. Therefore, we are not allowed to think there is no cutoff at all. This kind of reasoning is consistent with Davies's weak coupling approach as finite lifetime of (two point) reservoir's correlations functions yields frequency cutoff for its spectral density.
\end{remark}

\begin{remark}
At (global) temperature $T=0$, both in FA equation and in $\star\star$-approximation only the spontaneous emission persists. Therefore, in the case of a two-level system at these conditions, the dynamics of FA equation and $\star\star$-approximation are equivalent to dynamics described by the (renormalized) Davies-GKSL equation.
\end{remark}

\subsection{Completely positive and trace preserving dynamics (CPTP) of the  approximations of the cumulant equation}

The CPTP property follows from its GKSL structure and positive semi-definiteness of the $\gamma_{ij}(\omega,\omega^\prime,t)$. The latter property was shown in detail in~\cite{RenormalizationPaper}. In this Section we prove the positive semi-definiteness of $\gamma^{\star}_{ij}(\omega,\omega^\prime,t)$ and $\gamma^{\star\star}_{ij}(\omega,\omega^\prime,t)$ matrices. The proofs are based on the observation that $R_{ij}(\omega) = \frac{1}{2 pi]} \gamma_{ij} (\omega)$, where $\gamma_{ij} (\omega)$ is known to be positive semi-definite matrix~\cite{Breuer+2006,Rivas_2012}.

\subsubsection{Positive semi-definiteness of $\gamma^{\star}_{ij}(\omega,\omega^\prime,t)$ matrix} 

Let us notice that $\Big(\gamma^{\star}_{ij}(\omega,\omega^\prime,t)\Big)$ is a Hadamard product (we denote with ''$\circ$''), of two matrices (multiplied with scalar).
\begin{align}
    \gamma_{ij}^\star (\omega,\omega^\prime,t)&= 2\pi t e^{i \frac{\omega^\prime-\omega}{2}t}
	     \mathrm{sinc} \left(\frac{\omega^\prime-\omega}{2}t\right) \sum_k R^{\frac 12}_{jk} (\omega^\prime) R^{\frac 12}_{ki} (\omega) \\
	     &=2\pi t \left( \Big( e^{i \frac{\omega^\prime-\omega}{2}t}
	     \mathrm{sinc} \left(\frac{\omega^\prime-\omega}{2}t\right)\Big) \circ \Big(\sum_k R^{\frac 12}_{jk} (\omega^\prime) R^{\frac 12}_{ki} (\omega)\Big) \right)_{(\omega,i),(\omega^\prime,j)},
\end{align}
where indices $i$ and $j$ are dummy for the matrix $\Big( e^{i \frac{\omega^\prime-\omega}{2}t}
	     \mathrm{sinc} \left(\frac{\omega^\prime-\omega}{2}t\right)\Big)$.

It is straighforward to see that matrix  $\Big(\sum_k R^{\frac 12}_{jk} (\omega^\prime) R^{\frac 12}_{ki} (\omega)\Big)$ is positive semi-definite 
\begin{align}
    \forall_g ~~&\sum_{\omega, \omega^\prime} \sum_{ij} g_i^*(\omega) g_j(\omega^\prime) \sum_k R^{\frac 12}_{jk} (\omega^\prime) R^{\frac 12}_{ki} (\omega) =   \sum_k \left(\sum_{\omega^\prime,j} g_j(\omega^\prime)  R^{\frac 12}_{jk} (\omega^\prime)\right) \left(\sum_{\omega,i} g_i^*(\omega)R^{\frac 12}_{ki} (\omega)\right) \\
    &=\sum_k \left(\sum_{\omega^\prime,j} g_j(\omega^\prime)  R^{\frac 12}_{jk} (\omega^\prime)\right) \left(\sum_{\omega,i} g_i (\omega)R^{\frac 12}_{ik} (\omega)\right)^*
    =\sum_k \abs{\sum_{\omega,i} g_i(\omega)  R^{\frac 12}_{ik} (\omega)}^2 \ge 0
\end{align}

To show that the matrix $\Big( e^{i \frac{\omega^\prime-\omega}{2}t}  \mathrm{sinc} \left(\frac{\omega^\prime-\omega}{2}t\right)\Big)$ is positive semi-definite we observe that 
\begin{align}
    \mathrm{sinc}(at) =\mathcal{F}\left(\mathcal{F}^{-1}\left(\mathrm{sinc}(at) \right)\right)= \frac{1}{\sqrt{2\pi}} \int_{-\infty}^{+\infty}ds~ e^{ias} \frac{\sqrt{\frac{\pi}{2}}\left(\mathrm{sgn}(t-s)+\mathrm{sgn}(t-s)\right) }{2t}, \label{eqn:FourierSinc}
\end{align}
where $\mathcal{F}$ and $\mathcal{F}^{-1}$ are Fourier transform and inverse Fourier transform respectively ($t \ge 0$), and $\mathrm{sgn}$ is the sign function.

Therefore,
\begin{align}
    \forall_g~~ &\sum_{\omega, \omega^\prime} \sum_{ij} g_i^*(\omega) g_j(\omega^\prime) e^{i \frac{\omega^\prime-\omega}{2}t}  \mathrm{sinc} \left(\frac{\omega^\prime-\omega}{2}t\right) \\
    &= \sum_{\omega, \omega^\prime} \sum_{ij} g_i^*(\omega) g_j(\omega^\prime) e^{i \frac{\omega^\prime-\omega}{2}t} \frac{1}{\sqrt{2\pi}} \int_{-\infty}^{+\infty}ds~ e^{i\frac{\omega^\prime-\omega}{2}s} \frac{\sqrt{\frac{\pi}{2}}\left(\mathrm{sgn}(t-s)+\mathrm{sgn}(t-s)\right) }{2t} \\
    &=\frac{1}{2}\int_{-\infty}^{+\infty}ds~ \sum_{\omega, \omega^\prime} \sum_{ij} g_i^*(\omega) g_j(\omega^\prime) e^{i \frac{\omega^\prime}{2}(t-s)} e^{-i \frac{\omega}{2}(t-s)}  \frac{\left(\mathrm{sgn}(t-s)+\mathrm{sgn}(t-s)\right) }{2t} \\
    &=\frac{1}{2}\int_{-\infty}^{+\infty}ds~      \left(\sum_{\omega,i} g_i^*(\omega)e^{-i \frac{\omega}{2}(t-s)} \sqrt{\frac{\left(\mathrm{sgn}(t-s)+\mathrm{sgn}(t-s)\right) }{2t}}\right)\left(\sum_{\omega^\prime,j}  g_j(\omega^\prime)e^{i \frac{\omega^\prime}{2}(t-s)}\sqrt{\frac{\left(\mathrm{sgn}(t-s)+\mathrm{sgn}(t-s)\right) }{2t}} \right) \\
    &=\frac{1}{2}\int_{-\infty}^{+\infty}ds~      \left(\sum_{\omega,i} g_i(\omega)e^{i \frac{\omega}{2}(t-s)} \sqrt{\frac{\left(\mathrm{sgn}(t-s)+\mathrm{sgn}(t-s)\right) }{2t}}\right)^*\left(\sum_{\omega^\prime,j}  g_j(\omega^\prime)e^{i \frac{\omega^\prime}{2}(t-s)}\sqrt{\frac{\left(\mathrm{sgn}(t-s)+\mathrm{sgn}(t-s)\right) }{2t}} \right) \\
    &=\frac{1}{2}\int_{-\infty}^{+\infty}ds~      \abs{\sum_{\omega,i} g_i(\omega)e^{i \frac{\omega}{2}(t-s)} \sqrt{\frac{\left(\mathrm{sgn}(t-s)+\mathrm{sgn}(t-s)\right) }{2t}}}^2 \ge 0.
\end{align}
Hence, matrix $\Big( e^{i \frac{\omega^\prime-\omega}{2}t}  \mathrm{sinc} \left(\frac{\omega^\prime-\omega}{2}t\right)\Big)$ is positive semi-definite.

Because, both matrices $\Big( e^{i \frac{\omega^\prime-\omega}{2}t}  \mathrm{sinc} \left(\frac{\omega^\prime-\omega}{2}t\right)\Big)$ and $\Big(\sum_k R^{\frac 12}_{jk} (\omega^\prime) R^{\frac 12}_{ki} (\omega)\Big)$ are positive semi-definite, the Hadamard product of these is also positive semi-definite. Now, because $\gamma^{\star}_{ij}(\omega,\omega^\prime,t)$ in question is proportional to the aforementioned Hadamard product, it is positive semi-definite matrix.

\subsubsection{Positive semi-definiteness of $\gamma^{\star\star}_{ij}(\omega,\omega^\prime,t)$ matrix}

To prove positive semi-definiteness of $\gamma^{\star\star}_{ij}(\omega,\omega^\prime,t)$ matrix let us firstly notice that
\begin{align}
    R_{ij}(-\Omega) &= J_{ij}(-\Omega) \left(N\left(T(-\Omega),-\Omega\right)+1\right) =
    J_{ij}(-\Omega) \left(\frac{1}{e^{-\frac{\Omega}{T(-\Omega)}}-1}+1\right) \\
    &=J_{ij}(-\Omega) \frac{e^{-\frac{\Omega}{T(-\Omega)}}}{e^{-\frac{\Omega}{T(-\Omega)}}-1} 
    = J_{ji}(\Omega) \frac{1}{e^{\frac{\Omega}{T(\Omega)}}-1} = J_{ji}(\Omega) N\left(T(\Omega),\Omega\right).
\end{align}
We now decompose $\gamma^{\star \star}_{ij} (\omega,\omega^\prime,t)$ in equation \eqref{eqn:DoubleStarFormula} into three components.
\begin{align}
    &\gamma^{\star \star}_{ij} (\omega,\omega^\prime,t)  =  e^{i \frac{\omega^\prime-\omega}{2}t} \int_{0}^{\infty} d\Omega~  N\left(T(\Omega),\Omega\right) \left(
	    \left[t~ \mathrm{sinc} \left(\frac{\omega^\prime-\Omega}{2}t\right)\right]  \left[t~ \mathrm{sinc} \left(\frac{\omega-\Omega}{2}t\right)\right] J_{ji}(\Omega) \right. \nonumber \\
	   &\left. +
	    \left[t~ \mathrm{sinc} \left(\frac{\omega^\prime+\Omega}{2}t\right)\right]  \left[t~ \mathrm{sinc} \left(\frac{\omega+\Omega}{2}t\right)\right] J_{ij}(\Omega)
	    \right) +2\pi t H(\omega)H(\omega^\prime)  e^{i \frac{\omega^\prime-\omega}{2}t} 
	    \mathrm{sinc} \left(\frac{\omega^\prime-\omega}{2}t\right) \sum_k J^{\frac 12}_{jk}(\omega^\prime) J^{\frac 12}_{ki}(\omega)\\
	    &=e^{i \frac{\omega^\prime-\omega}{2}t} \int_{0}^{\infty} d\Omega~  J_{ji}(\Omega)N\left(T(\Omega),\Omega\right) 
	    \left[t~ \mathrm{sinc} \left(\frac{\omega^\prime-\Omega}{2}t\right)\right]  \left[t~ \mathrm{sinc} \left(\frac{\omega-\Omega}{2}t\right)\right]\nonumber \\
	    &+e^{i \frac{\omega^\prime-\omega}{2}t} \int_{0}^{\infty} d\Omega~  J_{ij}(\Omega) N\left(T(\Omega),\Omega\right)  \left[t~ \mathrm{sinc} \left(\frac{\omega^\prime+\Omega}{2}t\right)\right]  \left[t~ \mathrm{sinc} \left(\frac{\omega+\Omega}{2}t\right)\right]\nonumber \\
	    &+2\pi t H(\omega)H(\omega^\prime)  e^{i \frac{\omega^\prime-\omega}{2}t} 
	    \mathrm{sinc} \left(\frac{\omega^\prime-\omega}{2}t\right) \sum_k J^{\frac 12}_{jk}(\omega^\prime) J^{\frac 12}_{ki}(\omega)\\
	    &=e^{i \frac{\omega^\prime-\omega}{2}t} \int_{0}^{\infty} d\Omega~  R_{ij}(-\Omega)
	    \left[t~ \mathrm{sinc} \left(\frac{\omega^\prime-\Omega}{2}t\right)\right]  \left[t~ \mathrm{sinc} \left(\frac{\omega-\Omega}{2}t\right)\right] \nonumber\\
	    &+e^{i \frac{\omega^\prime-\omega}{2}t} \int_{0}^{\infty} d\Omega~  R_{ji}(-\Omega)   \left[t~ \mathrm{sinc} \left(\frac{\omega^\prime+\Omega}{2}t\right)\right]  \left[t~ \mathrm{sinc} \left(\frac{\omega+\Omega}{2}t\right)\right]\nonumber \\
	    &+2\pi t H(\omega)H(\omega^\prime)  e^{i \frac{\omega^\prime-\omega}{2}t} 
	    \mathrm{sinc} \left(\frac{\omega^\prime-\omega}{2}t\right) \sum_k J^{\frac 12}_{jk}(\omega^\prime) J^{\frac 12}_{ki}(\omega).  \label{eqn:DoubleStarDecomposition}
\end{align}
Therefore, matrix $\Big(\gamma^{\star \star}_{ij} (\omega,\omega^\prime,t) \Big)$ is a sum of three matrices.
We now show that each of the term in equation \eqref{eqn:DoubleStarDecomposition} above is an element of a positive semi-definite matrix.

We start with the first component of equation \eqref{eqn:DoubleStarDecomposition}, that corresponds to matrix being the first component of $\Big(\gamma^{\star \star}_{ij} (\omega,\omega^\prime,t) \Big)$ matrix.
\begin{align}
    \forall_g ~~&\sum_{\omega,\omega^\prime} \sum_{ij} g^*_i(\omega) g_j(\omega^\prime) e^{i \frac{\omega^\prime-\omega}{2}t} \int_{0}^{\infty} d\Omega~  R_{ij}(-\Omega)
	    \left[t~ \mathrm{sinc} \left(\frac{\omega^\prime-\Omega}{2}t\right)\right]  \left[t~ \mathrm{sinc} \left(\frac{\omega-\Omega}{2}t\right)\right]\\
	    &= \sum_{\omega,\omega^\prime} \sum_{ij} g^*_i(\omega) g_j(\omega^\prime) e^{-i \frac{\omega}{2}t} e^{i \frac{\omega^\prime}{2}t} \int_{0}^{\infty} d\Omega~  \sum_k R^{\frac 12}_{ik}(-\Omega)R^{\frac 12}_{kj}(-\Omega)
	    \left[t~ \mathrm{sinc} \left(\frac{\omega^\prime-\Omega}{2}t\right)\right]  \left[t~ \mathrm{sinc} \left(\frac{\omega-\Omega}{2}t\right)\right]\\
	    &=     \sum_k \int_{0}^{\infty} d\Omega~  \left( \sum_{\omega^\prime,j}g_j(\omega^\prime)e^{i \frac{\omega^\prime}{2}t}  R^{\frac 12}_{kj}(-\Omega)
	    \left[t~ \mathrm{sinc} \left( \frac{\omega^\prime-\Omega}{2}t\right)\right] \right)  \left(\sum_{\omega,i}g^*_i(\omega)e^{-i \frac{\omega}{2}t}R^{\frac 12}_{ik}(-\Omega)\left[t~ \mathrm{sinc} \left(\frac{\omega-\Omega}{2}t\right)\right]\right)\\
	    &=     \sum_k \int_{0}^{\infty} d\Omega~  \left( \sum_{\omega^\prime,j}g_j(\omega^\prime)e^{i \frac{\omega^\prime}{2}t}  R^{\frac 12}_{kj}(-\Omega)
	    \left[t~ \mathrm{sinc} \left( \frac{\omega^\prime-\Omega}{2}t\right)\right] \right)  \left(\sum_{\omega,i}g_i(\omega)e^{i \frac{\omega}{2}t}R^{\frac 12}_{ki}(-\Omega)\left[t~ \mathrm{sinc} \left(\frac{\omega-\Omega}{2}t\right)\right]\right)^*\\
	    &=     \sum_k \int_{0}^{\infty} d\Omega~    \abs{\sum_{\omega,i}g_i(\omega)e^{i \frac{\omega}{2}t}R^{\frac 12}_{ki}(-\Omega)\left[t~ \mathrm{sinc} \left(\frac{\omega-\Omega}{2}t\right)\right]}^2 \ge 0.
\end{align}
The same method can be used to show the positive semi-definiteness of the second term.
\begin{align}
    \forall_g ~~&\sum_{\omega,\omega^\prime} \sum_{ij} g^*_i(\omega) g_j(\omega^\prime) e^{i \frac{\omega^\prime-\omega}{2}t} \int_{0}^{\infty} d\Omega~  R_{ji}(-\Omega)
	    \left[t~ \mathrm{sinc} \left(\frac{\omega^\prime+\Omega}{2}t\right)\right]  \left[t~ \mathrm{sinc} \left(\frac{\omega+\Omega}{2}t\right)\right]\\
	    &= \sum_{\omega,\omega^\prime} \sum_{ij} g^*_i(\omega) g_j(\omega^\prime) e^{-i \frac{\omega}{2}t} e^{i \frac{\omega^\prime}{2}t} \int_{0}^{\infty} d\Omega~  \sum_k R^{\frac 12}_{jk}(-\Omega)R^{\frac 12}_{ki}(-\Omega)
	    \left[t~ \mathrm{sinc} \left(\frac{\omega^\prime+\Omega}{2}t\right)\right]  \left[t~ \mathrm{sinc} \left(\frac{\omega+\Omega}{2}t\right)\right]\\
	    &=     \sum_k \int_{0}^{\infty} d\Omega~  \left( \sum_{\omega^\prime,j}g_j(\omega^\prime)e^{i \frac{\omega^\prime}{2}t}  R^{\frac 12}_{jk}(-\Omega)
	    \left[t~ \mathrm{sinc} \left( \frac{\omega^\prime+\Omega}{2}t\right)\right] \right)  \left(\sum_{\omega,i}g^*_i(\omega)e^{-i \frac{\omega}{2}t}R^{\frac 12}_{ki}(-\Omega)\left[t~ \mathrm{sinc} \left(\frac{\omega+\Omega}{2}t\right)\right]\right)\\
	    &=     \sum_k \int_{0}^{\infty} d\Omega~  \left( \sum_{\omega^\prime,j}g_j(\omega^\prime)e^{i \frac{\omega^\prime}{2}t}  R^{\frac 12}_{jk}(-\Omega)
	    \left[t~ \mathrm{sinc} \left( \frac{\omega^\prime+\Omega}{2}t\right)\right] \right)  \left(\sum_{\omega,i}g_i(\omega)e^{i \frac{\omega}{2}t}R^{\frac 12}_{ik}(-\Omega)\left[t~ \mathrm{sinc} \left(\frac{\omega+\Omega}{2}t\right)\right]\right)^*\\
	    &=     \sum_k \int_{0}^{\infty} d\Omega~    \abs{\sum_{\omega,i}g_i(\omega)e^{i \frac{\omega}{2}t}R^{\frac 12}_{ik}(-\Omega)\left[t~ \mathrm{sinc} \left(\frac{\omega+\Omega}{2}t\right)\right]}^2 \ge 0.
\end{align}

The positive semi-definiteness comes from formula \eqref{eqn:FourierSinc}.
\begin{align}
     \forall_g ~~&\sum_{\omega,\omega^\prime} \sum_{ij} g^*_i(\omega) g_j(\omega^\prime) 2\pi t H(\omega)H(\omega^\prime)  e^{i \frac{\omega^\prime-\omega}{2}t} 
	    \mathrm{sinc} \left(\frac{\omega^\prime-\omega}{2}t\right) \sum_k J^{\frac 12}_{jk}(\omega^\prime) J^{\frac 12}_{ki}(\omega) \\
	    &=\sum_{\omega,\omega^\prime} \sum_{ij} g^*_i(\omega) g_j(\omega^\prime) 2\pi t H(\omega)H(\omega^\prime)  e^{i \frac{\omega^\prime-\omega}{2}t} 
	     \frac{1}{\sqrt{2\pi}} \int_{-\infty}^{+\infty}ds~ e^{i\frac{\omega^\prime-\omega}{2}s} \frac{\sqrt{\frac{\pi}{2}}\left(\mathrm{sgn}(t-s)+\mathrm{sgn}(t-s)\right) }{2t}  \sum_k J^{\frac 12}_{jk}(\omega^\prime) J^{\frac 12}_{ki}(\omega) \\
	     &=\frac{\pi}{2} \sum_k \int_{-\infty}^{+\infty}ds~
	     \left(\sum_{\omega^\prime, j} g_j(\omega^\prime) H(\omega^\prime)  e^{i\frac{\omega^\prime}{2}(t+s)} \sqrt{\mathrm{sgn}(t-s)+\mathrm{sgn}(t-s) }   J^{\frac 12}_{jk}(\omega^\prime) \right)\\
	     &\times
	     \left(\sum_{\omega,i} g^*_i(\omega) H(\omega) e^{-i\frac{\omega}{2}(t+s)} \sqrt{\mathrm{sgn}(t-s)+\mathrm{sgn}(t-s)} J^{\frac 12}_{ki}(\omega) \right)\\
	     &=\frac{\pi}{2} \sum_k \int_{-\infty}^{+\infty}ds~
	     \left(\sum_{\omega^\prime, j} g_j(\omega^\prime) H(\omega^\prime)  e^{i\frac{\omega^\prime}{2}(t+s)} \sqrt{\mathrm{sgn}(t-s)+\mathrm{sgn}(t-s) }   J^{\frac 12}_{jk}(\omega^\prime) \right)\\
	     &\times
	     \left(\sum_{\omega,i} g_i(\omega) H(\omega) e^{i\frac{\omega}{2}(t+s)} \sqrt{\mathrm{sgn}(t-s)+\mathrm{sgn}(t-s)} J^{\frac 12}_{ik}(\omega) \right)\\
	      &=\frac{\pi}{2} \sum_k \int_{-\infty}^{+\infty}ds~
	     \abs{\sum_{\omega,i} g_i(\omega) H(\omega) e^{i\frac{\omega}{2}(t+s)} \sqrt{\mathrm{sgn}(t-s)+\mathrm{sgn}(t-s)} J^{\frac 12}_{ik}(\omega) }^2 \ge 0.
\end{align}
Therefore, because $\Big(\gamma^{\star \star}_{ij} (\omega,\omega^\prime,t) \Big)$ is a sum of three positive semi-definite matrices, it is positive semi-definite itself.


 \color{black}
\section{Interpolation between the quasi- and fully-secular master equations}\label{app:sec:Relax}

\subsection{Reproducing the fully-secular master equation}\label{app:sec:RelaxGlob}

The convergence to the fully-secular master equation can be performed directly from the cumulant equation (\ref{eqn:ActionOfCumulant}), the most easily with $\gamma_{ij} (\omega,\omega^\prime,t)$ in the form (\ref{eqn:gammaPreIntegrated}). It is enough to notice that for $\omega^\prime \neq \omega$, and $\forall_{\omega,\omega^\prime} ~\max\{\abs{\omega},\abs{\omega^\prime}\} \ll \omega_c <+\infty $: 
\begin{align}
    \frac{\gamma_{ij} (\omega,\omega^\prime,t)}{t} &= \frac{1}{t} e^{i \frac{\omega^\prime-\omega}{2}t} \int_{-\infty}^{\infty} d\Omega~
	    \left[t~ \mathrm{sinc} \left(\frac{\omega^\prime-\Omega}{2}t\right)\right]  \left[t~ \mathrm{sinc} \left(\frac{\omega-\Omega}{2}t\right)\right]  R_{ji} (\Omega)\\
	  &= \frac{2\pi}{t} e^{i \frac{\omega^\prime-\omega}{2}t}  \int_{-+\infty}^{+\infty} d\Omega~
	    \delta^{(1)}_{\frac t2} (\omega^\prime-\Omega)  \left[t~ \mathrm{sinc} \left(\frac{\omega-\Omega}{2}t\right)\right]  R_{ji} (\Omega)\\
	   &\stackrel{t \to \infty}{\approx} \frac{2\pi}{t} e^{i \frac{\omega^\prime-\omega}{2}t}  \int_{-+\infty}^{+\infty} d\Omega~
	    \delta_{\frac t2} (\omega^\prime-\Omega)  \left[t~ \mathrm{sinc} \left(\frac{\omega-\Omega}{2}t\right)\right]  R_{ji} (\Omega)\\
	    &=\ 2\pi e^{i \frac{\omega^\prime-\omega}{2}t}   \mathrm{sinc} \left(\frac{\omega-\omega^\prime}{2}t\right)  R_{ji} (\omega^\prime) 
	    \stackrel{t \to \infty}{\approx} 2\pi \delta_{\omega,\omega^\prime}  R_{ji} (\omega^\prime).
\end{align}
Where we have used the model of delta function $\delta^{(1)}_\tau$, see equation (\ref{eqn:models}). The result above holds true also in $\omega=\omega^\prime$ case. To check this it is enough to consider to use  $\delta^{(2)}_\tau$ from equation (\ref{eqn:models}) in the calculations.

Therefore, from the definition in equation (\ref{eqn:ActionOfCumulant}) we obtain the following.
\begin{align}
	\tilde{K}^{(2)}(t) \tilde{\rho}_S &=
	\sum_{i,j} \sum_{\omega, \omega^\prime} \gamma_{ij}(\omega,\omega^\prime,t)  \left(A_i (\omega) \tilde{\rho}_S A_j^\dagger (\omega^\prime) - \frac{1}{2} \left\{A_j^\dagger (\omega^\prime) A_i (\omega), \tilde{\rho}_S \right\} \right) \\
	& \stackrel{t \to \infty}{\approx} \sum_{i,j} \sum_{\omega, \omega^\prime}  2\pi t \delta_{\omega,\omega^\prime}  R_{ji} (\omega^\prime) \left(A_i (\omega) \tilde{\rho}_S A_j^\dagger (\omega^\prime) - \frac{1}{2} \left\{A_j^\dagger (\omega^\prime) A_i (\omega), \tilde{\rho}_S \right\} \right) \\
	& = \sum_{i,j} \sum_{\omega}  2\pi t  R_{ji} (\omega) \left(A_i (\omega) \tilde{\rho}_S A_j^\dagger (\omega) - \frac{1}{2} \left\{A_j^\dagger (\omega) A_i (\omega), \tilde{\rho}_S \right\} \right) \\
	& =t \sum_{i,j} \sum_{\omega}  \gamma_{ij} (\omega) \left(A_i (\omega) \tilde{\rho}_S A_j^\dagger (\omega) - \frac{1}{2} \left\{A_j^\dagger (\omega) A_i (\omega), \tilde{\rho}_S \right\} \right) = t \tilde{\mathcal{L}}^\mathrm{fs} \tilde{\rho}_S
\end{align}
where $\gamma_{ij} (\omega)=2\pi R_{ji} (\omega)$, is the {\it Markovian relaxation rate}. The $\tilde{\mathcal{L}}^\mathrm{fs}$ is the time-independent generator of the fully-secular master equation's dynamical map. Finally, we obtain:
\begin{align}
    e^{\tilde{K}^{(2)}(t)} \stackrel{t \to \infty}{\approx} e^{t \tilde{\mathcal{L}}^\mathrm{fs}},
\end{align}
where it is evident that $\tilde{\mathcal{L}}^\mathrm{fs}$ is a generator of a semigroup~\cite{AlickiLendi1987}, and the r.h.s. of the above is the fully-secular master equation's dynamical map.

A similar asymptotic result holds true for $\tilde{K}^{(2,\star)}(t)$, i.e, the FA equation.
\begin{align}
    \frac{\gamma^\star_{ij} (\omega,\omega^\prime,t)}{t} &= 2\pi  e^{i \frac{\omega^\prime-\omega}{2}t}
	     \mathrm{sinc} \left(\frac{\omega^\prime-\omega}{2}t\right) \sum_k R^{\frac 12}_{jk} (\omega^\prime) R^{\frac 12}_{ki} (\omega)\\
	     &\stackrel{t \to +\infty}{\approx} 2\pi \delta_{\omega,\omega^\prime} \sum_k R^{\frac 12}_{jk} (\omega^\prime) R^{\frac 12}_{ki} (\omega)= 2\pi \delta_{\omega,\omega^\prime}  R_{ji} (\omega).
\end{align}
In the approximation above we used properties of $\mathrm{sinc}$ function. Therefore, we have:
\begin{align}
	\tilde{K}^{(2,\star)}(t) \tilde{\rho}_S &=
	\sum_{i,j} \sum_{\omega, \omega^\prime} \gamma^\star_{ij}(\omega,\omega^\prime,t)  \left(A_i (\omega) \tilde{\rho}_S A_j^\dagger (\omega^\prime) - \frac{1}{2} \left\{A_j^\dagger (\omega^\prime) A_i (\omega), \tilde{\rho}_S \right\} \right)\\
	&\stackrel{t \to +\infty}{\approx}\sum_{i,j} \sum_{\omega, \omega^\prime} 2\pi t \delta_{\omega,\omega^\prime}  R_{ji} (\omega)  \left(A_i (\omega) \tilde{\rho}_S A_j^\dagger (\omega^\prime) - \frac{1}{2} \left\{A_j^\dagger (\omega^\prime) A_i (\omega), \tilde{\rho}_S \right\} \right)\\
	&= t\sum_{i,j} \sum_{\omega} 2\pi R_{ji} (\omega)  \left(A_i (\omega) \tilde{\rho}_S A_j^\dagger (\omega) - \frac{1}{2} \left\{A_j^\dagger (\omega) A_i (\omega), \tilde{\rho}_S \right\} \right)\\
	&= t\sum_{i,j} \sum_{\omega} \gamma_{ij} (\omega)  \left(A_i (\omega) \tilde{\rho}_S A_j^\dagger (\omega) - \frac{1}{2} \left\{A_j^\dagger (\omega) A_i (\omega), \tilde{\rho}_S \right\} \right)
	=t \tilde{\mathcal{L}}^\mathrm{fs} \tilde{\rho}_S.
\end{align}
In the full analogy, with the cases considered above, for $\tilde{K}^{(2,\star\star)}(t)$, i.e, the $\star\star$-approximation of the cumulant equation we have:
\begin{align}
    &\frac{\gamma^{\star\star}_{ij} (\omega,\omega^\prime,t)}{t} = e^{i \frac{\omega^\prime-\omega}{2}t} \frac{1}{t}\int_{0}^{+\infty} d\Omega~  N\left(T(\Omega),\Omega\right) \left(
	    \left[t~ \mathrm{sinc} \left(\frac{\omega^\prime-\Omega}{2}t\right)\right]  \left[t~ \mathrm{sinc} \left(\frac{\omega-\Omega}{2}t\right)\right] J_{ji}(\Omega) \right. \nonumber \\
	   &\left. +
	    \left[t~ \mathrm{sinc} \left(\frac{\omega^\prime+\Omega}{2}t\right)\right]  \left[t~ \mathrm{sinc} \left(\frac{\omega+\Omega}{2}t\right)\right] J_{ij}(\Omega)
	    \right) +2\pi  H(\omega)H(\omega^\prime)  e^{i \frac{\omega^\prime-\omega}{2}t} 
	    \mathrm{sinc} \left(\frac{\omega^\prime-\omega}{2}t\right) \sum_k J^{\frac 12}_{jk}(\omega^\prime) J^{\frac 12}_{ki}(\omega)\\
	    &=e^{i \frac{\omega^\prime-\omega}{2}t} 2 \pi \int_{0}^{+\infty} d\Omega~  N\left(T(\Omega),\Omega\right) \left(
	    \delta^{(1)}_{\frac t2}(\omega^\prime-\Omega)   \mathrm{sinc} \left(\frac{\omega-\Omega}{2}t\right) J_{ji}(\Omega) \right. \nonumber \\
	   &\left. +
	    \delta^{(1)}_{\frac t2}(\omega^\prime+\Omega)  \mathrm{sinc} \left(\frac{\omega+\Omega}{2}t\right) J_{ij}(\Omega)
	    \right) +2\pi  H(\omega)H(\omega^\prime)  e^{i \frac{\omega^\prime-\omega}{2}t} 
	    \mathrm{sinc} \left(\frac{\omega^\prime-\omega}{2}t\right) \sum_k J^{\frac 12}_{jk}(\omega^\prime) J^{\frac 12}_{ki}(\omega)\\
	    &\stackrel{t \to + \infty}{\approx}2 \pi e^{i \frac{\omega^\prime-\omega}{2}t}  \int_{0}^{+\infty} d\Omega~  N\left(T(\Omega),\Omega\right) \left(
	    \delta(\omega^\prime-\Omega)   \mathrm{sinc} \left(\frac{\omega-\Omega}{2}t\right) J_{ji}(\Omega) \right. \nonumber \\
	   &\left. +
	    \delta(\omega^\prime+\Omega)  \mathrm{sinc} \left(\frac{\omega+\Omega}{2}t\right) J_{ij}(\Omega)
	    \right) +2\pi  H(\omega)H(\omega^\prime)  e^{i \frac{\omega^\prime-\omega}{2}t} 
	    \mathrm{sinc} \left(\frac{\omega^\prime-\omega}{2}t\right) \sum_k J^{\frac 12}_{jk}(\omega^\prime) J^{\frac 12}_{ki}(\omega)\\
	    &= 2 \pi e^{i \frac{\omega^\prime-\omega}{2}t}    \left( N\left(T(\omega^\prime),\omega^\prime\right)
	    H(\omega^\prime)   \mathrm{sinc} \left(\frac{\omega-\omega^\prime}{2}t\right) J_{ji}(\omega^\prime)  + N\left(T(-\omega^\prime),-\omega^\prime\right)
	    H(-\omega^\prime)  \mathrm{sinc} \left(\frac{\omega-\omega^\prime}{2}t\right) J_{ij}(-\omega^\prime)
	    \right) \nonumber \\
	    & +2\pi  H(\omega)H(\omega^\prime)  e^{i \frac{\omega^\prime-\omega}{2}t} 
	    \mathrm{sinc} \left(\frac{\omega^\prime-\omega}{2}t\right) \sum_k J^{\frac 12}_{jk}(\omega^\prime) J^{\frac 12}_{ki}(\omega)\\
	    &\stackrel{t \to + \infty}{\approx} 2\pi\delta_{\omega,\omega^\prime} \left( N\left(T(\omega^\prime),\omega^\prime\right)
	    H(\omega^\prime)    J_{ji}(\omega^\prime)+  H(\omega^\prime) 
	     J_{ji}(\omega^\prime) + N\left(T(-\omega^\prime),-\omega^\prime\right)
	    H(-\omega^\prime)   J_{ij}(-\omega^\prime)
	    \right)\\
	    &= 2\pi\delta_{\omega,\omega^\prime} \left( \left(N\left(T(\omega^\prime),\omega^\prime\right)+1\right)
	    H(\omega^\prime)    J_{ji}(\omega^\prime)+ N\left(T(-\omega^\prime),-\omega^\prime\right)
	    H(-\omega^\prime)   J_{ij}(-\omega^\prime)
	    \right)\\
	    &= 2\pi\delta_{\omega,\omega^\prime} R_{ji}(\omega^\prime),
\end{align}
where we employed both models of delta function from equation (\ref{eqn:models}), and properties of $\mathrm{sinc}$ function. The last equality is due to the Definition \ref{def:decentR} of the decent reservoir. Finally, we compute:
\begin{align}
	\tilde{K}^{(2,\star \star)}(t) \tilde{\rho}_S &=
	\sum_{i,j} \sum_{\omega, \omega^\prime} \gamma^{\star\star}_{ij}(\omega,\omega^\prime,t)  \left(A_i (\omega) \tilde{\rho}_S A_j^\dagger (\omega^\prime) - \frac{1}{2} \left\{A_j^\dagger (\omega^\prime) A_i (\omega), \tilde{\rho}_S \right\} \right)\\
	&\stackrel{t \to +\infty}{\approx}\sum_{i,j} \sum_{\omega, \omega^\prime} 2\pi t \delta_{\omega,\omega^\prime}  R_{ji} (\omega)  \left(A_i (\omega) \tilde{\rho}_S A_j^\dagger (\omega^\prime) - \frac{1}{2} \left\{A_j^\dagger (\omega^\prime) A_i (\omega), \tilde{\rho}_S \right\} \right)\\
	&= t\sum_{i,j} \sum_{\omega} 2\pi R_{ji} (\omega)  \left(A_i (\omega) \tilde{\rho}_S A_j^\dagger (\omega) - \frac{1}{2} \left\{A_j^\dagger (\omega) A_i (\omega), \tilde{\rho}_S \right\} \right)\\
	&= t\sum_{i,j} \sum_{\omega} \gamma_{ij} (\omega)  \left(A_i (\omega) \tilde{\rho}_S A_j^\dagger (\omega) - \frac{1}{2} \left\{A_j^\dagger (\omega) A_i (\omega), \tilde{\rho}_S \right\} \right)
	=t \tilde{\mathcal{L}}^\mathrm{fs} \tilde{\rho}_S.
\end{align}

\begin{remark}
The above long-time limits holds solely in the sense approximation. In fact, the long-time limit of the cumulant superoperator (and and its approximations) contains (bounded) non-diagonal elements~\cite{RenormalizationPaper}. Fortunately, in the effective sense of dynamics the above relations hold, as all types of dynamical maps converge to the same final state~\cite{MW_preparation}.  
\end{remark}

\subsection{Reproducing the quasi-secular master equation}

If the systems consist of more than two energy levels, we can arrange these levels into groups of level, what also induced a grouping of Bohr's frequencies. The grouping procedure is always arbitrary, but if we can arrange them into groups well-separated from each other, with respect to the energy separations within the group (see Fig. \ref{fig:grouping}), the quasi-secular master equation can be applied. Furthermore, to apply this approach, the spectral density can not vary too much within any group. 
\begin{figure}[h!]
    \centering
    \includegraphics[width=0.45 \columnwidth]{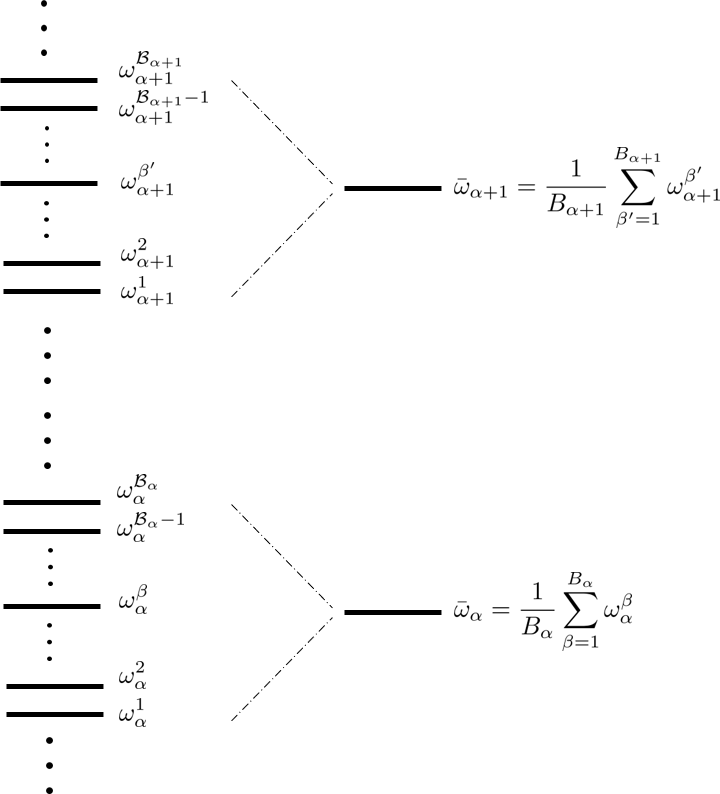}
    \caption{Diagrammatic representation of non-secular truncation.    }
    \label{fig:grouping}
\end{figure} 
We can denote now the energies of levels with $\omega_\alpha^\beta$, where $\alpha$ corresponds to the index of a group, and $\beta$ indicates the level within the group (for the diagrammatic representation see Fig. \ref{fig:grouping}). When the grouping is already done, we can introduce the following quantities.
\begin{align}
    \Delta \Omega &= \min_{\alpha,\nu} \abs{\bar{\omega}_\alpha-\bar{\omega}_{\nu}},\\
    \Delta \omega &= \max_{\alpha,\beta,\mu} \abs{\omega_\alpha^\beta-\omega_\alpha^{\mu}},
\end{align}
that define the ``short times scale'' , i.e,  $\frac{1}{\Delta\Omega} \ll t \ll \frac{1}{\Delta \omega}$. Indeed, we say that the system has ``well'' separated energy levels iff $\Delta \omega \ll \Delta \Omega$. Furthermore, it is important to notice that then there is no Bohr's frequency $\omega_{\emptyset}$, such that $\Delta \omega \le \omega_{\emptyset} \le \Delta \Omega$ (as a consequence of the well-separated levels assumption).

In order to derive the quasi-secular master equation we employ ``quasi-secular jump operators'' .
\begin{align}\label{def:jumpLocal}
	&{A}^\mathrm{qs}_i(\bar{\omega}_\alpha)=\sum_{\beta=1}^{B_\alpha}\sum_{\epsilon^\prime - \epsilon= \omega_\alpha^\beta} \Pi(\epsilon) {A}_i \Pi(\epsilon^\prime), 
\end{align}
for which it evident that $A(\omega) \neq A^\mathrm{qs} (\omega)$, unless group has only one member (see equation (\ref{eqn:jumpStd}) for comparison). 

Convergence to the quasi-secular master equation is a little more involved than in the case of the fully-secular master equation. We show the convergence property for FA equation and $\star\star$-approximation of the cumulant equation. In the case of the cumulant equation the proof of convergence to quasi-secular master equation is (mutatis mutandis) the same as in the case of $\star\star$-approximation. A clustering procedure of frequencies was also used in~\cite{trushechkin2021unified} to derive a universal GKSL equation.
We start by transforming the time-dependent generator $\tilde{K}^{(2)}(t)$ into the Schr{\"o}dinger picture. This must be done with the aid of the Baker–Campbell–Hausdorff formula, since $\tilde{K}^{(2)}(t)$ does not commute with the physical Hamiltonian $H_\mathcal{S}$~\cite{RenormalizationPaper}.
\begin{align}
	{K}^{(2)}(t) \rho_S &= -it \left[H_\mathcal{S},\rho_S\right]+
	\sum_{i,j} \sum_{\omega, \omega^\prime} \gamma_{ij}(\omega,\omega^\prime,t)   \left(A_i (\omega) \rho_S A_j^\dagger (\omega^\prime) - \frac{1}{2} \left\{A_j^\dagger (\omega^\prime) A_i (\omega), \rho_S \right\} \right) \nonumber\\
	&+ \frac 12
	\sum_{i,j} \sum_{\omega, \omega^\prime} [it(\omega-\omega^\prime)]\gamma_{ij}(\omega,\omega^\prime,t)   \left(A_i (\omega) \rho_S A_j^\dagger (\omega^\prime) - \frac{1}{2} \left\{A_j^\dagger (\omega^\prime) A_i (\omega), \rho_S \right\} \right)+ \cdots  \label{eqn:ActionOfCumulantSch}\\
	&=-it \left[H_\mathcal{S},\rho_S\right]+
	\sum_{i,j} \sum_{\omega, \omega^\prime} \left(1+\frac 12 [it(\omega-\omega^\prime)]\right)\gamma_{ij}(\omega,\omega^\prime,t)   \left(A_i (\omega) \rho_S A_j^\dagger (\omega^\prime) - \frac{1}{2} \left\{A_j^\dagger (\omega^\prime) A_i (\omega), \rho_S \right\} \right)+\cdots.\label{app:eqn:BCH1}
\end{align}
Where the third term in equation (\ref{eqn:ActionOfCumulantSch}) is equal to $-\frac{it}{2} \left[  \left[H_\mathcal{S},\cdot \right],\tilde{K}^{(2)}(t)\right]\rho_S$, and the form in equation (\ref{app:eqn:BCH1}) is obtained with ``higher'' order commutators. The first order not shown explicitly in the equation above contains the terms proportional to $\mathcal{O}((\omega-\omega^\prime)^2)$. 

\begin{remark}
The reason to perform the derivation of a quasi-secular Markovian equation in the Sch{\"o}dinger picture is the lack of ambiguity in choosing the interaction picture and also the no need for adding an ad hoc splitting Hamiltonian. 
\end{remark}

We can now write the equation (\ref{eqn:ActionOfCumulantSch}) in ``groups'' notation that is incorporated in this Section (see Fig. \ref{fig:grouping}):
\begin{align}
	&{K}^{(2)}(t) \rho_S = -it \left[H_\mathcal{S},\rho_S\right] \nonumber\\
	&+
	\sum_{i,j} \sum_{\alpha,\nu} \sum_{\beta=1}^{B_\alpha} \sum_{\mu=1}^{M_\nu} \left(1+[it(\omega_\alpha^\beta-\omega_\nu^{\mu})]\right)\gamma_{ij}(\omega_\alpha^\beta,\omega_\nu^{\mu},t)   \left(A_i (\omega_\alpha^\beta) \rho_S A_j^\dagger (\omega_\nu^{\mu}) - \frac{1}{2} \left\{A_j^\dagger (\omega_\nu^{\mu}) A_i (\omega_\alpha^\beta), \rho_S \right\} \right)+\cdots .
\end{align}

The first step makes use of the assumption that within any group the spectral density does not vary too much. Then, the next step is to perform the ``the secular approximation'' , that eliminates fast rotating terms.
\begin{align}
    &{K}^{(2)}(t) \rho_S = -it \left[H_\mathcal{S},\rho_S\right] \nonumber\\
    &+
	\sum_{i,j} \sum_{\alpha,\nu} \sum_{\beta=1}^{B_\alpha} \sum_{\mu=1}^{M_\nu} \left(1+[it(\omega_\alpha^\beta-\omega_\nu^{\mu})]\right)\gamma_{ij}(\omega_\alpha^\beta,\omega_\nu^{\mu},t)   \left(A_i (\omega_\alpha^\beta) \rho_S A_j^\dagger (\omega_\nu^{\mu}) - \frac{1}{2} \left\{A_j^\dagger (\omega_\nu^{\mu}) A_i (\omega_\alpha^\beta), \rho_S \right\} \right)+\cdots \\
	&\stackrel{(I)}{\approx}  -it \left[H_\mathcal{S},\rho_S\right]+
	\sum_{i,j} \sum_{\alpha,\nu} \sum_{\beta=1}^{B_\alpha} \sum_{\mu=1}^{M_\nu} \left(1+[it(\bar{\omega}_\alpha-\bar{\omega}_\nu)]\right) \gamma_{ij}(\bar{\omega}_\alpha,\bar{\omega}_\nu,t)   \left(A_i (\omega_\alpha^\beta) \rho_S A_j^\dagger (\omega_\nu^{\mu}) - \frac{1}{2} \left\{A_j^\dagger (\omega_\nu^{\mu}) A_i (\omega_\alpha^\beta), \rho_S \right\} \right)+\cdots \\
	&= -it \left[H_\mathcal{S},\rho_S\right] \nonumber \\
	&+
	\sum_{i,j} \sum_{\alpha,\nu}  \left(1+[it(\bar{\omega}_\alpha-\bar{\omega}_\nu)]\right) \gamma_{ij}(\bar{\omega}_\alpha,\bar{\omega}_\nu,t)   \left(\sum_{\beta=1}^{B_\alpha} A_i (\omega_\alpha^\beta) \rho_S \sum_{\mu=1}^{M_\nu} A_j^\dagger (\omega_\nu^{\mu}) - \frac{1}{2} \left\{ \sum_{\mu=1}^{M_\nu} A_j^\dagger (\omega_\nu^{\mu}) \sum_{\beta=1}^{B_\alpha} A_i (\omega_\alpha^\beta), \rho_S \right\} \right)+\cdots\\
	&\stackrel{(II)}{=} -it \left[H_\mathcal{S},\rho_S\right]+
	\sum_{i,j} \sum_{\alpha,\nu}   \left(1+[it(\bar{\omega}_\alpha-\bar{\omega}_\nu]\right)\gamma_{ij}(\bar{\omega}_\alpha,\bar{\omega}_\nu,t)   \left( A^\mathrm{qs}_i (\bar{\omega}_\alpha) \rho_S {A^\mathrm{qs}_j}^\dagger (\bar{\omega}_\nu) - \frac{1}{2} \left\{  {A^\mathrm{qs}_j}^\dagger (\bar{\omega}_\nu)  A^\mathrm{qs}_i (\bar{\omega}_\alpha), \rho_S \right\} \right)+\cdots\\
	&\stackrel{(III)}{\approx} -it \left[H_\mathcal{S},\rho_S\right]+
	\sum_{i,j} \sum_{\alpha}   \gamma_{ij}(\bar{\omega}_\alpha,\bar{\omega}_\alpha,t)   \left( A^\mathrm{qs}_i (\bar{\omega}_\alpha) \rho_S {A^\mathrm{qs}_j}^\dagger (\bar{\omega}_\alpha) - \frac{1}{2} \left\{  {A^\mathrm{qs}_j}^\dagger (\bar{\omega}_\alpha)  A^\mathrm{qs}_i (\bar{\omega}_\alpha), \rho_S \right\} \right)= (\clubsuit).
\end{align}
Here, in step $(I)$ we made the ``grouping''  approximation, and in $(II)$, we used the Definition in equation (\ref{def:jumpLocal}). In step $(III)$, we performed approximation, usually referred to as ``the secular approximation'' however, in a range of non-resonant (non-Bohr's) frequencies. Steps $(I)$ and $(III)$ combined account for ``the non-secular approximation''. Before we obtain the ``local'' generator, we still have to perform the Markovian approximation.

\begin{align}
    \gamma_{ij}(\bar{\omega}_\alpha,\bar{\omega}_\alpha,t) &= \int_0^t ds \int_0^t dw~ e^{i \bar{\omega}_\alpha( s -  w)} \left< \tilde{B}_j (s-w) \tilde{B}_i \right>_{\tilde{\rho}_B}
    =\int_{- \frac t2}^{\frac t2} ds \int_{- \frac t2}^{\frac t2} dw~ e^{i \bar{\omega}_\alpha( s -  w)} \left< \tilde{B}_j (s-w) \tilde{B}_i \right>_{\tilde{\rho}_B}\\
    &\stackrel{(I)}{=} \frac 12 \int_{-t}^{t} dv \int_{-t}^{t} du~ e^{i \bar{\omega}_\alpha u} \left< \tilde{B}_j (u) \tilde{B}_i \right>_{\tilde{\rho}_B}
    = t \int_{-t}^{t} du~ e^{i \bar{\omega}_\alpha u} \left< \tilde{B}_j (u) \tilde{B}_i \right>_{\tilde{\rho}_B}\\
    &\stackrel{(II)}{\approx} t \int_{-\infty}^{+\infty} du~ e^{i \bar{\omega}_\alpha u} \left< \tilde{B}_j (u) \tilde{B}_i \right>_{\tilde{\rho}_B}
    =2\pi t \left<\left(\frac{1}{2\pi}\int_{-\infty}^{+\infty} du~ e^{i \bar{\omega}_\alpha u}  \tilde{B}_j (u)\right) \tilde{B}_i \right>_{\tilde{\rho}_B}\\
    &\stackrel{(III)}{=}2\pi t \left<B_j(\bar{\omega}_\alpha ) \tilde{B}_i \right>_{\tilde{\rho}_B}= 2\pi t R_{ji} (\bar{\omega}_\alpha ) = t \gamma_{ij}(\bar{\omega}_\alpha).
\end{align}
Where in $(I)$ we performed change of variables, i.e., $u=s-w,~~v=s+v$, with $\frac{D(s,w)}{D(v,u)}=\frac 12$. Then the step $(II)$ accounts for the Markovian approximation. Finally in $(III)$ we used the definition of Fourier transform of bath operators (see equation (\ref{def:FTofB})).
\begin{align}\label{eqn:GenLoc}
    (\clubsuit) &\approx -it \left[H_\mathcal{S},\rho_S\right]+
	t \sum_{i,j} \sum_{\alpha}    \gamma_{ij}(\bar{\omega}_\alpha)   \left( A^\mathrm{qs}_i (\bar{\omega}_\alpha) \rho_S {A^\mathrm{qs}_j}^\dagger (\bar{\omega}_\alpha) - \frac{1}{2} \left\{  {A^\mathrm{qs}_j}^\dagger (\bar{\omega}_\alpha)  A^\mathrm{qs}_i (\bar{\omega}_\alpha), \rho_S \right\} \right)
	= t \mathcal{L}^\mathrm{qs} \rho_S
\end{align}

Let us firstly show how FA equation converges to the quasi-secular master equation. The first step, to achieve this is to split the sum in equation (\ref{eqn:GenLoc}) into three parts.
\begin{align}
    {K}^{(2,\star)}(t) \rho_S &= -it \left[H_\mathcal{S},\rho_S\right]+
	\sum_{i,j} \sum_{\omega_\alpha^\beta,\omega_\nu^{\mu}\in \mathrm{Bohr}} \gamma^{\star}_{ij}(\omega_\alpha^\beta,\omega_\nu^{\mu},t)   \left(A_i (\omega_\alpha^\beta) \rho_S A_j^\dagger (\omega_\nu^{\mu}) - \frac{1}{2} \left\{A_j^\dagger (\omega_\nu^{\mu}) A_i (\omega_\alpha^\beta), \rho_S \right\} \right)\\
	=  -it \left[H_\mathcal{S},\rho_S\right] &+
	\sum_{i,j} \sum_{\abs{\omega_\alpha^\beta-\omega_\nu^{\mu}} \le \Delta \omega } \gamma^{\star}_{ij}(\omega_\alpha^\beta,\omega_\nu^{\mu},t)   \left(A_i (\omega_\alpha^\beta) \rho_S A_j^\dagger (\omega_\nu^{\mu}) - \frac{1}{2} \left\{A_j^\dagger (\omega_\nu^{\mu}) A_i (\omega_\alpha^\beta), \rho_S \right\} \right) \nonumber \\
	&+	\sum_{i,j} \sum_{\Delta \omega < \abs{\omega_\alpha^\beta-\omega_\nu^{\mu}} < \Delta \Omega } \gamma^{\star}_{ij}(\omega_\alpha^\beta,\omega_\nu^{\mu},t)   \left(A_i (\omega_\alpha^\beta) \rho_S A_j^\dagger (\omega_\nu^{\mu}) - \frac{1}{2} \left\{A_j^\dagger (\omega_\nu^{\mu}) A_i (\omega_\alpha^\beta), \rho_S \right\} \right) \nonumber \\
	&+	\sum_{i,j} \sum_{\Delta \Omega  \le \abs{\omega_\alpha^\beta-\omega_\nu^{\mu}}  } \gamma^{\star}_{ij}(\omega_\alpha^\beta,\omega_\nu^{\mu},t)   \left(A_i (\omega_\alpha^\beta) \rho_S A_j^\dagger (\omega_\nu^{\mu}) - \frac{1}{2} \left\{A_j^\dagger (\omega_\nu^{\mu}) A_i (\omega_\alpha^\beta), \rho_S \right\} \right)=(\spadesuit).
\end{align}
Because the second summation runs over an empty set, as was mentioned before, it is equal to $0$ exactly.

When considering the first summation for $\abs{\omega_\alpha^\beta-\omega_\nu^{\mu}} < \Delta \omega$ and $\Delta \omega t \ll  1$, we can approximate the phase factor and the $\mathrm{sinc}$ function with $1$.
\begin{align}\label{app:eqn:CloseBohrSM}
    \frac{\gamma^{\star}_{ij}(\omega_\alpha^\beta,\omega_\nu^{\mu},t)}{t} &= 2\pi  e^{i \frac{\omega_\nu^{\mu}-\omega_\alpha^\beta}{2}t}
	     \mathrm{sinc} \left(\frac{\omega_\nu^{\mu}-\omega_\alpha^\beta}{2}t\right)  \sum_k R^{\frac 12}_{jk} (\omega_\nu^{\mu}) R^{\frac 12}_{ki} (\omega_\alpha^\beta) 
	    \approx 2\pi  R_{ji} (\omega_\nu^{\mu}).
\end{align}\label{app:eqn:FarBohrSM}
For the third summation, and $\Delta \Omega<\abs{\omega_\alpha^\beta-\omega_\nu^{\mu}}  $ and $1 \ll \Delta \Omega t $ we have:
\begin{align}\label{app:eqn:FarBohrSM}
    \frac{\gamma^{\star}_{ij}(\omega_\alpha^\beta,\omega_\nu^{\mu},t)}{t}  &= 2\pi  e^{i \frac{\omega_\nu^{\mu}-\omega_\alpha^\beta}{2}t}
	     \mathrm{sinc} \left(\frac{\omega_\nu^{\mu}-\omega_\alpha^\beta}{2}t\right)  \sum_k R^{\frac 12}_{jk} (\omega_\nu^{\mu}) R^{\frac 12}_{ki} (\omega_\alpha^\beta) 
	    \approx 0.
\end{align}
Where the last approximation is with comparison to the r.h.s. of equation (\ref{app:eqn:CloseBohrSM}). Therefore:
\begin{align}
    (\spadesuit)&\approx -it \left[H_\mathcal{S},\rho_S\right] +
	\sum_{i,j} \sum_{\abs{\omega_\alpha^\beta-\omega_\nu^{\mu}} < \Delta \omega } \gamma^{\star}_{ij}(\omega_\alpha^\beta,\omega_\nu^{\mu},t)   \left(A_i (\omega_\alpha^\beta) \rho_S A_j^\dagger (\omega_\nu^{\mu}) - \frac{1}{2} \left\{A_j^\dagger (\omega_\nu^{\mu}) A_i (\omega_\alpha^\beta), \rho_S \right\} \right)\\
	&\approx -it \left[H_\mathcal{S},\rho_S\right] +
	\sum_{i,j} \sum_{\abs{\omega_\alpha^\beta-\omega_\nu^{\mu}} < \Delta \omega } 2\pi t R_{ji} (\omega_\nu^{\mu})  \left(A_i (\omega_\alpha^\beta) \rho_S A_j^\dagger (\omega_\nu^{\mu}) - \frac{1}{2} \left\{A_j^\dagger (\omega_\nu^{\mu}) A_i (\omega_\alpha^\beta), \rho_S \right\} \right)\\
	&= -it \left[H_\mathcal{S},\rho_S\right] +
	\sum_{i,j} \sum_{\alpha,\nu,\beta\mu:~\abs{\omega_\alpha^\beta-\omega_\nu^{\mu}} < \Delta \omega}  2\pi t R_{ji} (\omega_\nu^{\mu})  \left(A_i (\omega_\alpha^\beta) \rho_S A_j^\dagger (\omega_\nu^{\mu}) - \frac{1}{2} \left\{A_j^\dagger (\omega_\nu^{\mu}) A_i (\omega_\alpha^\beta), \rho_S \right\} \right)\\
	&\stackrel{(I)}{=} -it \left[H_\mathcal{S},\rho_S\right] +
	\sum_{i,j} \sum_{\alpha,\beta\mu}  2\pi t R_{ji} (\omega_\alpha^{\mu})  \left(A_i (\omega_\alpha^\beta) \rho_S A_j^\dagger (\omega_\alpha^{\mu}) - \frac{1}{2} \left\{A_j^\dagger (\omega_\alpha^{\mu}) A_i (\omega_\alpha^\beta), \rho_S \right\} \right)\\
	&\stackrel{(II)}{\approx} -it \left[H_\mathcal{S},\rho_S\right] +
	\sum_{i,j} \sum_{\alpha,\beta\mu}  2\pi t R_{ji} (\bar{\omega}_\alpha)  \left(A_i (\omega_\alpha^\beta) \rho_S A_j^\dagger (\omega_\alpha^{\mu}) - \frac{1}{2} \left\{A_j^\dagger (\omega_\alpha^{\mu}) A_i (\omega_\alpha^\beta), \rho_S \right\} \right)\\
	&= -it \left[H_\mathcal{S},\rho_S\right] +
	\sum_{i,j} \sum_{\alpha}  2\pi t R_{ji} (\bar{\omega}_\alpha)  \left(\sum_\beta A_i (\omega_\alpha^\beta) \rho_S \sum_\mu A_j^\dagger (\omega_\alpha^{\mu}) - \frac{1}{2} \left\{\sum_\mu A_j^\dagger (\omega_\alpha^{\mu}) \sum_\beta A_i (\omega_\alpha^\beta), \rho_S \right\} \right)\\
	&= -it \left[H_\mathcal{S},\rho_S\right] +
	t \sum_{i,j} \sum_{\alpha}    \gamma_{ij} (\bar{\omega}_\alpha)  \left( A^\mathrm{qs}_i (\bar{\omega}_\alpha) \rho_S  {A^\mathrm{qs}_j}^\dagger (\bar{\omega}_\alpha) - \frac{1}{2} \left\{ {A^\mathrm{qs}_j}^\dagger (\bar{\omega}_\alpha)  A^\mathrm{qs}_i (\bar{\omega}_\alpha), \rho_S \right\} \right)=t \mathcal{L}^\mathrm{qs} \rho_S, 
\end{align}
where again step $(I)$ is because the ``grouping'' approximation, i.e, small frequency differences occur only within a group. The second step $(II)$ is due to an assumption that the spectral density does not vary too much within the group. Finally, upon transforming to the interaction picture, we obtain:
\begin{align}
    e^{\tilde{K}^{(2,\star)}(t)} \stackrel{\frac{1}{\Delta\Omega} \ll t \ll \frac{1}{\Delta \omega}}{\approx} e^{t \tilde{\mathcal{L}}^\mathrm{qs}}.
\end{align}
In this way, we have shown that FA equation approximates the quasi-secular master equation within short-times scale, i.e., $\frac{1}{\Delta\Omega} \ll t \ll \frac{1}{\Delta \omega}$.

The proof for the convergence of $\star\star$-type approximation of the cumulant equation to the quasi-secular master equation is similar to the $\star$ case. The only difference appears, at the step of equations (\ref{app:eqn:CloseBohrSM},\ref{app:eqn:FarBohrSM}).

For $\abs{\omega_\alpha^\beta-\omega_\nu^{\mu}} < \Delta \omega$ and $\Delta \omega t \ll  1$, we have:
\begin{align}\label{app:eqn:CloseBohrSM2}
    &\frac{\gamma^{\star\star}_{ij}(\omega_\alpha^\beta,\omega_\nu^{\mu},t)}{t} =e^{i \frac{\omega_\nu^{\mu}-\omega_\alpha^\beta}{2}t} \frac{1}{t} \int_{0}^{\infty} d\Omega~  N\left(T(\Omega),\Omega\right) \left(
	    \left[t~ \mathrm{sinc} \left(\frac{\omega_\nu^{\mu}-\Omega}{2}t\right)\right]  \left[t~ \mathrm{sinc} \left(\frac{\omega_\alpha^\beta-\Omega}{2}t\right)\right] J_{ji}(\Omega) \right. \nonumber \\
	   &\left. +
	    \left[t~ \mathrm{sinc} \left(\frac{\omega_\nu^{\mu}+\Omega}{2}t\right)\right]  \left[t~ \mathrm{sinc} \left(\frac{\omega_\alpha^\beta+\Omega}{2}t\right)\right] J_{ij}(\Omega)
	    \right) \\
	    &+2\pi  H(\omega_\alpha^\beta)H(\omega_\nu^{\mu})  e^{i \frac{\omega_\nu^{\mu}-\omega_\alpha^\beta}{2}t} 
	    \mathrm{sinc} \left(\frac{\omega_\nu^{\mu}-\omega_\alpha^\beta}{2}t\right) \sum_k J^{\frac 12}_{jk}(\omega_\nu^{\mu}) J^{\frac 12}_{ki}(\omega_\alpha^\beta) \\
	    &\approx  \frac{1}{t} \int_{0}^{\infty} d\Omega~  N\left(T(\Omega),\Omega\right) \left(
	    \left[t~ \mathrm{sinc} \left(\frac{\omega_\nu^{\mu}-\Omega}{2}t\right)\right]^2  J_{ji}(\Omega)
	    +\left[t~ \mathrm{sinc} \left(\frac{\omega_\nu^{\mu}+\Omega}{2}t\right)\right]^2   J_{ij}(\Omega)
	    \right) \\
	    &+2\pi  H(\omega_\alpha^\beta)H(\omega_\nu^{\mu})    J_{ji}(\omega_\nu^{\mu}) \\
	    &= 2\pi  \int_{0}^{+\infty} d\Omega~  N\left(T(\Omega),\Omega\right) \left(
	    \delta^{(2)}_{\frac t2}(\omega_\nu^{\mu}-\Omega)   J_{ji}(\Omega)
	    +\delta^{(2)}_{\frac t2}(\omega_\nu^{\mu}+\Omega)    J_{ij}(\Omega)
	    \right) \\
	    &+2\pi  H(\omega_\alpha^\beta)H(\omega_\nu^{\mu})    J_{ji}(\omega_\nu^{\mu}) \\
	    &\approx 2\pi  \int_{0}^{+\infty} d\Omega~  N\left(T(\Omega),\Omega\right) \left(
	    \delta(\omega_\nu^{\mu}-\Omega)   J_{ji}(\Omega)
	    +\delta(\omega_\nu^{\mu}+\Omega)    J_{ij}(\Omega)
	    \right) \\
	    &+2\pi  H(\omega_\nu^{\mu})    J_{ji}(\omega_\nu^{\mu}) \\
	    &=2\pi    \left(N\left(T(\omega_\nu^{\mu}),\omega_\nu^{\mu}\right)
	    H(\omega_\nu^{\mu})   J_{ji}(\omega_\nu^{\mu})
	    +N\left(T(\omega_\nu^{\mu}),\omega_\nu^{\mu}\right) H(-\omega_\nu^{\mu})    J_{ij}(\Omega)
	    \right) +2\pi  H(\omega_\nu^{\mu})    J_{ji}(\omega_\nu^{\mu}) \\
	    &=2\pi    \left( \left(N\left(T(\omega_\nu^{\mu}),\omega_\nu^{\mu}\right)+1\right)
	    H(\omega_\nu^{\mu})   J_{ji}(\omega_\nu^{\mu})
	    +N\left(T(\omega_\nu^{\mu}),\omega_\nu^{\mu}\right) H(-\omega_\nu^{\mu})    J_{ij}(\Omega)
	    \right) \\
	    & =2 \pi  R_{ji} (\omega_\nu^{\mu}),
	  \end{align}
Here, we notice that the quality of approximation is determined how well models approximate exact Dirac delta function, given particular spectral density.  Subsequently for $\Delta \Omega<\abs{\omega_\alpha^\beta-\omega_\nu^{\mu}}  $ and $1 \ll \Delta \Omega t $ ( (what implies $\abs{\omega_\alpha^\beta-\omega_\nu^{\mu}}t \gg 1$))
\begin{align}\label{app:eqn:FarBohrSM2}
    &\frac{\gamma^{\star\star}_{ij}(\omega_\alpha^\beta,\omega_\nu^{\mu},t)}{t} =e^{i \frac{\omega_\nu^{\mu}-\omega_\alpha^\beta}{2}t} \frac{1}{t} \int_{0}^{\infty} d\Omega~  N\left(T(\Omega),\Omega\right) \left(
	    \left[t~ \mathrm{sinc} \left(\frac{\omega_\nu^{\mu}-\Omega}{2}t\right)\right]  \left[t~ \mathrm{sinc} \left(\frac{\omega_\alpha^\beta-\Omega}{2}t\right)\right] J_{ji}(\Omega) \right. \nonumber \nonumber \\
	   &\left. +
	    \left[t~ \mathrm{sinc} \left(\frac{\omega_\nu^{\mu}+\Omega}{2}t\right)\right]  \left[t~ \mathrm{sinc} \left(\frac{\omega_\alpha^\beta+\Omega}{2}t\right)\right] J_{ij}(\Omega)
	    \right) \nonumber \\ 
	    &+2\pi  H(\omega_\alpha^\beta)H(\omega_\nu^{\mu})  e^{i \frac{\omega_\nu^{\mu}-\omega_\alpha^\beta}{2}t} 
	    \mathrm{sinc} \left(\frac{\omega_\nu^{\mu}-\omega_\alpha^\beta}{2}t\right) \sum_k J^{\frac 12}_{jk}(\omega_\nu^{\mu}) J^{\frac 12}_{ki}(\omega_\alpha^\beta) \\
     &= 2\pi \int_{0}^{+\infty} d\Omega~  N\left(T(\Omega),\Omega\right) \left(
	    \delta^{(1)}_{\frac t2}(\omega_\nu^{\mu}-\Omega)  \mathrm{sinc} \left(\frac{\omega_\alpha^\beta-\Omega}{2}\right) J_{ji}(\Omega)  +
	    \delta^{(1)}_{\frac t2}(\omega_\nu^{\mu}+\Omega)   \mathrm{sinc} \left(\frac{\omega_\alpha^\beta+\Omega}{2}\right) J_{ij}(\Omega)
	    \right) \nonumber \\ 
	    &+2\pi  H(\omega_\alpha^\beta)H(\omega_\nu^{\mu})  e^{i \frac{\omega_\nu^{\mu}-\omega_\alpha^\beta}{2}t} 
	    \mathrm{sinc} \left(\frac{\omega_\nu^{\mu}-\omega_\alpha^\beta}{2}t\right) \sum_k J^{\frac 12}_{jk}(\omega_\nu^{\mu}) J^{\frac 12}_{ki}(\omega_\alpha^\beta) \\
	   &\approx 2\pi \int_{0}^{+\infty} d\Omega~  N\left(T(\Omega),\Omega\right) \left(
	    \delta(\omega_\nu^{\mu}-\Omega)  \mathrm{sinc} \left(\frac{\omega_\alpha^\beta-\Omega}{2}t\right) J_{ji}(\Omega)  +
	    \delta(\omega_\nu^{\mu}+\Omega)   \mathrm{sinc} \left(\frac{\omega_\alpha^\beta+\Omega}{2}t\right) J_{ij}(\Omega) 
	     \right)\nonumber \\ 
	    &+2\pi  H(\omega_\alpha^\beta)H(\omega_\nu^{\mu})  e^{i \frac{\omega_\nu^{\mu}-\omega_\alpha^\beta}{2}t} 
	    \mathrm{sinc} \left(\frac{\omega_\nu^{\mu}-\omega_\alpha^\beta}{2}t\right) \sum_k J^{\frac 12}_{jk}(\omega_\nu^{\mu}) J^{\frac 12}_{ki}(\omega_\alpha^\beta) \\
	   &= 2\pi  R_{ji} (\omega_\nu^{\mu})  \mathrm{sinc} \left(\frac{\omega_\alpha^\beta-\omega_\nu^{\mu}}{2}t\right)  +2\pi  H(\omega_\alpha^\beta)H(\omega_\nu^{\mu})  e^{i \frac{\omega_\nu^{\mu}-\omega_\alpha^\beta}{2}t} 
	    \mathrm{sinc} \left(\frac{\omega_\nu^{\mu}-\omega_\alpha^\beta}{2}t\right) \sum_k J^{\frac 12}_{jk}(\omega_\nu^{\mu}) J^{\frac 12}_{ki}(\omega_\alpha^\beta) \approx 0
\end{align}
Where the last approximation is due to properties of $\mathrm{sinc}$ function (in the $\abs{\omega_\alpha^\beta-\omega_\nu^{\mu}}t \gg 1$ regime) and comparison to the r.h.s. of equation (\ref{app:eqn:CloseBohrSM2}). The other steps of the proof are the same as in the case of FA equation. Therefore, we have
\begin{align}
    e^{\tilde{K}^{(2,\star\star)}(t)} \stackrel{\frac{1}{\Delta\Omega} \ll t \ll \frac{1}{\Delta \omega}}{\approx} e^{t \tilde{\mathcal{L}}^\mathrm{qs}}.
\end{align}

Analogous proofs, both in quasi- and fully-secular  case, hold in the case of the (''original'') cumulant equation, that should a priori correctly describe all time scales.  We have     
\begin{align}
    &e^{\tilde{K}^{(2}(t)} \stackrel{t \to \infty}{\approx} e^{t \tilde{\mathcal{L}}^\mathrm{fs}},\\
    &e^{\tilde{K}^{(2}(t)} \stackrel{\frac{1}{\Delta\Omega} \ll t \ll \frac{1}{\Delta \omega}}{\approx} e^{t \tilde{\mathcal{L}}^\mathrm{qs}}.
\end{align}

\begin{remark}
The transformation to interaction picture (with respect to the physical Hamiltonian $H_\mathcal{S}$) for ${\mathcal{L}}^\mathrm{qs}$ is non-trivial. The difficulty comes from the fact that $\left[H_\mathcal{S},{\mathcal{L}}^\mathrm{qs}\right]\neq0$.
\end{remark}

\subsection{Comparison with exact numerics}\label{app:numerics}

We test the dynamics of the cumulant equation for the spin-boson model with the Hierarchical Equations of Motion (HEOM)~\cite{TanimuraHeom} and with the Bloch-Redfield equation. 
HEOM is a powerful non-perturbative approach and provided that the algorithm converges, the HEOM is thought to provide a numerically exact solution. The starting point of HEOM is a system interacting linearly with a bosonic environment, such that the Hamiltonian of the full system is given by:

\begin{equation}
    H=H_{S}(t)+\sum_{k} \omega_{k} a_{k}^{\dagger} a_{k} + Q \sum_{k} g_{k} (a_{k}^{\dagger}+ a_{k})
\end{equation}

the two main assumptions of the method are that at $t=0$ the system is in a product state, namely $\rho(t=0)=\rho_{S}(t=0)\otimes \rho_{B}$. and that the environment is an initially thermal equilibrium state 

\begin{equation}
    \rho_{B}=\frac{e^{-\beta \sum_{k} \omega_{k} a_{k}^{\dagger} a_{k} }}{Z}
\end{equation}

Furthermore the HEOM method typically uses a Drude-Lorentz spectral density as opposed to the exponential cutoff that we have been using in this paper, without lack of generality any spectral density can be written in terms of  a set of the Drude spectral densities in the under-damped regime by fitting to the one we want to simulate, or fitting the correlation functions directly as described in~\cite{QutipBonFin,Tannor}. In this paper we fit the correlation function to construct the bosonic bath. Without lack of generality correlation function of the environment for a bosonic gaussian bath can be written as:

\begin{equation}
    C(\tau) = \int_{0}^{\infty} d\omega \frac{J(\omega)}{\pi} \Big( \coth\left(\frac{\beta \omega}{2} \cos(\omega \tau) - i \sin(\omega \tau ) \right)\Big)
\end{equation}

for our choice of spectral density $J(\Omega) = \alpha \Omega \times e^{-\frac{\abs{\Omega}}{\omega_c}}$ it results in 

 \begin{align}\label{eq:analytical_cor}
 C(\tau) =& \: \frac{1}{\pi}\alpha  \beta^{- 2}  \Gamma(2) \left[ \zeta \left(2, \frac{1 + \beta \omega_c - i \omega_c \tau}{\beta \omega_c}\right) + \zeta \left(2, \frac{1 + i \omega_c \tau}{\beta \omega_c}\right) \right]
 \end{align}

where $\zeta$ is the Rienmann zeta function while $\Gamma$ is the gamma special function. When dealing with HEOM the correlation function is conveniently expressed as a sum of exponential functions, the reader interested in how this mapping can find a description in~\cite{QutipBonFin,Nori} and references therein :

\begin{align}\label{eq:series_corr}
      C(t)= C_{R}(t)+ C_{I}(t) \\
      C_{R}=  \sum_{k=1}^{N_{R}} c_{k}^{R} e^{-\gamma_{k}^{R} t} \\ 
      C_{I}= \sum_{k=1}^{N_{I}} c_{k}^{I} e^{-\gamma_{k}^{I} t}
\end{align}

After expressing the correlation function in this convenient form, the HEOM dynamics can be obtained by solving the following set of differential equations~\cite{TanimuraHeom,QutipBonFin}, once we have fitted Eq. \eqref{eq:series_corr} to \eqref{eq:analytical_cor}.

\begin{equation}\label{eq:HEOM_eqs}
    \dot{\rho}^{n}(t) =  -i [H_{S},\rho^{n}(t) ]  - \sum_{j=R,I}\sum_{k=1}^{N_{j}} n_{jk} \gamma_{k}^{j} \rho^{n}(t) - i \sum_{k=1}^{N_{R}} c_{k}^{R} n_{Rk} \{Q, \rho(t)^{n_{Rk}-1}\} + \sum_{k=1}^{N_{I}} c_{k}^{I} n_{Ik}\{Q, \rho(t)^{n_{Ik}-1}\}  - i  \sum_{j=R,I}\sum_{k=1}^{N_{j}} [Q,\rho^{n_{jk}+1}]
\end{equation}

where $n=(n_{R1},n_{R2},\hdots ,n_{R N_{R}},n_{I1},n_{I2},\hdots,n_{I N_{I}})$  is a multi-index label of integers for the different auxiliary density matrices (ADOs), $n_{jk} \epsilon \{0,N_{c}\}$, where $N_{c}$ is where we truncate our hierarchy. While $n=(0,0,\hdots,0)$ denotes the system density matrix. An important remark is that ADOs are not physical density operators in the standard sense, they correspond to terms collected for different exponents of the correlation functions that appear from the the application of the chain rule, during the derivation of Eq. \eqref{eq:HEOM_eqs} the interested reader may consult~\cite{TanimuraHeom,Tanimura1989} for the full derivation.

The comparison between the different approaches is done in the interaction picture. For the simulations, we have chosen  $\alpha=0.05$, $T=1$, $\omega=1$, and $\omega_{c}=5$, we consider seven  auxiliary density matrices (ADOs), as hinted in~\cite{Lobejko_Mean-Force} the dynamics of the spin boson model only  for the cumulant equation shows a shift on the diagonal with respect to Bloch-Redfield dynamics

We also include the standard Bloch-Redfield equation in our simulation, in the  Schrodinger picture it an be written as in Ref.~\cite{cattaneo2019local} 

\begin{align}
\frac{d \rho(t)}{dt}
=& - i \left[ H_{s}, \rho(t) \right] + D(\rho(t)) \\
D(\rho(t)) =& \sum_{\omega,\omega'}^{sec} \sum_{\alpha,\beta} \gamma_{\alpha,\beta}(\omega,\omega') \Big( A_{\beta}(\omega) \rho(t) A_{\alpha}^{\dagger}  -\frac{\{A_{\alpha}^{\dagger}(\omega') A_{\beta}(\omega), \rho(t)\}}{2} \Big),
\end{align}
where we have neglected the contribution from the Lamb-shift Hamiltonian term~\cite{RenormalizationPaper}. Whereas, the upper bound of the second summation \textit{sec} indicates summation over secular terms satisfying $|\omega-\omega'| \ll \tau_{decay}$. The $\gamma(\omega,\omega')$ coefficients are given by:

\begin{align}
    \gamma_{\alpha,\beta}(\omega,\omega') = \Gamma_{\alpha,\beta}(\omega) + \Gamma^{*}_{\beta,\alpha}(\omega')
\end{align}

where 

\begin{equation}
    \Gamma_{\alpha,\beta}(\omega)=\int_{0}^{\infty} dt' e^{i \omega t'} Tr\Big[ B_{\alpha}(t') B_{\beta}(0) \rho_{B}\Big]
\end{equation}


Both the Bloch-Redfield and HEOM equations are simulated using Qutip~\cite{qutip,QutipBonFin},
The dynamics derived via the cumulant equation ends up being similar to both the Bloch-Redfield equation and the HEOM, as it is shown in Fig.~\ref{fig:heom_plots} when describing coherences, and departs from the Bloch-Redfield description in terms of populations as hinted in~\cite{Lobejko_Mean-Force}.

From figure \ref{fig:heom_plots} (a) and (b) we can only see qualitatevely that the evolution from the Bloch-Redfield and the cumulant equation are closer to the the numerically exact solution than the standard GKLS equation. However we would like to see this quantitately and for that we need to introduce a distance between states. We will use quantum fidelity which is defined as 

\begin{align}
    \mathcal{F}(\rho,\sigma)= \left(Tr\left[\sqrt{\sqrt{\rho}\sigma \sqrt{\rho}}\right]\right)^{2}
\end{align}

Which is one when the states are the same, we compute the fidelity of the cumulant and Bloch-Redfield equation with respect to HEOM, the results are shown in \ref{fig:heom_plots} (c). The Markovian GKLS was left out because its fidelity is too low compared to its Non-Markovian counterparts, making their differences hard to visualize. On the plot we see that the secular version of Bloch-Redfield and Cumulant are comparable, however the cumulant does slightly better at really short times and long times.

The Bloch-Redfield equation compared with the cumulant agrees better with the HEOM when the secular approximation is not performed. Anyway, it is known that at certain times it can restitute a 
non-positive density matrix, as shown in~\cite{AlickiLendi1987,suarez1992memory} even for this simple case. 
To sum up, the cumulant equation is always completely-positive, ensuring physically meaningful solutions at every time, it shows the non-Markovian behavior of the dynamics 
in good agreement with the HEOM method, and eventually, it has the \textit{plug-and-play} easiness of the Davies equation.  

\newpage

\blk
\end{widetext}

\end{document}